\documentclass{amsart}

\usepackage{graphicx}
\usepackage{amsmath,amssymb,amsthm}
\usepackage{tikz}
\usetikzlibrary{arrows}

\usepackage{virginialake}

\usepackage{enumerate}
\usepackage{url}
\usepackage{subcaption}

\usepackage{hyperref} 
\usepackage[capitalize]{cleveref} 
\usepackage{thmtools}

\newcommand{\liff}{\leftrightarrow}
\newcommand{\seqar}{\Rightarrow}

\renewcommand{\epsilon}{\varepsilon}

\renewcommand{\emptyset}{\varnothing}

\makeatletter
\newcommand{\dotop}[1]{{\mathpalette\dotop@{#1}}}
\newcommand{\dotop@}[2]{%
  \vphantom{#2}%
  \ooalign{$\m@th#1\mathop#2$\cr\hidewidth$\m@th#1\cdot$\hidewidth\cr}%
}
\makeatother

\newtheorem{theorem}{Theorem}[section]
\newtheorem{maintheorem}[theorem]{Main Theorem}
\crefname{maintheorem}{Main Theorem}{Main Theorems}
\newtheorem{proposition}[theorem]{Proposition}
\newtheorem{lemma}[theorem]{Lemma}
\newtheorem{corollary}[theorem]{Corollary}

\newtheorem{fact}[theorem]{Fact}
\newtheorem{observation}[theorem]{Observation}

\theoremstyle{definition}
\newtheorem{remark}[theorem]{Remark}
\newtheorem{definition}[theorem]{Definition}

\newtheorem{example}[theorem]{Example}
\newtheorem{convention}[theorem]{Convention}

\newcommand{\Pow}{\mathcal P}
\newcommand{\pow}[1]{\Pow (#1)}
\newcommand{\Var}{\mathsf{Var}}

\newcommand{\Prop}{\mathsf{Pr}}
\newcommand{\war}{\mathsf{Wl}}

\newcommand{\M}{\mathfrak{M}}
\newcommand{\PP}{\mathfrak{P}}
\newcommand{\B}{\mathfrak{B}}
\newcommand{\RRR}{{\mathfrak R}}

\renewcommand{\models}{\vDash}
\newcommand{\modelsM}[1]{\models_{#1}}

\newcommand{\Fmla}{\mathsf{Fm}}
\newcommand{\lFmla}{\ell\mathsf{Fm}}

\newcommand{\FV}{\mathsf{FV}}
\newcommand{\fv}[1]{\FV(#1)}

\newcommand{\limp}{\to}

\renewcommand{\blacksquare}{\Boxblack}
\renewcommand{\blacklozenge}{\Diamondblack}

\newcommand{\R}{R}

\newcommand{\rels}[1]{\mathbf {#1}}
\newcommand{\RR}{{\rels {R}}}

\newcommand{\proves}{\vdash}

\newcommand{\stoup}{\, |\, }

\newcommand{\sequent}{\mathcal S}

\newcommand{\n}[1]{#1^N}
\renewcommand{\gg}[1]{\n{#1}}

\newcommand{\id}{\mathsf{id}}

\newcommand{\lr}[1]{#1_{l}}
\newcommand{\rr}[1]{#1_{r}}

\newcommand{\wk}{\mathsf{w}}

\newcommand{\cntr}{\mathsf{c}}
\newcommand{\cut}{\mathsf{cut}}

\newcommand{\infrul}{\mathsf{r}}

\newcommand{\K}{\mathsf{K}}

\newcommand{\Kt}{\K\mathsf t}
\newcommand{\KtSO}{\Kt 2}

\newcommand{\WK}{\mathsf{WK}}

\newcommand{\iK}{\mathsf{iK}}

\newcommand{\IPL}{\mathsf{IPL}}

\renewcommand{\mp}{\mathsf{mp}}

\newcommand{\CK}{\mathsf{CK}}

\newcommand{\IK}{\mathsf{IK}}

\newcommand{\IKt}{\IK \mathsf{t}}
\newcommand{\IKtSO}{\IKt 2}

\newcommand{\IKtSOdia}{\IKtSO(\Diamond, \blacklozenge)}

\newcommand{\lab}[1]{\ell #1}

\newcommand{\labKt}{\lab\Kt}
\newcommand{\labKtSO}{\lab\KtSO}

\newcommand{\labIKtSO}{\lab\IKtSO}

\newcommand{\multi}[1]{\mathsf{m}#1}
\newcommand{\mlIKtSO}{\multi\labIKtSO}

\newcommand{\diadistlor}{\mathsf N_{\Diamond\lor}}
\newcommand{\diadistbot}{\mathsf N_{\Diamond\bot}}
\newcommand{\diaimpbox}{\mathsf I_{\Diamond \Box}}

\newcommand{\diaimpboxb}{\mathsf I_{\blacklozenge \blacksquare}}

\newcommand{\gen}{\mathsf{gen}}
\newcommand{\necw}{\mathsf{nec}_{\Box}}
\newcommand{\necb}{\mathsf{nec}_{\blacksquare}}

\newcommand{\Kax}{\mathsf K}
\newcommand{\Sax}{\mathsf S}
\newcommand{\Funct}{\mathsf D}
\newcommand{\funct}[1]{\Funct_{#1}}

\newcommand{\functforall}{\funct\forall}
\newcommand{\functwb}{\funct\Box}
\newcommand{\functwd}{\funct\Diamond}
\newcommand{\functbb}{\funct\blacksquare}
\newcommand{\functbd}{\funct\blacklozenge}

\newcommand{\CA}{\mathsf{C}}

\newcommand{\Vacax}{\mathsf{V}}

\newcommand{\Tense}{\mathsf A}
\newcommand{\wbbd}{\Tense_{\Box\blacklozenge}}
\newcommand{\bdwb}{\Tense_{\blacklozenge\Box}}
\newcommand{\bbwd}{\Tense_{\blacksquare\Diamond}}
\newcommand{\wdbb}{\Tense_{\Diamond\blacksquare}}

\newcommand{\States}{\Omega}

\newcommand{\Worlds}{W}

\newcommand{\set}[1]{\mathcal{#1}}
\newcommand{\setW}{{\set {W}}}

\newcommand{\interp}[3]{#3_{#1}^{#2}}
\newcommand{\interpM}[2]{#2_{#1}}

\newcommand{\interpa}[2]{#2^{#1}}

\newcommand{\Baire}{\lFmla^{<\omega}}

\newcommand{\extension}[1]{[#1]}

\newcommand{\prfsrchtree}{\mathfrak{S}}
\newcommand{\possext}[1]{\mathtt{#1}}
\newcommand{\textends}{\geq_{\prfsrchtree}}
\newcommand{\minext}[1]{\lfloor {#1} \rfloor }

\newcommand{\lfm}[3][]{\mathtt{lfm}_{#1}^{#2}(#3)}
\newcommand{\rfm}[3][]{\mathtt{rfm}_{#1}^{#2}(#3)}
\newcommand{\ifm}[3][]{\mathtt{ifm}_{#1}^{#2}(#3)}
\newcommand{\cfm}[3][]{\mathtt{cfm}_{#1}^{#2}(#3)}

\newcommand{\conn}[1]{\overset{#1}{\leftrightsquigarrow}}
\newcommand{\RAtMin}[2]{\RR_{#1\setminus#2}}
\newcommand{\RAtMinMin}[3]{\RR_{#1\setminus#2\setminus#3}}

\newcommand{\Diamondplus}{\mathord{\scalebox{.9}{\raise.1ex\hbox{$\diamondplus$}}}}
\newcommand{\Diamondtimes}{\mathord{\scalebox{.9}{\raise.1ex\hbox{$\diamondtimes$}}}}

\newcommand{\Fl}[2][]{F_{#1}^{\AND} \{ #2 \}}
\newcommand{\Fr}[2][]{F_{#1}^{\IMP} \{ #2 \}}

\newcommand{\spa}{\;|\;}
\newcommand{\LLra}{\Longleftrightarrow}

\newcommand{\W}{\mathcal{W}}
\newcommand{\vP}{v^\prime}
\newcommand{\vPP}{v^{\prime \prime}}
\newcommand{\wP}{w^\prime}
\newcommand{\tBox}{\blacksquare}
\newcommand{\Pred}{\mathfrak{P}}
\newcommand{\Dmnd}{\DIA}
\newcommand{\Ra}{\Rightarrow}
\newcommand{\tDmnd}{\blacklozenge}
\newcommand{\Band}[1]{\underset{#1}{\bigwedge}}
\newcommand{\Bor}[1]{\underset{#1}{\bigvee}}

\newcommand{\AND}{\mathbin{\wedge}}
\newcommand{\OR}{\mathbin{\vee}}
\newcommand{\IMPLIES}{\mathbin{\to}}
\newcommand{\IMP}{\IMPLIES}
\newcommand{\IFF}{ \leftrightarrow }
\newcommand{\DIA}{\Diamond}
\newcommand{\BOX}{\Box}
\newcommand{\DIAB}{\Diamondblack}
\newcommand{\BOXB}{\Boxblack}
\newcommand{\bottom}{\bot}
\newcommand{\BOT}{\bottom}
\newcommand{\TOP}{\top}

\newcommand{\brelsat}[3][\B]{#2 \models_{#1} #3}

\newcommand{\brel}[1][\B]{R_{#1}}
\newcommand{\irel}{\leq}
\newcommand{\irelop}{\geq}
\newcommand{\class}{\setW}
\newcommand{\fall}[1][X]{\forall #1}
\newcommand{\subst}[3][X]{#2[#3/#1]}
\newcommand{\fandempt}[1][F]{#1^{\AND}}
\newcommand{\fand}[2][F]{\fandempt[#1] \{ #2 \}}
\newcommand{\fimp}[2][F]{#1^{\IMP} \{ #2 \}}
\newcommand{\pprove}[2][\IKtSO]{#1 \vdash #2}
\newcommand{\unimod}{\mathfrak{R}}

\newcommand{\fm}[2]{#1 : #2}
\newcommand{\seq}[3]{#1 \stoup #2 \seqarrow #3}
\newcommand{\lseq}[2]{#1 \stoup #2}
\newcommand{\seqarrow}{\Rightarrow}

\newcommand{\symb}{\Boxdot}
\newcommand{\symbdiamond}{\Diamonddot}
\newcommand{\kax}[1]{%
  \ifnum #1=1 \mathsf{F}_{\BOX}%
  \else\ifnum #1=2 \mathsf{F}_{\DIA}%
  \else ? \fi\fi}
\newcommand{\kaxb}[1]{%
  \ifnum #1=1 \mathsf{F}_{\BOXB}%
  \else\ifnum #1=2 \mathsf{F}_{\DIAB}%
  \else ? \fi\fi}
\newcommand{\kaxsymb}[1]{%
  \ifnum #1=1 \mathsf{F}_{\symb}%
  \else\ifnum #1=2 \mathsf{F}_{\symbdiamond}%
  \else ? \fi\fi}

  \newcommand{\pre}[1]{#1^-}

  \newcommand{\brcol}[1]{\B_{#1}}
  \newcommand{\brcoldat}[2]{#2_{\brcol #1}}
  \newcommand{\IH}{\mathit{IH}}

\DeclareSymbolFont{symbolsA}{U}{txsya}{m}{n}
\DeclareMathSymbol{\Box}{\mathord}{symbolsA}{3}
\DeclareMathSymbol{\Boxblack}{\mathord}{symbolsA}{4}
\DeclareMathSymbol{\Boxdot}{\mathord}{symbolsA}{0}
\DeclareMathSymbol{\Boxplus}{\mathord}{symbolsA}{1}
\DeclareMathSymbol{\Boxtimes}{\mathord}{symbolsA}{2}

\DeclareSymbolFont{symbolsC}{U}{txsyc}{m}{n}
\DeclareMathSymbol{\Diamond}{\mathord}{symbolsC}{94}
\DeclareMathSymbol{\Diamondblack}{\mathord}{symbolsC}{95}
\DeclareMathSymbol{\Diamonddot}{\mathord}{symbolsC}{144}

\DeclareFontFamily{U} {MnSymbolC}{}
\DeclareFontShape{U}{MnSymbolC}{m}{n}{
	<-6>  MnSymbolC5
	<6-7>  MnSymbolC6
	<7-8>  MnSymbolC7
	<8-9>  MnSymbolC8
	<9-10> MnSymbolC9
	<10-12> MnSymbolC10
	<12->   MnSymbolC12}{}
\DeclareFontShape{U}{MnSymbolC}{b}{n}{
	<-6>  MnSymbolC-Bold5
	<6-7>  MnSymbolC-Bold6
	<7-8>  MnSymbolC-Bold7
	<8-9>  MnSymbolC-Bold8
	<9-10> MnSymbolC-Bold9
	<10-12> MnSymbolC-Bold10
	<12->   MnSymbolC-Bold12}{}
\DeclareSymbolFont{MnSyC}{U}{MnSymbolC}{m}{n}
\DeclareMathSymbol{\diamondplus}{\mathord}{MnSyC}{124}
\DeclareMathSymbol{\diamondtimes}{\mathord}{MnSyC}{125}

\title[Second-order (inuitionistic) tense logic]{The proof theory and semantics of \\ second-order (intuitionistic) tense logic}
\date{\today}
\author[]{Justus Becker
\and
Anupam Das
\and
Sonia Marin
\and
Paaras Padhiar}

\begin{document}

\begin{abstract}
We develop a second-order extension of intuitionistic modal logic, allowing quantification over propositions, both syntactically and semantically. 
A key feature of second-order logic is its capacity to define positive connectives from the negative fragment. 
Duly we are able to recover the diamond (and its associated theory) using only boxes, as long as we include both forward and backward modalities (`tense' modalities). 

We propose axiomatic, proof theoretic and model theoretic definitions of `second-order intuitionistic tense logic', and ultimately prove that they all coincide. 
In particular we establish completeness of a labelled sequent calculus via a proof search argument,
yielding at the same time a cut-admissibility result.
Our methodology also applies to the classical version of second-order tense logic, which we develop in tandem with the intuitionistic case.

\end{abstract}
\maketitle

\section{Introduction}

\subsection{Background and motivation}

\emph{Second-order} logic extends first-order logic by allowing quantification over predicates.
It has become a standard tool in mathematical and computational logic, including in type theory and programming languages (e.g.\ \cite{Girard1972:thesis,Reynolds1974Towards,Milner1978Theory,parigot1997proofs}), computability and complexity (e.g.\ \cite{Simpson2009Subsystems,Dzhafarov2022Reverse,Buss1986Bounded,Cook2010Logical}), and, more recently, knowledge representation in artificial intelligence (e.g.\ \cite{belardinelli2015epistemic,belardinelli2016semantical,belardinelli2018second-order,blackburn2006handbook}).
The metalogical and proof theoretic foundations of second-order logic have posed significant challenges to logicians over the last century. 
Completeness for \emph{standard} (or \emph{full}) semantics, where properties vary over the full powerset, fails, necessitating the more general \emph{Henkin} semantics~\cite{henkin1950completeness}.
The corresponding theories typically admit the \emph{comprehension} axiom, requiring impredicative techniques for metalogical analysis.
Indeed \emph{cut-admissibility} of second-order logic, known as \emph{Takeuti's conjecture} \cite{sep-proof-theory}, remained an open problem since the '30s, before being resolved in the late '60s by seminal works of Tait \cite{Tait1966Nonconstructive}, Prawitz \cite{Prawitz1968:SOL,Prawitz1968:STT}, Takahashi \cite{Takahashi1967} and Girard \cite{Girard1972:thesis}.

Over \emph{intuitionistic} logic, second-order quantification notably allows the encoding of positive connectives by the negative fragment.
For example $A \lor B$ is logically \emph{equivalent} to $\forall X ((A \limp X) \limp (B\limp X) \limp X) $.
This is in stark contrast to the first-order setting, where connectives are infamously intuitionistically independent. 
In this work we apply the second-order methodology to \emph{modal logic}, allowing us to \emph{recover} a theory of positive modalities ($\Diamond$s) from the negative ($\Box$s) over an intuitionistic base.

Unlike propositional and predicate logic, there is no consensus on what the \emph{intuitionistic} fragment of modal logic is.
While it is natural to admit distribution, $\Box (A\limp B) \limp \Box A \limp \Box B$ and necessitation, $A / \Box A$, these principles say nothing about $\Diamond $ which, recall, cannot be defined in terms of $\Box$. 
For instance one might also want to admit $\Diamond$-distribution, $\Box (A\limp B) \limp \Diamond A \limp \Diamond B$.
This has led to several distinct proposals, e.g.~\cite{bellin2001extended,wijesekera1990constructive,das2023intuitionistic,degroot2025semantical,fischerservi1984axiomatizations,simpson1994proof}.
Notably the choice of $\Diamond$-axioms also affects even the $\Box$-only (i.e.~$\Diamond$-free) theorems  of the logic \cite{grefe1999fischer,das2023intuitionistic}. 
(See \cite{dasmarin:blog22} for a related survey.)

Instead we show that we can \emph{encode} the $\Diamond$ over second-order intuitionistic logic, similarly to $ \lor$ earlier, as long as we admit also a \emph{backwards} box $\blacksquare$, as in \emph{tense} logic~\cite{prior1957time}:
\begin{equation}
\label{eq:diamond-equiv-so-dfn-intro}
    \Diamond A \ \iff \ \forall X (\Box (A \limp \blacksquare X) \limp  X)
\end{equation}
A symmetric equivalence holds for the backwards diamond $\blacklozenge$.
This allows us to \emph{recover} a theory of $\Diamond$ (and $\blacklozenge$), instead of choosing an arbitrary axiomatisation.
The resulting second-order logic $\IKtSO$ conservatively extends Fischer Servi and Simpson's $\IK$ \cite{fischerservi1984axiomatizations,simpson1994proof} and, furthermore, Ewald's $\IKt$ \cite{ewald1986intuitionistic}.

\begin{figure*}[t]
    \centering
    \subcaptionbox{Classical results\label{fig:class-tour}}
    {
 \begin{tikzpicture}
        \node (Kt2) at (-4,3) {$\KtSO\proves A$};
        \node (rel) at (0,3) {$\forall \unimod \, \modelsM \unimod A$}; 
        \node (lKt2) at (0,0) {$\labKtSO\setminus\cut\proves A$};
        \node (lKt2c) at (-4,0) {$\labKtSO\proves A$};
        \draw [->] (Kt2) to node [below=.15cm]
        {\cref{classical-rel-soundness}
        } (rel) ;
        \draw [->] (rel) to node [right] 
        {\cref{cut-free-completeness-classical}
        } (lKt2);
        \draw[left hook->] (lKt2) to (lKt2c);
        \draw[->] (lKt2c) to node [left]
        {\cref{prop:axiomsoundness-classical}\eqref{axiomsoundness-classical}
        } (Kt2);
        \draw[dashed, blue, ->] (rel) edge[bend right=20] node[above] 
        {
        {\cref{mainthm:soundness-completeness-classical}}} (Kt2);
        \draw[dashed, blue, ->] (lKt2c) edge[bend left=20] node [above] 
        {\cref{mainthm:hauptsatz}\eqref{hauptsatz-classical}} (lKt2);
    \end{tikzpicture}
    }
    \subcaptionbox{Intuitionistic results\label{fig:int-tour}}{
    \begin{tikzpicture}
        \node (IKt2) at (-4,3) {$\pprove[\IKtSO]{A}$};
        \node (birel) at (0,3) {$\forall \B \, \modelsM \B A$}; 
        \node (pred) at (4,3) {$\forall \PP \, \modelsM \PP A$}; 
        \node (mlIKt2) at (4,0) {$\pprove[\mlIKtSO\setminus\cut]{A}$};
        \node (lIKt2) at (0,0) {$\pprove[\labIKtSO\setminus\cut]{A}$};
        \node (lIKt2c) at (-4,0) {$\pprove[\labIKtSO]{A}$};
        \draw [->] (IKt2) to node [below=.15cm] 
        {\cref{prop:brelsoundness}
        } (birel) ;
        \draw [right hook->] (birel) to node [below=.15cm]
        {\cref{prop:breltoprel}
        } (pred) ; 
        \draw [->] (pred) to node 
        {\cref{thm:completeness}
        } (mlIKt2);
        \draw[->] (mlIKt2) to node [above=.15cm]
        {\cref{prop:multitosingle}
        } (lIKt2);
        \draw[left hook->] (lIKt2) to (lIKt2c);
        \draw[->] (lIKt2c) to node 
        {\cref{prop:axiomsoundness}\eqref{axiomsoundness-intuitionistic}
        } (IKt2);
        \draw[dashed, blue, ->] (birel) edge[bend right=20] node[above] {\cref{mainthm:soundness-completeness-intuitionistic}} (IKt2);
        \draw[dashed, blue, ->] (lIKt2c) edge[bend left=20] node[above] {\cref{mainthm:hauptsatz}\eqref{hauptsatz-intuitionistic}} (lIKt2);
    \end{tikzpicture}
    }
    \caption{The `grand tours' of this work.}
    \label{fig:tours}
\end{figure*}
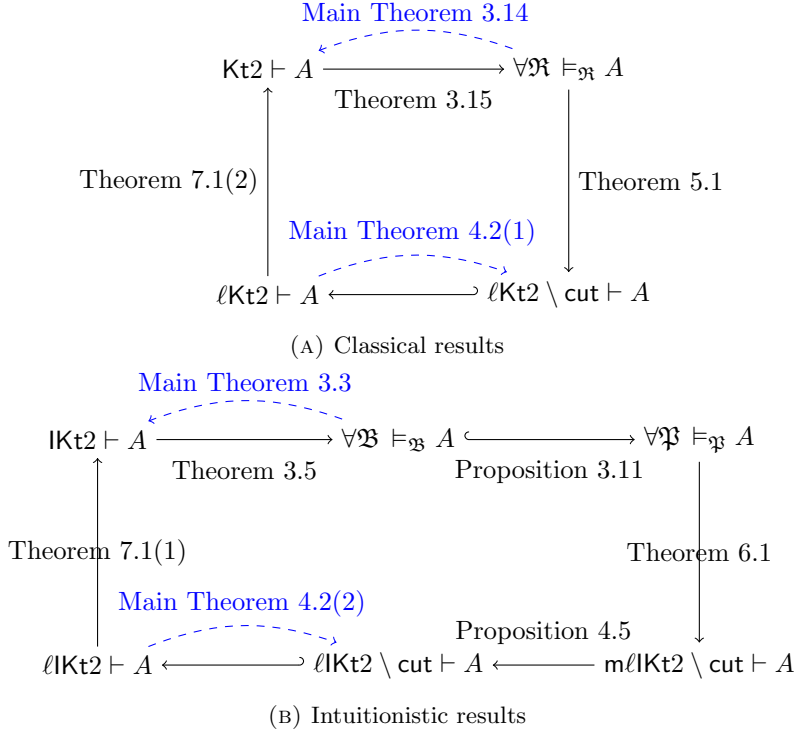

\subsection{Contributions}

In this work we develop axiomatic, semantic and proof theoretic foundations for second-order tense logic, over both classical and intuitionistic bases.
Our main results are: 
\begin{enumerate}[(i)]
    \item soundness and completeness of our systems for corresponding (bi)relational semantics; and,
    \item cut-admissibility for associated \emph{labelled} sequent calculi. 
\end{enumerate}
These results are recovered from the `grand tours' visualised in \cref{fig:tours}, where they are indicated in blue, dashed.
Our main completeness arguments, \cref{cut-free-completeness-classical,thm:completeness}, are proved via \emph{proof search}, notably adapting Sch\"utte's notion of \emph{semivaluation} \cite{Schutte1960:semivaluations}
to the intuitionistic setting in order to overcome the hurdle of impredicativity (cf.~\cite{Tait1966Nonconstructive,Prawitz1968:STT,prawitz1970some}).
The main technicalities behind the converse direction lie in the translation from labelled proofs to axiomatic proofs, \cref{prop:axiomsoundness-classical,prop:axiomsoundness}, driven by a tree-like representation of labelled sequents in tense logic (see, e.g.,~\cite{ciabattoni2021display}).


\subsection{Related works}

Second-order extensions of (classical) modal logic have been proposed since the '70s, cf.~\cite{bull1969modal,fine1970propositional,kaplan1970s5}.
Many aspects have since been investigated, including expressivity\cite{ten2006expressivity,kaminski1996expressive}, applications to \emph{interpolation} theory~\cite{fitting2002interpolation,bilkova2007uniform}, and to the meta-theory of provability logics~\cite{artemov1993propositional}. 
However most of these works focus on full semantics, and so are not suitable for proof theoretic investigations.

More recently, second-order modal logics have been proposed as a \emph{specification language} for knowledge representation in artificial intelligence~\cite{belardinelli2015epistemic,belardinelli2016semantical,belardinelli2018second-order}.
This is more relevant to our approach as they interpret the second-order language over Henkin structures rather than full structures, although they only treat quantifier-free (or \emph{predicative}) comprehension.
Notably in \cite{belardinelli2018second-order} the authors conclude that, for second-order modal logic ``\emph{to be adopted as a specification language in artificial intelligence and knowledge representation, appropriate theoretical results and formal tools need to be developed}''.
Our work may be viewed as a contribution in this direction, developing the metalogical and proof theoretic foundations therein.

\subsection{Structure of the paper}
Our main results are visualised in \cref{fig:tours}.
In \cref{sec:axiomatisation}
we introduce the axiomatisations $\KtSO$ and $\IKtSO$ of classical and intuitionistic second-order tense logic, respectively.
Both include a full \emph{comprehension} axiom, and duly prove \cref{eq:diamond-equiv-so-dfn-intro}.
In \cref{sec:semantics} we define relational semantics for $\KtSO$, and two extensions by a partial order for $\IKtSO$, following Simpson's methodology~\cite{simpson1994proof}.
The main results of this section are \emph{soundness} of the axiomatisations for their semantics, \cref{classical-rel-soundness,prop:brelsoundness,prop:breltoprel}.

In \cref{sec:prooftheory} we present \emph{labelled} sequent calculi $\labKtSO$ and $\labIKtSO$, again inspired by Simpson \cite{simpson1994proof}, in particular including a left-$\forall$ rule implementing full comprehension. We also present an intermediate \emph{multi-succedent} version $\mlIKtSO$, \`a la Maehara \cite{maehara1954}, more suitable for completeness-via-proof-search arguments (cf., e.g., \cite{prawitz1970some} and \cite[Section~15]{Takeuti1987:pt-book}), that is ultimately conservative over $\labIKtSO$, \cref{prop:multitosingle}.
Our main proof search arguments are presented in \cref{sec:completeness-classical} (classical) and \cref{sec:completeness} (intuitionistic), yielding cut-free completeness of $\labKtSO$, \cref{cut-free-completeness-classical}, and $\mlIKtSO$, \cref{thm:completeness}.
The intuitionistic case, in particular, exhibits novel intricacies, combining ideas from analogous constructions for simple type theory \cite{Prawitz1968:STT} and first-order intuitionistic logic \cite[Section~15]{Takeuti1987:pt-book}.

In \cref{sec:labelsoundness} we present a translation from labelled systems $\labKtSO$ and $\labIKtSO$ to axiomatic systems $\KtSO$ and $\IKtSO$, \cref{prop:axiomsoundness-classical,prop:axiomsoundness}, respectively, completing the cycle of implications, viz.~\cref{fig:class-tour,fig:int-tour}.
Finally we conclude the paper with some additional perspectives in \cref{sec:further}, in particular relating $\KtSO$ and $\IKtSO$ by a \emph{negative} translation, and make some concluding remarks in \cref{sec:conclusions}.

\section{A second-order extension of tense logic}
\label{sec:axiomatisation}

In this section we present the syntax of second-order tense logic and provide axiomatisations of its intuitionistic and classical theory.
One of the main points here is to show how diamonds may actually be defined in terms of the other connectives, even intuitionistically, allowing us to \emph{recover} its theory, instead of providing arbitrary axioms therein.

We point the reader to, e.g., \cite[Chapter~12]{sorensen2006lectures} for some useful background on second-order (intuitionistic) propositional logic.

\subsection{Syntax of second-order tense logic}
Let us fix a set $\Prop$ of \textbf{propositional symbols}, written $P,Q,R$ etc., and a (disjoint) set $\Var$ of \textbf{(formula/second-order) variables}, written $X,Y,Z$ etc.
We shall work with second-order tense \textbf{formulas}, given by:
\[
A,B,C,\dots \quad ::= \quad P \in \Prop\  \mid\  
X\in \Var \ \mid \ A \limp B 
\ \mid \ 
\Box A \ \mid \  \blacksquare A
\ \mid \ \forall X A
\]
Write $\Fmla$ for the set of all formulas.
We shall frequently omit external brackets of formulas and internal brackets of long implications, understanding them as rightmost bracketed. I.e.\ $A\limp B\limp C = (A \limp (B\limp C))$ and so on.

The set of \textbf{free variables} of a formula $A$, written $\fv A$, is defined as expected:
\[
\begin{array}{r@{\ := \ }l}
	\fv P & \emptyset \\
	\fv X & \{X\} \\
	\fv {A\limp B} & \fv A \cup \fv B 
\end{array}
\qquad
\begin{array}{r@{\ := \ }l}
	\fv {\Box A} & \fv A \\
	\fv {\blacksquare A} & \fv A \\
	\fv {\forall X A} & \fv A \setminus \{X\}
\end{array}
\]
Note that propositional symbols are \emph{not} variables. 
If $\fv A = \emptyset$ then $A$ is \textbf{closed}
(or a \textbf{sentence}). 
Otherwise it is \textbf{open}.

\begin{remark}
	[Propositional symbols vs variables]
	We could have worked without propositional symbols at all, using only variables.
	However we take the current formulation so that we may safely deal with only closed formulas in the systems and semantics we present. 
	This choice allows us to avoid the need for explicit \emph{environments} when interpreting syntax, lightening the notation therein.
\end{remark}

It is well known that other propositional connectives and quantifiers can be \emph{defined} from our minimal syntax via \emph{impredicative} encodings.
In particular, over pure second-order intuitionistic logic, we have the following equivalences:\footnote{In all cases, the variable $X$ should be chosen not occurring free in $A$ and $B$. Since we shall typically only deal with closed formulas, we shall gloss over this technicality throughout.}
\begin{equation}
	\label{eq:so-dfns-of-nonmodal-connectives}
	\begin{array}{r@{\ \iff \ }l}
		\BOT & \forall X X \\
		A \lor B & \forall X ((A \limp X) \limp (B\limp X) \limp X)  \\
		A\land B & \forall X ( (A \limp B \limp X) \limp X) \\
		\exists Y A & \forall X (\forall Y(A \limp X) \limp X)
	\end{array}
\end{equation}
Note that these equivalences hold even \emph{intuitionistically}, in stark contrast to (first-order) intuitionistic logic: there all the propositional connectives are independent.
See, e.g., \cite[Section~12.4]{sorensen2006lectures} or \cite[Section~11.3]{girard1989proofs-and-types} for a more detailed account.
In the same vein, we will be able to \emph{define} diamond modalities in the second-order setting by appropriate equivalences:
\begin{equation}
	\label{eq:diamond-definitions}
	\begin{array}{r@{\ \iff \ }l}
		\Diamond A & \forall X (\Box (A \limp \blacksquare X) \limp  X) \\
		\blacklozenge A & \forall X ( \blacksquare ( A \limp \Box X) \limp  X)
	\end{array}
\end{equation}

For now this is all rather informal, for the sake of motivation, as we have not yet given any meaning to our formulas. To justify these equivalences, we shall now turn to axiomatisations over our syntax, before presenting semantics in the next section.

\subsection{A minimal axiomatisation}
We shall present only a minimal axiomatisation of formulas in this work.
There are two reasons for this: 
\begin{enumerate}[(i)]
	\item\label{item:minimal-since-diamond-dfn} we want to justify the equivalences from \cref{eq:diamond-definitions} in the most general way, for \emph{any} extension of our axiomatisation; and,
	\item\label{item:minimal-since-robust} we will later argue for the robustness of the minimal axiomatisation.
\end{enumerate}
Towards~\cref{item:minimal-since-diamond-dfn}, let us temporarily expand the language of formulas by unary operators $\Diamond$ and $\blacklozenge$.
We shall later drop these once we demonstrate that they are unnecessary.
We consider a minimal axiomatisation extending second-order intuitionistic propositional logic $\IPL 2$ only by normality of modalities, and adjunction of the pairs $(\Box, \blacklozenge)$ and $(\blacksquare , \Diamond)$:

\begin{definition}
	[Axiomatisation with diamonds]
	\label{dfn:ikt2}
	$\IKtSOdia $ is the logic axiomatised by:
	\begin{enumerate}
		\item\label{ipl2-axioms-rules} All of second-order intuitionistic propositional logic $\IPL2$, i.e.\ the axioms and rules:\footnote{Note that the propositional axioms and rules we give are rather those of \emph{minimal} logic, without $\bot$. However in the second-order setting, given the definability of the latter, the difference between `minimal' and `intuitionistic' disappears.}
		\[
		\begin{array}{r@{\ :\quad }l}
			\Kax & A \limp B \limp A \\
			\Sax & (A \limp B \limp C) \limp (A \limp B) \limp A \limp C \\
			\functforall &  \forall X (A \limp B) \limp \forall X A \limp \forall X B  \\
			\Vacax &   A \limp \forall X A \text{ (when $(X \notin \fv A)$)} \\
			\noalign{\smallskip}
			\CA &   \forall X A \limp A[C/X]
		\end{array}
		\qquad
		\begin{array}{c}
			\vliinf{\mp}{}{B}{A\limp B}A
			\\
			\noalign{\smallskip}\qquad \quad \vlinf{\gen}{\text{$P$ fresh}}{\forall X A}{A[P/X]}
		\end{array}
		\]
		where `$P$ fresh' means that the propositional symbol $P$ does not occur in the conclusion of the rule.

		\item\label{modal-axioms-rules-wb} Normality of white and black modalities, i.e.~the distributivity axioms and necessitation rules:
		\[
		\begin{array}{r@{\ : \quad }l}
			\functwb & \Box (A \limp B) \limp \Box A \limp \Box B \\
			\functwd & \Box (A \limp B) \limp \Diamond A \limp \Diamond B \\
			\noalign{\smallskip}
			\functbb & \blacksquare (A\limp B) \limp \blacksquare A \limp \blacksquare B \\
			\functbd & \blacksquare (A\limp B) \limp \blacklozenge A \limp \blacklozenge B
		\end{array}
		\qquad
		\begin{array}{c}
			\vlinf{\necw}{}{\Box A}{A} \\
			\noalign{\smallskip}
			\vlinf{\necb}{}{\blacksquare A}{A}
		\end{array}
		\]

		\item\label{tense-axioms} Adjunction of $(\Box, \blacklozenge)$ and $(\blacksquare,\Diamond)$:
		\[
		\begin{array}{r@{\ : \quad}l}
			\bdwb & \blacklozenge\Box A \limp A \\
			\wbbd & A \limp \Box \blacklozenge A 
		\end{array}
		\qquad
		\begin{array}{r@{\ : \quad}l}
			\wdbb & \Diamond \blacksquare A \limp A \\
			\bbwd & A \limp \blacksquare \Diamond A 
		\end{array}
		\]
	\end{enumerate}
\end{definition}

\begin{remark}
	[Full comprehension]
	Note that the choice of $C$ in the \emph{comprehension} axiom $\CA$ is unrestricted:
	the formula $A[C/X]$ may be more complex than $\forall X A$.
	For instance we can even set $C = \forall X A$.
	This apparent circularity (known as \emph{impredicativity}) complicates, e.g., the semantics of second-order logic, as we shall see in \cref{sec:semantics}.
\end{remark}

\begin{remark}
	[Other logical operators]
	As already mentioned, if we formulated $\IPL2$ with native operators $\bot, \lor, \land, \exists$, then the equivalences in \cref{eq:so-dfns-of-nonmodal-connectives} are all already provable.
	For example $\forall X X \limp \bot$ is just an instance of $\CA$ and $\forall X ((A\limp X) \limp (B\limp X) \limp X) \limp A\lor B$ is derivable by setting $C:= A\lor B$ in $\CA$.
	We thus henceforth freely use those logical connectives in what follows, recasting those equivalences as \emph{definitions}.
	We also write $\lnot A := A \limp \bot$ and $A \liff B := (A\limp B) \land (B\limp A)$.
\end{remark}

\begin{remark}
	[Minimality]
	\label{rem:minimality}
	Over $\IPL$ the axioms and rules for white (or black) modalities from \cref{modal-axioms-rules-wb} determine what is known as \emph{constructive modal logic} $\CK$~\cite{bellin2001extended}. 
	This is the smallest intuitionistic version of modal logic (with $\Box$ and $\Diamond$) usually considered, with other common ones obtained by adding further principles of (classical) modal logic, in particular among:
	\begin{equation}
		\label{eq:other-IK-axioms}
		\begin{array}{l@{\ : \quad}l}
			\diadistlor &     \Diamond (A_0\lor A_1) \limp \Diamond A_0 \lor \Diamond A_1 \\
			\diadistbot &    \Diamond \bot \limp \bot \\
			\diaimpbox &    (\Diamond A \limp \Box B) \limp \Box (A\limp B)
		\end{array}
	\end{equation}
	For instance $\WK := \CK + \diadistbot$ was studied in~\cite{wijesekera1990constructive}, $\WK + \diadistlor$ was studied in~\cite{das2023intuitionistic}, $\IK^N := \CK + \diadistbot + \diaimpbox$ was studied in~\cite{degroot2025semantical}, and $\IK$ is the extension of $\CK$ by all the axioms above~\cite{fischerservi1984axiomatizations,simpson1994proof}. 
	The smallest intuitionistic modal logic without $\Diamond$, $\iK$, defined as the extension of $\IPL$ by $\functwb$ and $\necw$, is conservatively extended by $\CK$.
	In this sense \cref{modal-axioms-rules-wb} is a minimal commitment in terms of extending the underlying modal logics.
	
	On the other hand, the tense axioms in \cref{tense-axioms} state only that $(\Box, \blacklozenge)$ and $(\blacksquare,\Diamond)$ are adjoint pairs, assuming no further relationship.
	These are standard axioms in presentations of tense logic~\cite{blackburn2006handbook},
	and so again~\cref{tense-axioms} is a minimal commitment in terms of relating the white and black modalities.
\end{remark}

\begin{example}
	[$\Box$ distributes over $\forall$]
	\label{ex:box-dist-forall}
	Let us see a simple example of $\IKtSO (\Diamond, \blacklozenge)$ reasoning in action, not least so we can explain how we present axiomatic proofs:
	\[
	\begin{array}{ll}
		\forall X A \limp A[P/X] & \text{by $\CA$, setting $X=P$} \\
		\Box \forall X A \limp \Box A[P/X] & \text{by $\necw$ and $\functwb$} \\
		\Box \forall X A \limp \forall X \Box A & \text{by $\gen, \functforall, \Vacax$}
	\end{array}
	\]
	Note that we leave routine $\IPL$ reasoning here mostly implicit, rather focussing on the modal and quantifier axioms necessary at each step. 
	To expand out the final step a little, consider the following gadget, where $X\notin \fv A$:
	\[
	\begin{array}{ll}
		A \limp B & \\
		\forall X (A\limp B) & \text{by $\gen$} \\
		\forall X A \limp \forall X B & \text{by $\functforall$} \\
		A \limp \forall X B & \text{by $\Vacax$}
	\end{array}
	\]    
	We shall continue to write axiomatic proofs in this fashion henceforth, without additional explanation.
	
	Note that we can also recover a proof of $\blacksquare\forall X A \limp \forall X \blacksquare A$, by simply exchanging white and black modalities in the proof above.
	We shall also continue to use this observation throughout, now just alluding to `symmetry'. 
\end{example}

Before showing further examples involving diamonds, 
let us first justify \cref{eq:diamond-definitions}, as promised:
\begin{theorem}
	\label{dia-iff-so-dfn}
	$\IKtSOdia$ (and so all its extensions) proves:
	\begin{enumerate}
		\item $\Diamond A \liff \forall X (\Box (A \limp \blacksquare X) \limp  X) $
		\item $     \blacklozenge A \liff \forall X ( \blacksquare ( A \limp \Box X) \limp  X)$
	\end{enumerate}
\end{theorem}
\begin{proof}
	By symmetry, we need only prove the first item, for which we give each direction separately:
	\[
	\begin{array}{ll}
		\Diamond \blacksquare P \limp P & \text{by $\wdbb$} \\
		\Diamond A \limp (\Diamond A \limp \Diamond \blacksquare P) \limp P & \text{by $\IPL$ reasoning} \\
		\Diamond A \limp \Box (A \limp \blacksquare P) \limp P & \text{by $\functwd$}\\
		\Diamond A \limp \forall X (\Box (A\limp \blacksquare X) \limp X) & \text{by $\gen, \functforall, \Vacax$}
		\\
		\noalign{\medskip}
		A \limp \blacksquare\Diamond A & \text{by $\bbwd$}\\
		\Box (A \limp \blacksquare\Diamond A) & \text{by $\necw$} \\
		(\Box (A \limp \blacksquare\Diamond A) \limp \Diamond A) \limp \Diamond A & \text{by $\IPL$ reasoning} \\
		\forall X (\Box (A \limp \blacksquare X) \limp X) \limp \Diamond A & \text{by $\CA$, setting $X = \Diamond A$} \qedhere
	\end{array}
	\]
\end{proof}

As we indicated earlier, we may now dispense with the native diamonds, under the equivalences we have just proved:

\begin{definition}
	$\IKtSO$ is obtained from $\IKtSOdia$ by setting:
	\[
	\begin{array}{r@{\ := \ }l}
		\Diamond A    & \forall X (\Box (A \limp \blacksquare X) \limp X)  \\
		\blacklozenge A   & \forall X (\blacksquare (A \limp \Box X) \limp X)
	\end{array}
	\]
\end{definition}
All further references to $\Diamond $ and $\blacklozenge$ are bound by the definitions displayed above.

\begin{example}
	[Redundancy]
	Note that, under the definitions of $\Diamond,\blacklozenge$ above, some of the axioms we gave in $\IKtSOdia$ become redundant, in the sense that they are already derivable from the others. 
	In particular this is the case for  half of the adjunction axioms, $\wdbb$ and $\bdwb$, and the distribution axioms for diamonds, $\functwd$ and $\functbd$:
	\[
	\begin{array}{ll}
		\blacksquare A \limp \blacksquare A & \text{by $\IPL$ reasoning} \\
		\Box (\blacksquare A \limp \blacksquare A) & \text{by $\necw$} \\
		(\Box (\blacksquare A \limp \blacksquare A) \limp A) \limp A & \text{by $\IPL$ reasoning}\\
		\Diamond \blacksquare A \limp A & \text{by $\CA$ and definition of $\Diamond$}
	\end{array}
	\]
	\[
	\begin{array}{ll}
		(A \limp B) \limp (B\limp \blacksquare P) \limp A \limp \blacksquare P & \text{by $\IPL$ reasoning}\\
		\Box (A\limp B) \limp \Box (B \limp \blacksquare P) \limp \Box (A \limp \blacksquare P) & \text{by $\necw$ and $\functwb$s} \\
		\Box (A\limp B) \limp ((\Box A \limp \blacksquare P) \limp P) \limp \Box (B\limp \blacksquare P) \limp P & \text{by $\IPL$ reasoning} \\
		\Box (A\limp B) \limp \Diamond A \limp \Box (B \limp \blacksquare P) \limp P & \text{by $\CA$ and definition of $\Diamond$} \\
		\Box (A\limp B) \limp \Diamond A \limp \Diamond B & \text{by $\gen, \functforall, \Vacax$ and definition of $\Diamond$}
	\end{array}
	\]
	Again, derivability of $\functbd$ and $\bdwb$ follow by symmetry.

\end{example}

\begin{example}
	[$\forall$ distributes over $\Box$]
	\label{ex:forall-distributes-over-box}\label{ex:Barcan}
	Referring to \cref{ex:box-dist-forall}, we also have the converse principle:\footnote{The reader familiar with predicate modal logics will notice the similarity to the so-called `Barcan' formulas.}
	\[
	\begin{array}{ll}
		\forall X \Box A \limp \Box A[P/X] & \text{by $\CA$, setting $X=P$}\\
		\blacksquare (\forall X \Box A \limp \Box A[P/X]) & \text{by $\necb$} \\
		(\blacksquare (\forall X \Box A \limp \Box A[P/X]) \limp A[P/X]) \limp A[P/X] & \text{by $\IPL$ reasoning} \\
		\blacklozenge \forall X \Box A \limp A[P/X] & \text{by $\CA$, setting $X = A[P/X]$ and definition of $\Diamond$} \\
		\blacklozenge \forall X \Box A \limp \forall X A & \text{by $\gen,\functforall,\Vacax$} \\
		\Box \blacklozenge \forall X \Box A \limp \Box \forall X A & \text{by $\necw$ and $\functwb$} \\
		\forall X \Box A \limp \Box \forall X A & \text{by $\wbbd$}
	\end{array} 
	\]
	The two directions together comprise a sort of infinitary version of the usual distribution of $\Box$s over $\land$: $\Box (A\land B) \liff \Box A \land \Box B$.
	
	Again, by symmetry, we also have $\blacksquare \forall X A \liff \forall X \blacksquare A$.

\end{example}

\subsection{On the underlying (first-order) modal and tense logics}
\label{sec:underlying-modal-tense-logics}
Now that we have addressed \cref{item:minimal-since-diamond-dfn}, the general correctness of our definition of diamonds, let us move to \cref{item:minimal-since-robust}, the robustness of the minimal axiomatisation $\IKtSO$ we have presented.

As we mentioned earlier in \cref{rem:minimality}, intuitionistic modal logics found in the literature may contain further axioms among \cref{eq:other-IK-axioms}.
It turns out that these are all derivable within the current presentation:

\[
\begin{array}{ll}
	\bot \limp \blacksquare \bot & \text{by $\IPL$ reasoning} \\
	\Box (\bot \limp \blacksquare \bot) & \text{by $\necw$} \\
	(\Box (\bot \limp \blacksquare \bot) \limp \bot )\limp \bot & \text{by $\IPL$ reasoning} \\
	\Diamond \bot \limp \bot & \text{by $\CA$ and definition of $\Diamond$}
	\\
	\noalign{\medskip}
	\Diamond A_i \limp \Diamond A_0 \lor \Diamond A_1 & \text{for $i=0,1$, by $\IPL$ reasoning} \\
	\blacksquare \Diamond A_i \limp \blacksquare (\Diamond A_0 \lor \Diamond A_1) & \text{by $\necb$ and $\functbb$} \\
	A_i  \limp \blacksquare (\Diamond A_0 \lor \Diamond A_1) & \text{by $\bbwd$} \\
	A\lor B \limp \blacksquare (\Diamond A_0 \lor \Diamond A_1) & \text{by $\IPL$ reasoning} \\
	\Box (A\lor B \limp \blacksquare (\Diamond A_0 \lor \Diamond A_1)) & \text{by $\necw$} \\
	(\Box (A\lor B \limp \blacksquare (\Diamond A_0 \lor \Diamond A_1)) \limp \Diamond A \lor \Diamond B) \limp \Diamond A \lor \Diamond B & \text{by $\IPL$ reasoning} \\
	\Diamond (A\lor B) \limp \Diamond A \lor \Diamond B & \text{by $\CA$, setting $X=\Diamond A \lor \Diamond B$, and definition of $\Diamond$}
	\\
	\noalign{\medskip}
	(\blacksquare \Diamond A \limp \blacklozenge \Box B) \limp A \limp B & \text{by $\bbwd$, $\bdwb$ and $\IPL$ reasoning} \\
	\blacklozenge (\Diamond A \limp \Box B) \limp A \limp B & \text{by $\functbb$, $\functbd$ and $\necb$} \\
	\Box \blacklozenge (\Diamond A \limp \Box B) \limp \Box (A\limp B) & \text{by $\necw$ and $\functwb$} \\
	(\Diamond A \limp \Box B) \limp \Box (A\limp B) & \text{by $\wbbd$}
\end{array}
\]
To explain a little the step justified `by $\functbb$, $\functbd$ and $\necb$', note that $\blacklozenge (C\limp D) \limp \blacksquare C \limp \blacklozenge D$ is indeed readily derivable using those axioms and rule.

So $\IKtSO$ contains the intuitionistic modal logic $\IK$ of Fischer Servi and Simpson \cite{FischerServi1977,simpson1994proof}.
By symmetry, it also contains all the black versions of the principles in \cref{eq:other-IK-axioms} too.
Altogether, we now have that $\IKtSO$ furthermore contains Ewald's $\IKt$ \cite{ewald1986intuitionistic}, justifying our chosen nomenclature.\footnote{Ewald includes a couple other axioms too, namely $\Box (A\land B) \liff \Box A \land \Box B$ and $\Diamond (A\limp B) \limp \Box A \limp \Diamond B$ (and their black analogues). Both are routinely derivable in even $\CK$ (and its black analogue, respectively).}
	
Let us point out that the derivation of $\IKt$ in $\IKtSO$ does not really rely on the availability of second-order reasoning. 
A closer inspection of the arguments reveals that we need only the following (first-order) principles:
\begin{itemize}
	\item $\Diamond A \limp \Box (A \limp \blacksquare C) \limp C$
	\item $\blacklozenge A \limp \blacksquare (A \limp \Box C) \limp C$
	\item $C \limp \Box (\blacklozenge C \limp A) \limp \Box A$
	\item $C \limp \blacksquare (\Diamond C \limp A) \limp \blacksquare A$
\end{itemize}
These are derivable already from the modal axioms and rules, \cref{modal-axioms-rules-wb}, and the tense axioms, \cref{tense-axioms}, under $\IPL$ similarly to the proof of \cref{dia-iff-so-dfn}.
That such a minimal axiomatisation already generates all of $\IKt$ seems to be a folklore fact in the (intuitionistic) tense community, e.g.\ as stated in \cite[Remark~2.3]{lin2022negneg-tense}. 
We have provided the derivations above nonetheless as we have not been able to find them explicitly in the literature.

\begin{proposition} \label{prop:IK-thm}
	$\IKt$ proves the following:
	\begin{enumerate}
		\item $(\Box A \wedge \Dmnd B) \to \Dmnd (A \wedge B)$.
		\item $(\tBox A \wedge \tDmnd B) \to \tDmnd (A \wedge B)$
		\item $(\DIA A \IMP \BOX B) \IMP \BOX (A \IMP B)$
		\item $(\DIAB A \IMP \BOXB B) \IMP \BOXB (A \IMP B)$
		\item $\DIA (A \IMP B) \IMP (\DIA A \IMP \BOX B)$
		\item $\DIAB (A \IMP B) \IMP (\DIAB A \IMP \BOXB B)$
	\end{enumerate}
\end{proposition}

\subsection{Classical theory}
Finally, let us conclude this section by giving a classical version of second-order tense logic. 
This is obtained, as expected, by adding double negation elimination:\footnote{We could have made other equivalent choices, e.g.\ by adding Peirce's law $((A\limp B)\limp A) \limp A$.}

\begin{definition}
	$\KtSO := \IKtSO + \lnot \lnot A \limp A$. 
\end{definition}

We shall address the relationship between our intuitionistic and classical theories later in \cref{sec:further}.

\section{Model theory: (bi)relational structures}\label{sec:semantics}

In this section we introduce a standard form of semantics for intuitionistic modal logics, exhibiting two relations: the partial order from Kripke structures for intuitionistic logic, and an `accessibility' relation from relational structures for modal and tense logic.
The most interesting aspect herein is the accommodation of \emph{sets} (or \emph{predicates}), i.e.\ the domain over which variables $X,Y$ etc.\ vary.

We also consider a special subset of birelational structures, what we call \emph{predicate} structures, and their specialisation to classical models.

\subsection{Two-sorted birelational semantics under full comprehension}

The birelational semantics of (first-order) intuitionistic tense formulas will be defined as usual, following the intuitionistic modal and tense traditions \cite{FischerServi1977,ewald1986intuitionistic,PlotkinStirling1986,simpson1994proof}.
To account for second-order quantifiers, we must further include a domain of sets over which variables vary:

\begin{definition}
[Birelational structures]
\label{dfn:birel-structures}
A \textbf{(two-sorted) (birelational) structure} $\B$ includes the following data:
\begin{itemize}
    \item A set $\Worlds$ of \textbf{worlds} of $\B$.
    \item A partial order $\leq $ on $\Worlds$.
    \item A class $\setW \subseteq \pow \Worlds$ of \textbf{sets} (or \textbf{predicates}) that are upwards-closed, i.e.\ if $V \in \setW$ and $v \leq v' $ then $ Vv \implies Vv'$.
    \item An interpretation $\interp \B {} P \in \setW$ for each $P \in \Prop$.
    \item An \textbf{accessibility relation} $\interp \B {} R \subseteq \Worlds \times \Worlds$.
\end{itemize}
We furthermore require in $\B$ that 
$\interp \B {} R$ is a \emph{bisimulation} on $\leq$, i.e.:
    \begin{itemize}
        \item $\forall v,w,w' \in W \,  (v\interpM \B Rw\leq w' \implies \exists v' \geq v\, v'\interpM \B Rw')$.
        \item $\forall v,v',w \in W\,  (v'\geq v \interpM \B R w \implies \exists w' \geq w\, v'\interpM \B Rw'$.
    \end{itemize}

Now, let us temporarily expand the language of formulas by including each $V \in \setW$ as a propositional symbol, setting $\interpM \B V = V$.
The judgement $v \modelsM \B A $, for $v\in W$, is defined by induction on the size of a closed formula $A$:
\begin{itemize}
    \item $v\modelsM \B P$ if $ \interpM\B P v $.\footnote{Note that this clause accounts for the new propositional symbols $V \in \setW$ too.}
    \item $v \modelsM \B  A \limp B$ if, whenever $v \leq v'$ and $v' \modelsM\B  A$, we have $v' \modelsM \B  B$.
    \item $v\modelsM \B \Box A$ if, whenever $v \leq v' $ and $ v \interpM\B R w'$, we have $w' \modelsM \B A$.
    \item $v \modelsM \B \blacksquare A$ if, whenever $v \leq v' $ and $u' \interpM \B R v' $, we have $u' \modelsM \B A$.
    \item $v \modelsM \B \forall X A$ if, whenever $v\leq v'$ and $V \in \setW$, we have $v' \modelsM \B A[V/X]$.
\end{itemize}
We write simply $\modelsM \B A$ if $w \modelsM \B A$ for every $w \in \Worlds$.

We say that $\M$ is \textbf{comprehensive} if, for each closed formula $C$ (of the expanded language), it has a set $\extension C  = \{w \in W \ | \ w \modelsM \B C\}$ in $\setW$.
A \textbf{birelational model} is a comprehensive birelational structure.
\end{definition}

Let us point out that `comprehensive structures' have several alternative names in the literature, including `full structures', `complete structures', `principal structures' or even just `structures' (where `pre-structures' are not necessarily comprehensive). 
We prefer the present terminology as it is less ambiguous (e.g.\ full semantics, complete axiomatisation,...) and is suggestive of the role this property plays in modelling the comprehension axiom, $\CA$.

\begin{remark}
    [Full vs Henkin]
    A naive domain of sets is simply the full powerset $\pow W$. 
    Such a structure is comprehensive by default, since it has \emph{every} possible set, it includes in particular the extensions $\extension C$.
    This is often referred to as the \emph{full} or \emph{standard} semantics of second-order logic.\footnote{Traditionally the  nomenclature `second-order' would be reserved for only such semantics, while our framework is perhaps more correctly dubbed `two-sorted first-order'. We shall refrain from rehashing this discussion here but refer the reader to REF for a comprehensive explanation. The current terminology `second-order' has become standard in computational logic. }
    However such a semantics for second-order logic admits no complete proof systems, as its validities are not even analytical, let alone recursively enumerable.
    It is more typical in proof theoretic investigations to admit the \emph{Henkin} structures that we have presented here, treating `second-order' as simply another sort.
    A useful discussion of this distinction and source of further references is available in \cite{sep-logic-second-higher-order}, in particular Sections 5 and 9.
    
    On the other hand, we cannot admit arbitrary domains of sets if $\IKtSO$ is to be sound for our models. The condition of comprehensivity is required to ensure that structures model comprehension, $\CA$.  
    Note the awkwardness here: the class of sets $\W$ must be specified outright, but whether it is comprehensive or not depends on the resulting notion of entailment. 
    For this reason typical defininitions of comprehensive structures are \emph{impredicative}.
\end{remark}

Henceforth let us reserve the metavariable $\B$ to vary over birelational models.
One of the main results of this work is that our axiomatic system $\IKtSO$ and our birelational semantics above actually induce the same logic:

\begin{maintheorem}
    [Soundness and completeness]
    \label{mainthm:soundness-completeness-intuitionistic}
    $\IKtSO \proves A \iff \forall \B \, \modelsM \B A$.
\end{maintheorem}

The proof of the $\impliedby$ direction, completeness,  will be broken up into several steps, in fact factoring through a further \emph{proof theoretic} presentation of the logic from \cref{sec:prooftheory}.
We shall turn to this soon, but for the remainder of this subsection let us establish the $\implies$ direction, soundness.

First we need a standard intermediate result:

\begin{lemma}
    [Monotonicity]
    \label{monotonicity}
    If $v\leq v'$ and $v \models_\B A $ then $ v' \models_\B A$.
\end{lemma}
\begin{proof}
     By induction on the structure of $A$. Assume $v \leq v^\prime$ and $v \vDash_\B A$.

\begin{itemize}
    \item  $A = P \in \Prop$: $P_\B v$ for $P_\B \in \W$. Thus, as $P_\B$ is upwards closed wrt.~$\leq$, $P_\B v^\prime$ and $v^\prime \vDash_\B P$.
\item
    $A = B \to C$: Consider arbitrary $v^{\prime \prime}$ with $v^\prime \leq v^{\prime \prime}$ and $v^{\prime \prime} \vDash_\B B$. By transitivity of $\leq$, $v \leq v^{\prime \prime}$ and therefore $v^{\prime \prime} \vDash C$. By arbitrariness of $v^{\prime \prime}$: $v^\prime \vDash_\B B \to C$.
\item
    $A= \Box B$: Consider arbitrary $\vPP$ and $\wP$ with $\vP \leq \vPP$ and $\vPP R_\B \wP$. Again by transitivity of $\leq$: $v \leq \vPP$, thus $\wP \vDash_\B B$ and $\vP \vDash_\B \Box B$.
\item
    $A= \tBox B$: Consider arbitrary $\vPP$ and $\wP$ with $\vP \leq \vPP$ and $\wP R_\B \vPP$. As before, $v \leq \vPP$, thus $\wP \vDash_\B B$ and $\vP \vDash_\B \tBox B$.
\item
    $A= \forall X B$: Consider arbitrary $\vPP$ with $\vP \leq \vPP$ and some $V \in \W$. As before, we have $v \leq \vPP$ and therefore $\vPP \vDash_\B B[V/X]$. This gives us $\vP \vDash_\B \forall X B$. \qedhere
\end{itemize}
   
\end{proof}

\begin{theorem}
    [Soundness]\label{prop:brelsoundness}
    $\IKtSO \proves A \implies \forall \B\, \modelsM\B A$.
\end{theorem}
\begin{proof}
    We proceed by induction on $\IKtSO \proves A$.

The axioms and rules of $\IPL2$, namely are standard, following from soundness of $\IPL 2 $ for comprehensive intuitionistic structures (see, e.g., \cite[Section~11.1]{sorensen2006lectures}). In particular notice that their verification, being modality-free, does not involve the accessibility relation $R_\B$.

The modal axioms $\functwb$, $\functbb$ and rules $\necw,\necb$ are also standard,
        following from the soundness of $\IK$ (equivalently its black variant) for birelational structures \cite[Theorem~4]{fischerservi1984axiomatizations}. 
        In particular notice that their verification, being quantifier-free, does not involve the class $\set W$ of sets.

It remains to verify the axioms involving $\Diamond, \blacklozenge$, as they are coded by second-order formulas.

    \begin{itemize}

        \item $\BOX(A \IMP B) \IMP (\DIA A \IMP \DIA B)$.
        Let $w_1 \irelop w$ with $\brelsat{w_1}{\BOX(A \IMP B)}$.
        Let $w_2 \irelop w_1$ with $\brelsat{w_2}{\DIA A = \fall (\BOX(A \IMP \BOXB X) \IMP X)}$ with $X$ fresh for $A$. 
        We show $\brelsat{w_2}{\DIA B = \fall (\BOX(B \IMP \BOXB X) \IMP X)}$ for $X$ fresh for $B$.
        Let $w_3 \irelop w_2$ and $V \in \class$.
        We need to show $\brelsat{w_3}{\BOX(B \IMP \BOXB V) \IMP V}$.
        Let $w_4 \irelop w_3$ with $\brelsat{w_4}{\BOX(B \IMP \BOXB V)}$ and so we are left to show that $\brelsat{w_4}{V}$.
        For $w_4 \irel w_5 \brel v_5$, $\brelsat{v_5}{B \IMP \BOXB V}$.
        Now, as $\brelsat{w_1}{\BOX(A \IMP B)}$ and $w_1 \irel w_5 \brel v_5$, $\brelsat{v_5}{A \IMP B}$.
        Through standard reasoning, we have $\brelsat{v_5}{A \IMP \BOXB V}$ and so $\brelsat{w_4}{\BOX(A \IMP \BOXB V)}$.
        As $\brelsat{w_2}{\DIA A = \fall (\BOX(A \IMP \BOXB X) \IMP X)}$ and $w_4 \irelop w_2$, $\brelsat{w_4}{\BOX(A \IMP \BOXB V) \IMP V}$ and so it follows that $\brelsat{w_4}{V}$.

        \item $A \IMP \BOX \DIAB A$.
        Let $w_1 \irelop w$ with $\brelsat{v_1}{A}$.
        We show for $w_2 \irelop w_1$ and $w_2 \brel v_2$ that $\brelsat{v_2}{\DIAB A = \fall (\BOXB (A \IMP \BOX X) \IMP X)}$ for $X$ fresh for $A$.
        Let $v_3 \irelop v_2$ and $V \in \class$.
        We show $\brelsat{v_3}{\BOXB (A \IMP \BOX V) \IMP V}$.
        Let $v_4 \irelop v_3$ with $\brelsat{v_4}{\BOXB (A \IMP \BOX V)}$.
        As $w_2 \brel v_2 \irel v_4$, there exists $w_4 \irelop w_2$ with $w_4 \brel v_4$.
        By Lemma~\ref{monotonicity}, $\brelsat{w_4}{A}$, and as $\brelsat{v_4}{\BOXB (A \IMP \BOX V)}$, $\brelsat{w_4}{A \IMP \BOX V}$.
        Through standard reasoning this means that $\brelsat{w_4}{\BOX V}$, and as $w_4 \brel v_4$, we must have $\brelsat{v_4}{V}$.

        \item $\DIA \BOXB A \IMP A$.
        Let $w_1 \irelop w$ with $\brelsat{w_1}{\DIA \BOXB A = \fall (\BOX(\BOXB A \IMP \BOXB X) \IMP X)}$ for $X$ fresh for $\BOXB A$.
        Then by definition as $A \in \class$, $\brelsat{w_1}{\BOX(\BOXB A \IMP \BOXB A) \IMP A}$.
        For all worlds $v$, we must have $\brelsat{v}{\BOXB A \IMP \BOXB A}$, so therefore $\brelsat{w_1}{\BOX(\BOXB A \IMP \BOXB A)}$.
        So it follows that $\brelsat{w_1}{A}$.
    \end{itemize}
    The remaining axioms $\functbd, \bbwd,\bdwb$ follow by a symmetric argument.
\end{proof}

\subsection{Two-sorted predicate semantics under full comprehension}

We shall also consider intuitionistic predicate structures, in which every world is a classical modal model. 
We introduce this additional semantics for two reasons: (i) in order to further test the robustness of the logic $\IKtSO$; and, (ii) since we shall use these for our ultimate countermodel construction in \cref{sec:completeness}.

\begin{definition}
    [Predicate structures]
    \label{dfn:pred-structures}
A \textbf{(two-sorted) predicate structure} $\PP$ includes the following data:

\begin{itemize}
    \item A set $\States$ of \textbf{intuitionistic worlds} (or \textbf{states}).
    \item A partial order $\leq$ on $\States$.
    \smallskip
    \item A set $W  $ of \textbf{(modal) worlds}. 
    \item A set $\setW \supseteq \Prop$ of \textbf{(modal) sets} (or \textbf{(modal) predicates}).

    \smallskip

    \item An interpretation $\interpa a P \subseteq \Worlds$ for each $P \in \setW$ and $a \in \States$. We require monotonicity: $a \leq b \implies \interpa a P \subseteq \interpa  b P$.
    \item An \textbf{accessibility relation} $\interpa a R \subseteq  \Worlds \times \Worlds$ at $a$.
    We require monotonicity: $a \leq b \implies \interpa a R \subseteq \interpa b R$.
\end{itemize}

Let us temporarily expand the language of formulas by including each $P \in \setW$ as a propositional symbol.
The judgement $ a , v \modelsM \PP  A$, for $a \in \States , v \in \Worlds$, is defined by induction on the size of a formula $A$ as follows:
\begin{itemize}
    \item $  a , v \modelsM \PP P$ if $v \in \interpa a P$.
    \item $ a , v \modelsM \PP A \limp B$ if, whenever $a \leq b$ and $b, v \modelsM \PP A$, we have $b, v \modelsM \PP B$.
    \item $ a, v \modelsM \PP \Box A$ if, whenever $a \leq b$ and $v \interpa b R w$, we have $ b, w \modelsM \PP A$.
    \item $ a , v \modelsM \PP \blacksquare A$ if, whenever $a \leq b$ and $u \interpa b R v$, we have $ b , u \modelsM\PP A$.
    \item $ a , v \modelsM \PP \forall X A$ if, whenever $a \leq b$ and $P\in \setW$, we have $ b , v \modelsM\PP A[P/X]$.
\end{itemize}
We write simply $\modelsM \PP A $ if $a,v \modelsM \PP A$ for every $a \in \States, v \in \Worlds$.

We say that $\M$ is \textbf{comprehensive} if, for each formula $C $ (of the expanded language), it has a set $\extension C \in \setW$ with $\interpa a {\extension C} = \{w \in W \ | \ a,w \modelsM \PP C\}$.
A \textbf{predicate model} is a comprehensive predicate structure.
\end{definition}

Henceforth let us reserve the metavariable $\PP$ to vary over predicate models.
In fact, predicate structures can be construed as particular birelational structures. Formally:
\begin{definition}
    [Birelational collapse]
    Given a predicate structure $\PP$, as presented in \cref{dfn:pred-structures}, we define a birelational structure $\brcol\PP$ by:
    \begin{itemize}
        \item The set of worlds is $\brcoldat \PP W := \States \times W$.
        \item The partial order $\brcoldat \PP \leq$ is given by $(a,v) \brcoldat \PP \leq (b,w)$ if $a\leq b$ and $v=w$.
        \item The class of sets $\brcoldat \PP{\set W}$ includes each $\brcoldat \PP V := \{(a,v) \in \brcoldat \PP W \, | \, v \in V^a\}$, for $V \in \set W$.
        \item The interpretation of propositional symbols $\brcoldat \PP P$ is given by $\{(a,v)\in \brcoldat \PP W \, |\, v \in P^a\}$.
        \item The accessibility relation $\brcoldat \PP R$ is given by $(a,v)R_{\brcol\PP} (b,w)$ if $a=b$ and $vR^a w$.
    \end{itemize}
\end{definition}

We better show that $\brcol \PP$ satisfies the appropriate conditions of \cref{dfn:birel-structures}:
\begin{proposition}
    ${\brcol\PP}$ is a well-defined birelational structure.
\end{proposition}
\begin{proof}
First we show that each $V \in \set W_{\brcol \PP}$ is upwards closed.
Suppose $(a,v) \brcoldat \PP \leq (b,v)$, WLoG, i.e.\ $a\leq b$. We have:
\[
\begin{array}{r@{ \ \implies \ }ll}
     (a,v) \in \brcoldat \PP V & v \in V^a & \text{by definition of $\brcoldat \PP V$} \\
        & v \in V^b &  \text{by monotonicity of $V^-$} \\
        & (b,v) \in \brcoldat \PP V & \text{by definition of $\brcoldat \PP V$}
\end{array}
\]

Now let us show that $\brcoldat \PP R$ is a bisimulation on $\brcoldat \PP\leq $:
\begin{itemize}
    \item Suppose $(a,v) \brcoldat \PP R (a,w) \brcoldat \PP \leq (b,w)$, WLoG, so $vR^a w$ and $a\leq b$. Then:
    \begin{itemize}
        \item $(a,v) \brcoldat \PP \leq (b,v)$ by definition of $\brcoldat \PP \leq$; and,
        \item $vR^b w $ by monotonicity of $R^-$, so $(b,v) \brcoldat \PP R (b,w)$ by definition of $\brcoldat \PP R$.
    \end{itemize}
    So we can choose $(b,v)$ as the required world.
    \item Suppose $(b,v) \brcoldat \PP \geq (a,v) \brcoldat \PP R (a,w)$, WLoG, so $b\geq a$ and $vR^a w $. Then:
    \begin{itemize}
        \item $(b,w) \brcoldat \PP\geq (a,w)$ by definition of $\brcoldat \PP \leq$; and,
        \item $vR^b w$ by monotonicity of $R^-$, so $(b,v) \brcoldat \PP R (b,w)$ by definition of $\brcoldat \PP R$.
    \end{itemize}
    So we can choose $(b,w)$ as the required world. \qedhere
\end{itemize}
\end{proof}

We can exhibit a rather strong equivalence between the theories of $\PP$ and $\brcol \PP$:

\begin{proposition}
[Equivalence]
\label{pred-str-equiv-brcol}
Let $A(\vec X)$ be a formula, all free variables displayed, and $\vec V \in \set W$ with $|\vec V| = |\vec X|$.
    Then $ (a,v) \modelsM {\brcol \PP} A(\vec V) \iff a,v \modelsM \PP A(\brcoldat \PP {\vec V}) $.
\end{proposition}
\begin{proof}
    By induction on the structure of $A(\vec X)$:
    \begin{itemize}
    \item If $A(\vec X) = X$ and $\vec V = V$ then:
    \[
    \begin{array}{r@{ \ \iff \ }ll}
         (a,v) \modelsM {\brcol \PP} \brcoldat \PP V & v \in V^a & \text{by definition of $\modelsM{\brcol \PP}$} \\
            & a,v \modelsM\PP V & \text{by definition of $\modelsM\PP$}
    \end{array}
    \]
    
        \item $A= P$. $(a,v) \vDash_{\brcol\PP} P$ $\LLra$ $(a,v) \in P_{\brcol\PP}$ $\LLra$ $v \in P^a$ $\LLra$ $a,v \vDash_\Pred P$.

    \item $A(\vec X)= B(\vec X) \to C(\vec X)$: Assume $(a,v) \vDash_{\brcol\PP} B(\vec V) \to C(\vec V)$ and consider any $b \geq a$ with $b,v \vDash_\Pred B(\brcoldat \PP {\vec V})$. By inductive hypothesis $(b,v) \vDash_{\brcol\PP} B(\vec V)$ and by $(a,v) \leq_{\brcol\PP} (b,v)$ we have $(b,v) \vDash_{\brcol\PP} C(\vec V)$. Again by inductive hypothesis we get $b,v \vDash_\Pred C(\brcoldat \PP {\vec V})$, which shows $a,v \vDash_\Pred B(\brcoldat \PP {\vec V}) \to C(\brcoldat \PP {\vec V})$.
    For the converse assume $a,v \vDash_\Pred B(\brcoldat \PP {\vec V}) \to C(\brcoldat \PP {\vec V})$ and consider any $(b,v) \geq_{\brcol\PP} (a,v)$ with $(b,v) \vDash_{\brcol\PP} B(\vec V)$. By inductive hypothesis $b,v \vDash_\PP B(\brcoldat \PP {\vec V})$ and by $a \leq b$ we have $b,v \vDash_\Pred C(\brcoldat \PP {\vec V})$. Again by inductive hypothesis we get $(b,v) \vDash_{\brcol\PP} C(\vec V)$, which shows $(a,v) \vDash_{\brcol\PP} B(\vec V) \to C(\vec V)$.

    \item $A=\Box B(\vec X)$: $(a,v) \vDash_{\brcol\PP} \Box B(\vec V)$ $\LLra$ $(b,w) \vDash_{\brcol\PP} B(\vec V)$ for any $(a,v) \leq_{\brcol\PP} (b,v) R_{\brcol\PP} (b,w)$ $\overset{\text{I.H.}}{\LLra}$ $b,w \vDash_\Pred B(\brcoldat \PP {\vec V})$ for any $a \leq b$ and $v R^b w$ $\LLra$ $a,v \vDash_\Pred \Box B(\brcoldat \PP {\vec V})$

    \item $A= \tBox B(\vec X)$: $(a,v) \vDash_{\brcol\PP} \tBox B(\vec V)$ $\LLra$ $(b,w) \vDash_{\brcol\PP} B(\vec V)$ for any $(a,v) \leq_{\brcol\PP} (b,v)$ and $(b,w) R_{\brcol\PP} (b,v)$ $\overset{\text{I.H.}}{\LLra}$ $b,w \vDash_\Pred B(\brcoldat \PP {\vec V})$ for any $a \leq b$ and $w R^b v$ $\LLra$ $a,v \vDash_\Pred \tBox B(\brcoldat \PP {\vec V})$

    \item $A = \forall X B(\vec X, X)$ : $(a,v) \vDash_{\brcol\PP} \forall X B(\vec V, X)$ $\LLra$ $(b,v) \vDash_{\brcol\PP} B(\vec V, V)$ for any $(b,v) \geq_{\brcol\PP} (a,v)$ and $V \in \W_{\brcol\PP}$ $\overset{\text{I.H.}}{\LLra}$
    $b,v \vDash_\Pred B(\brcoldat \PP {\vec V}, \brcoldat \PP V)$ for any $a \leq b$ and $\brcoldat \PP V \in \W$ $\LLra$ $a,v \vDash_\Pred \forall X B(\brcoldat \PP {\vec V}, X)$. 
    \qedhere
    \end{itemize}
\end{proof}

We can now recover a couple further useful results from this identification:

\begin{corollary}
    [Comprehensivity]
    \label{brcol-is-comp}
    If $\PP$ is comprehensive then so is $\brcol \PP$.
\end{corollary}
\begin{proof}
    Suppose $\B$ is comprehensive, and consider the sets $\extension C$ with $v \in \extension C^a \iff a,v \modelsM \PP C$. We have as required:
    \[
    \begin{array}{r@{\ = \ }ll}
    \brcoldat \PP {\extension C} & \{ (a,v) \in \brcoldat \PP W \, | \, v \in \extension C^a\} & \text{by definition of $\brcoldat \PP -$} \\
        & \{(a,v) \in \brcoldat \PP W \, |\, a,v \modelsM\PP C \} & \text{by definition of $\extension C$} \\
        & \{ (a,v) \in \brcoldat \PP W \, | \, (a,v) \modelsM {\brcol \PP} C \} & \text{by \cref{pred-str-equiv-brcol}} \qedhere
    \end{array}
    \]
\end{proof}

Putting the previous \cref{pred-str-equiv-brcol,brcol-is-comp} together we have:

\begin{proposition}
\label{prop:breltoprel}
    $\forall \B\, \modelsM\B A \implies \forall \PP\, \modelsM\PP A$.
\end{proposition}

Finally it will be useful in our proof search argument in \cref{sec:completeness} to also inherit:

\begin{lemma}
[Monotonicity]
If $a\leq b$ and $a,v \modelsM{\PP} A$ then $b,v \modelsM {\PP} A$.
\end{lemma}
\begin{proof}
    Given $a\leq b$, we have:
    \[
    \begin{array}{r@{\ \implies \ }ll}
    a,v \modelsM \PP A & (a,v) \modelsM {\brcol\PP} A & \text{by \cref{pred-str-equiv-brcol}} \\
        & (b,v) \modelsM {\brcol \PP} A & \text{since $(a,v) \brcoldat\PP\leq (b,v)$ and \cref{monotonicity} } \\
        & b,v \modelsM\PP A & \text{by \cref{pred-str-equiv-brcol}} \qedhere
    \end{array}
    \]
\end{proof}

\subsection{Two-sorted (uni)relational semantics under full comprehension}
Finally let us adapt the semantics we have presented to the classical setting.
We can view 
classical structures as special cases of either birelational or predicate structure where $\leq$ is trivial, i.e.:
\begin{itemize}
    \item a birelational structure where all points are incomparable; or,
    \item a predicate model with only one intuitionistic world.
\end{itemize}

For the sake of completeness, let us give a self-contained definition:

\begin{definition}
    [Classical structures]
    A \textbf{(uni)relational structure} $\mathfrak R$ includes the following data:
    \begin{itemize}
        \item A set $W$ of \textbf{worlds}.
        \item A class $\set W \subseteq \pow W$ of \textbf{sets} (or \textbf{predicates}).
        \item An interpretation $P_{\mathfrak R} \in \set W$ for each $P \in \Prop$.
        \item An \textbf{accessibility relation} $R \subseteq W \times W$.
    \end{itemize}
    Now, let us temporarily expand the language of formulas by including each $V\in \setW$ as a propositional symbol, setting $\interp{}\RRR V := V$. 
    The judgement $v \modelsM \RRR A$, for $v \in W$, is defined by induction on the size of a closed formula $A$:
    \begin{itemize}
        \item $v \modelsM \RRR P$ if $\interp  \RRR {} P v$.
        \item $v \modelsM \RRR A\limp B$ if, whenever  $v \modelsM \RRR A $, we have $ v\modelsM \RRR B$.
        \item $v \modelsM \RRR \Box A$ if, whenever $vRw$, we have $w \modelsM \RRR A$.
        \item $v \modelsM \RRR \blacksquare A$ if, whenever $uRv$, we have $u \modelsM \RRR A$.
        \item $v \modelsM \RRR \forall X A$ if, whenever $V \in\setW$, we have $v\modelsM \RRR A [V/X]$.
    \end{itemize}
    We write simply $\modelsM\RRR A$ if $v\modelsM\RRR A$ for every $v\in W$.

    We say that $\RRR$ is \textbf{comprehensive} if, for each formula $C$ (of the expanded language), it has a predicate $\extension C \in \W $ with ${\extension C} = \{w \in W \ | \ w \modelsM \RRR C\}$.
    A \textbf{(uni)relational model} is a comprehensive relational structure.
\end{definition}

Henceforth, let us reserve the metavariable $\RRR$ to vary over relational models.
Our main metalogical result for classical second-order tense logic, analogous to \cref{mainthm:soundness-completeness-intuitionistic} in the intuitionistic case, is:

\begin{maintheorem}
[Soundness and Completeness]
\label{mainthm:soundness-completeness-classical}
    $\KtSO \proves A \iff \forall \RRR \, \modelsM \RRR A$.
\end{maintheorem}

Like the intuitionistic case completeness, the $\impliedby$ direction, is factored through a proof theoretic presentation of the logic from \cref{sec:prooftheory}.
Given how we have defined relational models, we can factor soundness, the $\implies$ direction, into (a) already established soundness of $\IKtSO$ for birelational structures; and (b) verification of the additional axiom $\lnot \lnot A \limp A$:

\begin{theorem}
    [Soundness]
    \label{classical-rel-soundness}
    $\KtSO \proves A \implies \forall \RRR \, \modelsM\RRR A$.
\end{theorem}
\begin{proof}
[Proof sketch]
A relational model can be expanded into a birelational model by simply setting $\leq $ to be equality on worlds.
Thus all the axioms and rules of $\IKtSO$ are already sound for relational models, by \cref{prop:brelsoundness}.
It remains to verify the double negation axiom, $\lnot \lnot A \limp A$.
Suppose $v\modelsM\RRR \lnot \lnot A$, then either $v\not\modelsM\RRR A$ or $v\modelsM \RRR \bot$. We consider each case separately:
\[
\begin{array}{r@{\ \implies \ }ll}
v\not\modelsM\RRR \lnot A & v\modelsM\RRR A \text{ and } v\not\modelsM \RRR \bot & \text{by definition of $\modelsM\RRR$ and $\lnot$} \\
    & v\modelsM \RRR A \\
\noalign{\smallskip}
v \modelsM\RRR \bot & v\modelsM\RRR \extension A & \text{since $\bot = \forall X X $ and by definition of $\modelsM \RRR$} \\
& v \in  \extension A & \text{by definition of $\modelsM\RRR$} \\
    & v\modelsM\RRR A & \text{by definition of $\extension A$} \qedhere
\end{array}
\]
\end{proof}

\section{Proof theory: labelled systems}\label{sec:prooftheory}
We shall now turn to a proof theoretic presentation of second-order tense logic.
As well as for self contained interest this will, as already mentioned, serve to factor our earlier stated axiomatic completeness results.
Unlike the previous two sections, we shall first present the classical system, before recovering the intuitionistic versions via appropriate constraints.
The only reason for this is brevity, allowing us to define all systems we consider without repeating rules.

\subsection{Labelled sequent calculi}
 \emph{Labelled deductive systems} were proposed by Gabbay~\cite{gabbay1991labelled} as a uniform proof-theoretic framework for a wide range of logics.
The idea has been applied to modal logic by using the strength of its (uni)relational semantics~\cite{russo1996modal} in order to resolve the difficulty of designing proof systems for modal logic using standard Gentzen sequents.
They reason about formulas labelled by the world in which they are evaluated, while keeping track of a `control' constraining the accessibility relation between worlds.
They were fully developed for intuitionistic modal logic by Simpson~\cite{simpson1994proof}, before being widely applied to e.g.~classical modal logic with Horn~\cite{vigano2000labelled} or coherent~\cite{negri2005proof} extensions and beyond~\cite{negri2014proof},
justification logic~\cite{ghari2017labeled}, non-normal modal logics~\cite{dalmonte2018non}, conditional logics~\cite{girlando2021uniform}, 
first-order modal logic~\cite{orlandelli2021labelled}, etc.

Let us now fix a set $\war$ of \textbf{world symbols}, written $u,v,w$ etc.
A \textbf{relational atom} is an expression $v\R w$, where $v,w \in \war$.
A \textbf{labelled formula} is an expression $v : A$ where $v\in \war$ and $A$ is a formula.
Write $\lFmla$ for the set of labelled formulas.

A \textbf{(labelled) sequent} is an expression $\rels R \stoup \Gamma \seqar \Delta$ where $\rels R$ is a set of relational atoms, called the \textbf{relational context}, and $\Gamma$ and $ \Delta $ are multisets of labelled formulas, called the \textbf{(LHS) cedent} and \textbf{(RHS) cedent}, respectively.
`$\stoup$' and `$\seqar$' here are just syntactic delimiters.
Informally, we may read sequents as ``if all of the LHS holds, then some of the RHS is true''. 
This intuition is developed more formally later in \cref{sec:labelsoundness}.

\begin{figure}
\emph{Identity and cut:}
\medskip
\[
\begin{array}{cc}
\vlinf{\id}{}{\rels R \stoup  v : A \seqar v : A}{}
& 
\vliinf{\cut}{}{\rels R \stoup \Gamma, \Gamma' \seqar \Delta, \Delta'}{\rels R \stoup \Gamma \seqar \Delta, v : A}{\rels R \stoup \Gamma' , v : A \seqar \Delta'}
\end{array}
\]

\medskip
{\emph{Strutural rules:}}
\medskip
\[
\begin{array}{cc}
\vlinf{\lr\wk}{}{\rels R \stoup \Gamma , v : A \seqar \Delta}{\rels R \stoup \Gamma \seqar \Delta}
& 
\vlinf{\rr \wk}{}{\rels R \stoup \Gamma \seqar \Delta, v:A}{\rels R \stoup \Gamma \seqar \Delta}
\\
\noalign{\medskip}
\vlinf{\lr\cntr}{}{\rels R \stoup \Gamma, v : A \seqar \Delta}{\rels R \stoup \Gamma, v : A, v : A \seqar \Delta}
&
\vlinf{\rr \cntr}{}{\rels R \stoup \Gamma \seqar \Delta, v:A}{\rels R \stoup \Gamma \seqar \Delta, v:A , v:A}
\end{array}
\]

\medskip
\emph{Logical rules:}
\medskip
\[
\begin{array}{cc}
\vliinf{\lr \limp}{}{\rels R \stoup \Gamma, \Gamma', v:A \limp B \seqar \Delta, \Delta'}{\rels R \stoup \Gamma \seqar \Delta, v : A}{\rels R \stoup \Gamma' , v :B \seqar \Delta'}
& 
\vlinf{\rr \limp}{}{\rels R \stoup \Gamma \seqar \Delta, v : A \limp B}{\rels R \stoup \Gamma, v:A \seqar \Delta, v:B}
\\
\noalign{\medskip}
\vlinf{\lr \forall}{}{\rels R \stoup \Gamma, v : \forall X A\seqar \Delta}{\rels R \stoup \Gamma, v : A [C/X] \seqar \Delta}
&
\vlinf{\rr \forall}{\text{\footnotesize $P$ fresh}}{\rels R\stoup \Gamma \seqar \Delta, v: \forall X A }{\rels R\stoup \Gamma \seqar \Delta, v : A[P/X]}
\\
\noalign{\medskip}
\vlinf{\lr \Box}{}{\rels R , v \R w \stoup \Gamma, v : \Box A \seqar \Delta }{\rels R , v \R w \stoup \Gamma, w : A \seqar \Delta}
&
\vlinf{\rr \Box}{\text{\footnotesize $w$ fresh}}{\rels R \stoup \Gamma \seqar \Delta , v : \Box A}{\rels R , v \R w \stoup \Gamma \seqar \Delta, w : A }
\\
\noalign{\medskip}
\vlinf{\lr \blacksquare}{}{\rels R , u \R v \stoup \Gamma , v : \blacksquare A \seqar \Delta }{\rels R , u \R v  \stoup \Gamma, u : A\seqar \Delta }
&
\vlinf{\rr \blacksquare}{\text{\footnotesize $u$ fresh}}{\rels R \stoup \Gamma \seqar \Delta, v : \blacksquare A}{\rels R , u \R v \stoup \Gamma \seqar \Delta, u : A}
\end{array}
\]
    \caption{Rules of the labelled system $\labKtSO$ (with cut). Here a symbol is \emph{fresh} if it does not occur in the lower sequent.}
    \label{fig:lab-sys-KtSO}
\end{figure}

\begin{definition}
[Sequent calculi]
    The system $\labKtSO$ is given by the rules in \cref{fig:lab-sys-KtSO}.
    We also define two \emph{intuitionistic} restrictions of this system:
    \begin{itemize}
        \item $\labIKtSO$ is the restriction of $\labKtSO$ to sequents with singleton RHS. (In particular, there can be no right structural steps, $\rr \wk$ and $\rr \cntr$).
        \item $\mlIKtSO$ is the restriction of $\labKtSO$ where each right logical step has singleton RHS in its premiss (i.e.\ $\Delta = \emptyset$).
    \end{itemize}
    
The lower sequent of any inference step is the \textbf{conclusion}, and any upper sequents are \textbf{premisses}.
\textbf{Proofs} and \textbf{derivations} in a system are defined as usual.
We write $\mathsf L \proves \rels R \stoup \Gamma \seqar \Delta$ if the calculus $\mathsf L$ proves the $\rels R \stoup \Gamma \seqar \Delta$. 
We write simply $\mathsf L \proves A$ if $\mathsf L \proves \cdot \stoup \cdot \seqar v:A$.
\end{definition}

$\labKtSO$ is nothing more than the extension of the labelled calculus for tense logic~\cite{boretti2010finitization,ciabattoni2021display} by the usual rules for second-order quantifiers (see, e.g., \cite[Section~3.A.1]{Girard1987:pt-log-comp}, \cite[Section~5.1]{sep-proof-theory} or \cite[Definition~15.3]{Takeuti1987:pt-book}).
The singleton RHS restriction defining $\labIKtSO$ is standard for intuitionistic sequent calculi, with calculi for intuitionistic modal and tense logics obtained in the same way~\cite{simpson1994proof,strassburger2013cut,lyon_nested_2025}.
Finally $\mlIKtSO$ is a somewhat intermediate calculus, in the spirit of Maehara \cite{maehara1954,kuznets2019maehara}.
The reason we introduce it is that it is necessary for our proof search argument later, in \cref{sec:completeness}.
As we shall soon see, over the language we consider, there is no material difference between our two intuitionistic systems, cf.~\cref{prop:multitosingle}.

First let us state our main proof theoretic results:
\begin{maintheorem}
    [Hauptsatz]
    \label{mainthm:hauptsatz}
    We have the following:
    \begin{enumerate}
    \item\label{hauptsatz-classical} $\labKtSO \proves A \implies \labKtSO \setminus \cut \proves A$.
        \item\label{hauptsatz-intuitionistic} $\labIKtSO \proves A \implies \labIKtSO \setminus \cut \proves A$.
    \end{enumerate}
\end{maintheorem}

The admissibility of cut is a key desideratum in sequent calculus proof theory.
In particular it renders the system more amenable to \emph{proof search}, reducing non-determinism therein.
As for our soundness and completness results, \cref{mainthm:soundness-completeness-intuitionistic,mainthm:soundness-completeness-classical}, this result is obtained by the `grand tours' of \cref{fig:class-tour,fig:int-tour}.

\subsection{Interlude: false positives}

    Usual labelled calculi for $\IK$, e.g.\ from \cite[Section~7.2]{simpson1994proof}, include the rule,
    \begin{equation}
        \label{eq:native-bot-rule}
    \vlinf{}{}{\rels R \, |\, \Gamma , v: \bot \seqar w:A }{}
    \end{equation}
    where $\rels R$ (as well as $\Gamma $) may be \emph{arbitrary}. 
    Under our second-order definition of falsity, $\bot := \forall X X$, the above is derivable in $\labIKtSO$ when $v $ and $w$ are connected by some (undirected) path in $\rels R$.
    Formally, let us say that $v$ and $w$ are \textbf{connected} in $\rels R$ if there is a sequence $v = v_0, \dots, v_{n} = w$ where, for each $i<n$, either $v_i R v_{i+1} \in \rels R$ or $v_{i+1} R v_i \in \rels R$.
    We have:

\begin{lemma}\label{lem:simbotrule}
    $\labIKtSO\proves \seq{\rels R}{\Gamma, \fm{v}{\BOT}}{\fm{w}{A}}$ whenever $v$ and $w $ are connected in $\rels R$.
\end{lemma}
\begin{proof}
We proceed by induction on the length of a path $v=v_0,\dots, v_n=w$ connecting $v $ and $w$ in $\rels R$.
\begin{itemize}
    \item If $n=0$, so $v=w$, we have:
    $$
        \vlderivation
        {
            \vlin{\lr \forall}{}
            {
                \seq{\rels R}{\Gamma, \fm{w}{ \fall X}}{\fm{w}{A}}
            }
            {
                \vlin{\id}{}
                {
                    \seq{\rels R}{\Gamma, \fm{w}{A}}{\fm{w}{A}}
                }
                {
                    \vlhy{}
                }
            }
        }
    $$
    \item Otherwise either $vRv_1$ or $v_1 R v$. We handle the two cases respectively by,
    \[
    \vlderivation{
    \vlin{\lr \forall}{}{\rels R, vRv_1 \stoup \Gamma , v:\forall X X \seqar w:A}{
    \vlin{\lr\Box}{}{\rels R, vRv_1 \stoup \Gamma, v:\Box \forall X X \seqar w:A}{
    \vliq{\IH}{}{\rels R , vRv_1 \stoup \Gamma, v_1 : \forall X X \seqar w:A}{\vlhy{}}
    }
    }
    }
    \qquad
     \vlderivation{
    \vlin{\lr \forall}{}{\rels R, v_1Rv \stoup \Gamma , v:\forall X X \seqar w:A}{
    \vlin{\lr\Box}{}{\rels R, v_1Rv \stoup \Gamma, v:\blacksquare \forall X X \seqar w:A}{
    \vliq{\IH}{}{\rels R ,v_1Rv \stoup \Gamma, v_1 : \forall X X \seqar w:A}{\vlhy{}}
    }
    }
    }
    \]
    where derivations marked $\IH$ are obtained by the inductive hypothesis. \qedhere
\end{itemize}

\end{proof}

On the other hand, clearly the instance $\cdot \stoup v: \forall X X \seqar w:P$ of \cref{eq:native-bot-rule} has no cut-free proof in $\labIKtSO$ (and so no proof at all, cf.~\cref{mainthm:hauptsatz}).
As it happens the generality of \cref{eq:native-bot-rule} turns out to be inconsequential: in any $\labKtSO$ proof of a formula, or even a labelled sequent with connected relational context, all relational contexts appearing in the proof remain connected, by inspection of the rules of \cref{fig:lab-sys-KtSO}.
This is exemplary of a more general phenomenon in second-order logic with a negative basis: while we may be able to correctly define the positive connectives, we do not necessarily recover all of their proof theoretic behaviour (see, e.g., \cite[Section~6.2]{TroelstraSchwichtenberg2000} for a pertinent related discussion).\footnote{Note that our encoding of falsity captures rather the \emph{additive} false (which is positive), rather than multiplicative, in the sense of linear logic (see, e.g., \cite[Section~9.5]{miller2025}). It is probably more accurate to call it `$0$' instead accordingly, but this seems to be uncommon notation in the literature on (second-order) intuitionistic logic. }

For a more topical example, let us now turn to the diamond modalities.
We can simulate the usual labelled $\rr \Diamond$ rule for $\IK$ (see \cite[Section~7.2]{simpson1994proof}) using the encoding of \cref{eq:diamond-definitions} as follows:
\begin{equation}
    \label{eq:derivation-of-dia-right}
    \vlderivation{
\vlin{\rr \forall, \rr \limp}{}{\rels R, v\R w \stoup \Gamma \seqar v : \forall X (\Box (A \limp \blacksquare X) \limp X)}{
\vlin{\lr \Box}{}{\rels R, v\R w \stoup \Gamma, v : \Box (A \limp \blacksquare X) \seqar v: X}{
\vliin{\lr\limp}{}{\rels R, v\R w \stoup \Gamma , w : A \limp \blacksquare X \seqar v : X}{
    \vlhy{\rels R \stoup \Gamma \seqar w : A}
}{
    \vlin{\lr \blacksquare}{}{v\R w \stoup w : \blacksquare A \seqar v : A}{
    \vlin{\id}{}{\ \stoup v: A \seqar v: A}{\vlhy{}}
    }
}
}
}
}
\end{equation}
Similarly for $\blacklozenge$.
However, notwithstanding \cref{dia-iff-so-dfn}, it does not seem possible to carry out such a local interpretation of the left rules:
\begin{equation}
    \label{eq:dia-left-and-bdia-left}
    \vlinf{\lr \Diamond}{\text{$v$ fresh}}{\rels R \stoup \Gamma, u:\Diamond A \seqar w:B}{\rels R, uRv \stoup \Gamma , v:A \seqar w:B}
\qquad
\vlinf{\lr \blacklozenge}{\text{$u$ fresh}}{\rels R \stoup \Gamma, v:\blacklozenge A \seqar w:B}{\rels R, uRv \stoup \Gamma , u:A \seqar w:B}
\end{equation}

\subsection{Relating the two intuitionistic calculi}
We mentioned earlier that there is no material difference between our two intuitionistic calculi, in terms of the logic they define.
Let us now make this formal:

\begin{lemma}
\label{lem:multi-to-single}
    If $\mlIKtSO (\setminus \cut) \proves \rels R \stoup \Gamma \seqar \Delta$ then there is some $w:A \in \Delta $ s.t.~$\labIKtSO (\setminus \cut) \proves \rels R \stoup \Gamma \seqar w:A$ (respectively).
\end{lemma}

\begin{proof}
    By induction on the proof of $\rels R \stoup \Gamma \seqar \Delta$ in $\mlIKtSO$ and a case analysis on the last rule applied in it:
    \begin{itemize}
        \item $\id$ and right logical rules: 
        immediate as they have only one formula on the RHS.

        \item Right weakening: 
        $\vlinf{\rr \wk}{}{\rels R \stoup \Gamma \seqar \Delta, v:A}{\rels R \stoup \Gamma \seqar \Delta}$ 

        \noindent
        By inductive hypothesis, there is a $w:C$ in $\Delta$ such that $\labIKtSO \proves \rels R \stoup \Gamma \seqar w:C$; which is all we need.

        \item Right contraction: 
        $\vlinf{\rr \cntr}{}{\rels R \stoup \Gamma \seqar \Delta, v:A}{\rels R \stoup \Gamma \seqar \Delta, v:A, v:A}$

        \noindent
        By inductive hypothesis, there either is a $w:C$ in $\Delta$ such that $\labIKtSO \proves \rels R \stoup \Gamma \seqar w:C$ or $\labIKtSO \proves \rels R \stoup \Gamma \seqar v:A$; which is all we need.

        \item Non-branching left rules: 
        $\vlinf{}{}{\rels R \stoup \Gamma \seqar \Delta}{\rels R \stoup \Gamma' \seqar \Delta}$

        \noindent
        By inductive hypothesis there is $w:A\in\Delta$ s.t.~$\labIKtSO \proves \rels R \stoup \Gamma' \seqar w:A$, by applying the same rule, 
        $\labIKtSO \proves \rels R \stoup \Gamma \seqar w:A$.

        \item Implication left rule: $\vliinf{\lr \limp}{}{\rels R \stoup \Gamma, \Gamma', v:A \limp B \seqar \Delta, \Delta'}{\rels R \stoup \Gamma \seqar \Delta, v : A}{\rels R \stoup \Gamma' , v :B \seqar \Delta'}$

        \noindent
        By inductive hypothesis, either there is $w:C\in\Delta$ s.t.~$\labIKtSO \proves \rels R \stoup \Gamma \seqar w:C$, in which case, by weakening $\labIKtSO \proves \rels R \stoup \Gamma, \Gamma', v:A\limp B \seqar w:C$, 
        or $\labIKtSO \proves \rels R \stoup \Gamma \seqar v:A$ and there is $w:C\in\Delta'$ s.t.~$\labIKtSO \proves \rels R \stoup \Gamma',v:B \seqar w:C$, in which case, by $\lr\limp$, $\labIKtSO \proves \rels R \stoup \Gamma, \Gamma', v:A\limp B \seqar w:C$.

        \item $\cut$:
        $\vliinf{\cut}{}{\rels R \stoup \Gamma, \Gamma' \seqar \Delta, \Delta'}{\rels R \stoup \Gamma \seqar \Delta, v : A}{\rels R \stoup \Gamma' , v : A \seqar \Delta'}$

        \noindent
        By inductive hypothesis, either there is $w:C\in\Delta$ s.t.~$\labIKtSO \proves \rels R \stoup \Gamma \seqar w:C$, in which case, by weakening $\labIKtSO \proves \rels R \stoup \Gamma, \Gamma' \seqar w:C$, 
        or $\labIKtSO \proves \rels R \stoup \Gamma \seqar v:A$ and there is $w:C\in\Delta'$ s.t.~$\labIKtSO \proves \rels R \stoup \Gamma'\seqar w:C$, in which case, by $\cut$, $\labIKtSO \proves \rels R \stoup \Gamma, \Gamma' \seqar w:C$. \qedhere
    \end{itemize}
\end{proof}

From here we have immediately: 
\begin{proposition}
\label{prop:multitosingle}
    $\mlIKtSO (\setminus \cut) \proves A \implies \labIKtSO (\setminus \cut) \proves A$ (respectively).
\end{proposition}

 \begin{remark}
 \label{rem:multi-to-single-surprising}
     The result above may seem surprising at first glance, as it does not typically hold for even (first-order) intuitionistic propositional logic (without modalities), in the presence of positive connectives. 
     For instance there is a multi-succedent intuitionistic proof of $A_0\lor A_1 \seqar A_0, A_1$, but clearly no $A_0 \lor A_1 \seqar A_i$ is provable. 
     We avoid this issue as we do not work with native positive connectives, including $\lor$, only their impredicative encodings from \cref{eq:so-dfns-of-nonmodal-connectives}.
    %
     Similarly note, in \cref{lem:multi-to-single} above, there is no requirement for $\Delta$ to be nonempty: it is simply not possible to prove sequents with empty RHS in $\mlIKtSO$, by analysis of its rules.
 \end{remark}

\subsection{Further examples}

We conclude this section with a few more examples of labelled proofs.
First, recalling the axiomatic derivation in \cref{sec:underlying-modal-tense-logics}, here is a labelled proof of $\diaimpbox$:

$$
    \vlderivation
    {
    \vlid{=}{}{\cdot \stoup \cdot \seqar w: (\Diamond A \limp \Box B) \limp \Box (A\limp B)}{
        \vlin{\rr \IMP}{}
        {
            \seq{\cdot}{\cdot}{\fm{w}{(\fall (\BOX(A \IMP \BOXB X) \IMP X) \IMP \BOX B) \IMP \BOX (A \IMP B)}}
        }
        {
            \vlin{\rr \BOX}{}
            {
                \seq{\cdot}{\fm{w}{\fall (\BOX(A \IMP \BOXB X) \IMP X) \IMP \BOX B}}{\fm{w}{\BOX (A \IMP B)}}
            }
            {
                \vlin{\rr \IMP}{}
                {
                    \seq{wRv}{\fm{w}{\fall (\BOX(A \IMP \BOXB X) \IMP X) \IMP \BOX B}}{\fm{v}{A \IMP B}}
                }
                {
                    \vliin{\lr \IMP}{}
                    {
                        \seq{wRv}{\fm{w}{\fall (\BOX(A \IMP \BOXB X) \IMP X) \IMP \BOX B}, \fm{v}{A}}{\fm{v}{B}}
                    }
                    {
                        \vlin{\rr \forall}{}
                        {
                            \seq{wRv}{\fm{v}{A}}{\fm{w}{\fall (\BOX(A \IMP \BOXB X) \IMP X)}}
                        }
                        {
                            \vlin{\rr \IMP}{}
                            {
                                \seq{wRv}{\fm{v}{A}}{\fm{w}{\BOX(A \IMP \BOXB P) \IMP P}}
                            }
                            {
                                \vlin{\lr \BOX}{}
                                {
                                    \seq{wRv}{\fm{v}{A}, \fm{w}{\BOX(A \IMP \BOXB P)}}{\fm{w}{P}}
                                }
                                {
                                    \vliin{\lr \IMP}{}
                                    {
                                        \seq{wRv}{\fm{v}{A}, \fm{v}{A \IMP \BOXB P}}{\fm{w}{P}}
                                    }
                                    {
                                        \vlin{\id}{}
                                        {
                                            \seq{wRv}{\fm{v}{A}}{\fm{v}{A}}
                                        }
                                        {
                                            \vlhy{}
                                        }
                                    }
                                    {
                                        \vlin{\lr \BOXB}{}
                                        {
                                            \seq{wRv}{\fm{v}{\BOXB P}}{\fm{w}{P}}
                                        }
                                        {
                                            \vlin{\id}{}
                                            {
                                                \seq{wRv}{\fm{w}{P}}{\fm{w}{P}}
                                            }
                                            {
                                                \vlhy{}
                                            }
                                        }
                                    }
                                }
                            }
                        }
                    }
                    {
                        \vlin{\lr \BOX}{}
                        {
                            \seq{wRv}{\fm{w}{\BOX B}}{\fm{v}{B}}
                        }
                        {
                            \vlin{\id}{}
                            {
                                \seq{wRv}{\fm{v}{B}}{\fm{v}{B}}
                            }
                            {
                                \vlhy{}
                            }
                        }
                    }
                }
            }
        }
    }
}
$$

Now, recalling \cref{ex:forall-distributes-over-box}, here is a labelled proof of 
 $\fall \BOX A \IMP \BOX \fall A$:
$$
    \vlderivation
    {
        \vlin{\rr \IMP}{}
        {
            \seq{\cdot}{\cdot}{\fm{w}{\fall \BOX A \IMP \BOX \fall A}}
        }
        {
            \vlin{\rr \BOX}{}
            {
                \seq{\cdot}{\fm{w}{\fall \BOX A}}{\fm{w}{\BOX \fall A}}
            }
            {
                \vlin{\rr \forall}{}
                {
                    \seq{wRv}{\fm{w}{\fall \BOX A}}{\fm{v}{\fall A}}
                }
                {
                    \vlin{\lr \forall}{}
                    {
                        \seq{wRv}{\fm{w}{\fall \BOX A}}{\fm{v}{\subst{A}{P}}}
                    }
                    {
                        \vlin{\lr \BOX}{}
                        {
                            \seq{wRv}{\fm{w}{\BOX \subst{A}{P}}}{\fm{v}{\subst{A}{P}}}
                        }
                        {
                            \vlin{\id}{}
                            {
                                \seq{wRv}{\fm{v}{\subst{A}{P}}}{\fm{v}{\subst{A}{P}}}
                            }
                            {
                                \vlhy{}
                            }
                        }
                    }
                }
            }
        }
    }
$$

Let us point out that even $\blacksquare$-free theorems might need formulas including $\blacksquare$ in their cut-free proofs, due to the $\lr\forall$ steps involved,
exemplifying the \emph{non-analyticity} of second-order logic:
\begin{equation}
    \label{eq:lab-prf-neg-box}
    \vlderivation{
\vlin{\rr \Box}{}{\cdot \stoup v:\lnot \lnot \Box \bot \seqar v:\Box \bot}{
\vliin{\lr \limp}{}{vRw \stoup v:\lnot \lnot \Box \bot \seqar w:\bot}{
    \vlin{\rr\limp}{}{vRw \stoup \cdot \seqar v:\lnot \Box \bot}{
    \vlin{\rr\forall}{}{vRw \stoup v: \Box \bot \seqar v:\forall X X }{
    \vlin{\lr \Box}{}{vRw \stoup v:\Box \bot \seqar v: P}{
    \vlin{\lr \forall}{}{vRw \stoup w:\forall X X \seqar v:P}{
    \vlin{\lr\blacksquare}{}{vRw \stoup w:\blacksquare P \seqar v:P}{
    \vlin{\id}{}{vRw \stoup v:P \seqar v:P}{\vlhy{}}
    }
    }
    }
    }
    }
}{
    \vlin{\lr \forall}{}{vRw \stoup v:\forall X X \seqar w:\bot}{
    \vlin{\lr\Box}{}{vRw \stoup v:\Box \bot \seqar w:\bot}{
    \vlin{\id}{}{vRw \stoup w:\bot \seqar w:\bot}{\vlhy{}}
    }
    }
}
}
}
\end{equation}

Previous examples were intuitionistic proofs, of $\labIKtSO$.
The following are examples of classical proofs.

$$\vlderivation
{
    \vlin{\rr \BOX}{}
    {
        \seq{\cdot}{\fm{v}{\BOX \BOXB A}}{\fm{v}{A}, \fm{v}{\BOX \BOT}}
    }
    {
        \vlin{\lr \BOX}{}
        {
            \seq{vRw}{\fm{v}{\BOX \BOXB A}}{\fm{v}{A}, \fm{w}{\BOT}}
        }
        {
            \vlin{\lr \BOXB}{}
            {
                \seq{vRw}{\fm{w}{\BOXB A}}{\fm{v}{A}, \fm{w}{\BOT}}
            }
            {
                \vlin{\id}{}
                {
                    \seq{vRw}{\fm{v}{A}}{\fm{v}{A}, \fm{w}{\BOT}}
                }
                {
                    \vlhy{}
                }
            }
        }
    }
}
\qquad
\vlderivation
{
    \vlin{\rr \BOXB}{}
    {
        \seq{\cdot}{\fm{v}{\BOXB \BOX A}}{\fm{v}{A}, \fm{v}{\BOXB \BOT}}
    }
    {
        \vlin{\lr \BOXB}{}
        {
            \seq{vRw}{\fm{v}{\BOXB \BOX A}}{\fm{v}{A}, \fm{w}{\BOT}}
        }
        {
            \vlin{\lr \BOX}{}
            {
                \seq{vRw}{\fm{w}{\BOX A}}{\fm{v}{A}, \fm{w}{\BOT}}
            }
            {
                \vlin{\id}{}
                {
                    \seq{vRw}{\fm{v}{A}}{\fm{v}{A}, \fm{w}{\BOT}}
                }
                {
                    \vlhy{}
                }
            }
        }
    }
}$$
Using the two proofs above, we can obtain the following:
$$
\vlderivation
{
    \vlin{\lr \forall}{}
    {
        \seq{\cdot}{\fm{v}{\fall((\BOX \BOXB A \IMP X ) \IMP (\BOXB \BOX A \IMP X) \IMP X)}}{\fm{v}{\BOX \BOT}, \fm{v}{A}, \fm{v}{\BOXB \BOT}}
    }
    {
                \vliin{\lr \IMP}{}
                {
                    \seq{\cdot}{\fm{v}{(\BOX \BOXB A \IMP A ) \IMP (\BOXB \BOX A \IMP A) \IMP A}}{\fm{v}{\BOX\BOT}, \fm{v}{A}, \fm{v}{\BOXB\BOT}}
                }
                {
                    \vlin{\rr \IMP}{}
                    {
                        \seq{}{\cdot}{\fm{v}{\BOX \BOXB A \IMP A},\fm{v}{\BOX\BOT}}
                    }
                    {
                        \vlin{}{}
                        {
                            \seq{\cdot}{\fm{v}{\BOX \BOXB A}}{\fm{v}{ A},\fm{v}{\BOX\BOT}}
                        }
                        {\vlhy{\text{proof above}}}
                    }
                }
                {
                    \vliin{\lr \IMP}{}
                    {
                        \seq{\cdot}{\fm{v}{(\BOXB \BOX A \IMP A) \IMP A}}{\fm{v}{A}, \fm{v}{\BOXB\BOT}}
                    }
                    {
                        \vlin{\rr \IMP}{}
                        {
                            \seq{\cdot}{\cdot}{\fm{v}{\BOXB \BOX A \IMP A}, \fm{v}{\BOXB\BOT}}
                        }
                        {
                            \vlin{}{}
                            {
                                \seq{\cdot}{\fm{v}{\BOXB \BOX A}}{\fm{v}{A}, \fm{v}{\BOXB\BOT}}
                            }
                            {
                                \vlhy{\text{proof above}}
                            }
                        }
                    }
                    {
                        \vlin{\id}{}
                        {
                            \seq{\cdot}{\fm{v}{A}}{\fm{v}{A}}
                        }
                        {
                            \vlhy{}
                        }
                    }
                }
            }
        }
$$

As a final example, let us see a classical $\labKtSO$ proof of a non-constructive principle, $\forall X(A\lor B) \limp A \lor \forall X B$ for $X \notin \fv A$.
To save space, we shall omit ``$\cdot \stoup$'' at the beginning of each LHS as well as all labels for formulas (they are all the same).

\[
\vlderivation{
\vlin{\rr \limp}{}{           \seqar    \forall X (A\lor B) \limp A \lor \forall X B}{
\vlid{=}{}{           \forall X (A\lor B) \seqar    A \lor \forall X B}{
\vlin{\rr\forall}{}{           \forall X (A\lor B) \seqar     \forall Y ((A \limp Y) \limp (\forall X B \limp Y) \limp Y)}{
\vlin{2\rr\limp}{}{           \forall X (A\lor B) \seqar w : (A\limp P) \limp (\forall X B \limp P) \limp P}{
\vlin{\rr\cntr}{}{           \forall X (A\lor B) ,    A\limp P ,    \forall X B \limp P \seqar    P}{
\vliin{\lr \limp}{}{           \forall X (A\lor B) ,    A\limp P ,    \forall X B \limp P \seqar    P,   P}{
    \vlin{\rr\forall}{}{           \forall X (A\lor B) ,    A\limp P  \seqar    \forall X B,    P}{
    \vlin{\lr \forall}{}{           \forall X (A\lor B) ,    A\limp P  \seqar    B[Q/X],    P}{
    \vliin{\lr \limp}{}{            A\lor B[Q/X] ,    A\limp P  \seqar    B[Q/X],    P}{
        \vlid{=}{}{           A \lor B[Q/X] \seqar    A,   B[Q/X] }{
        \vlin{\lr \forall}{}{           \forall Y ((A \limp Y) \limp (B[Q/X] \limp Y) \limp Y \seqar    A,    B[Q/X]}{
        \vliiin{2\lr\limp}{}{   (A\limp A) \limp (B\limp A) \limp A \seqar    A,    B}{
            \vlin{\rr\limp}{}{            \seqar    A\limp A}{
            \vlin{\id}{}{           A \seqar    A}{\vlhy{}}
            }
        }{
            \vlin{\rr\limp}{}{           \seqar    B\limp A,    B}{
            \vlin{\rr\wk}{}{           B \seqar    A,    B}{
            \vlin{\id}{}{          B \seqar    B}{\vlhy{}}
            }
            }
        }{
            \vlin{\id}{}{          A \seqar    A}{\vlhy{}}
        }
        }
        }
    }{
        \vlin{\id}{}{           P \seqar    P}{\vlhy{}}
    }
    }
    }
}{
    \vlin{\id}{}{           P \seqar    P}{\vlhy{}}
}
}
}
}
}
}
}
\]

\section{Completeness via proof search: the classical case}
\label{sec:completeness-classical}

One of the main technical contributions of the present work is the cut-free completeness of our sequent calculi via proof search.
For second (and higher) order (classical) predicate logic cut-admissibility remained an open problem since the late '30s, what came to be known as \emph{Takeuti's conjecture} (see, e.g., \cite[Section~5.1]{sep-proof-theory}). 
In the late '60s, the problem was finally resolved positively by Prawitz \cite{Prawitz1968:SOL} and Tait \cite{Tait1966Nonconstructive} for second order logic, and extended to simple type theory (i.e.\ higher order logic) by Prawitz \cite{Prawitz1968:STT} and Takahashi \cite{Takahashi1967}.
Both exploited Sch\"utte's technique of \emph{partial valuations} \cite{Schutte1960:semivaluations}, the key result being that they may be appropriately extended into \emph{total valuations}. See, e.g., \cite[Chapter~IV]{Schutte1977:pt-book} or \cite[Section~21]{Takeuti1987:pt-book} for textbook presentations of the methodology for simple type theory.

This is the starting point for our argument, from which we must employ a number of adaptations.
We shall avoid working explicitly with the notion of partial valuation, rather only demonstrating the properties we need, but making due remarks to the classical methodology throughout.
In the intuitionistic setting our argument is rather more involved, with little prior related work.
For this reason we shall start with the classical case, in this section, showcasing the main ideas of the methodology before addressing the further technicalities necessary for the intuitionistic case in the next section.
Thus the main result of this section is:

\begin{theorem}
    [Classical cut-free completeness]
    \label{cut-free-completeness-classical}
$\forall \RRR\, \modelsM\RRR A \implies \labKtSO \setminus \cut \proves A$.
\end{theorem}

As expected we shall approach this via contraposition, extracting a countermodel from failed proof search.
As already mentioned, our proof exploits ideas about partial valuations originating from \cite{Schutte1960:semivaluations}; our method for extending them to total valuations adapts Prawitz' concept of \emph{possible values} from his work on simple type theory \cite{Prawitz1968:STT}.\footnote{Prawitz used a different method via transfinite recursion in his original work on second order logic \cite{Prawitz1968:SOL}, but we find the possible values approach cleaner.}

\subsection{Setting up proof search}
\label{sec:setting-up-prf-srch}
First and foremost, we shall think of building proofs bottom-up, from the conclusion towards initial sequents. We shall describe the proof search process with this view in mind.

\subsubsection{Terminology for identifying formula occurrences}
    We shall use standard terminology about relationships between labelled formula occurrences in proofs and inference steps.
    In particular the \textbf{principal} formula of a logical step is the distinguished labelled formula occurrence in the lower sequent, as typeset in \cref{fig:lab-sys-KtSO}. 
    \textbf{Auxiliary} formulas are any distinguished labelled formula occurrences in the upper sequent(s).
    A good account for this and related terminology can be found in \cite[Section~1.2.3]{handbook-of-pt}.

\subsubsection{Cedents-as-sets}
In proof search, it is useful to construe cedents as \emph{sets} rather than multisets of labelled formulas.
Of course this makes no difference to usual provability, in the presence of structural rules. 
In practice, if we have two (or more) occurrences of the same labelled formula in a cedent, one can be safely \textbf{weakened} (i.e.\ deleted by applying $\lr \wk$ or $ \rr \wk$ steps). 
In case an additional copy is required, we always \textbf{contract} (i.e.\ duplicate by applying $\lr \cntr$ or $\rr\cntr$ steps) principal formulas of logical steps.
Concretely the steps of our proof search algorithm will be composed of the following \textbf{macro rules}:

\begin{equation}
\hspace{-5em}
\label{eq:mon-rules}
    \begin{array}{cc}
\vliinf{\lr \limp}{}{\rels R \stoup \Gamma, v: A\limp B \seqar \Delta}{\rels R \stoup \Gamma, v: A\limp B \seqar \Delta, v: A}{\rels R \stoup \Gamma,v:A\limp B , v: B \seqar \Delta}
&
\vlinf{\rr\limp }{}{\rels R \stoup \Gamma \seqar \Delta, v:A\limp B}{\rels R\stoup \Gamma , v:A \seqar \Delta, v:A\limp B, v:B}
\\
\noalign{\smallskip}
\vlinf{\lr \forall}{}{\rels R \stoup \Gamma , v : \forall X A \seqar \Delta}{\rels R \stoup \Gamma , v : \forall X A , v : A [C/X] \seqar \Delta}
&
\vlinf{\rr \forall}{\text{\footnotesize $P$ fresh}}{\rels R \stoup \Gamma \seqar \Delta, v:\forall X A}{\rels R \stoup \Gamma \seqar \Delta, v:\forall X A , v: A[P/X] }
\\
\noalign{\smallskip}
\vlinf{\lr \Box}{}{\rels R , uRv \stoup \Gamma , u: \Box A \seqar \Delta}{\rels R , uRv \stoup \Gamma , u : \Box A , v : A \seqar \Delta }
&
\vlinf{\rr\Box}{\text{\footnotesize $w$ fresh}}{\rels R \stoup \Gamma \seqar \Delta, v:\Box A}{\rels R , vRw \stoup \Gamma \seqar \Delta , v:\Box A , w:A}
\\
\noalign{\smallskip}
\vlinf{\lr \blacksquare}{}{\rels R , uRv \stoup \Gamma, v : \blacksquare A \seqar \Delta}{\rels R , uRv \stoup \Gamma , v : \blacksquare A , u : A \seqar \Delta}
&
\vlinf{\rr\blacksquare}{\text{$u$ fresh}}{\rels R \stoup \Gamma \seqar \Delta, v:\blacksquare A}{\rels R , uRv \stoup \Gamma \seqar \Delta, v:\blacksquare A , u :A}
\end{array}
\end{equation}
As mentioned, these macro rules are readily derivable from the ones of $\labKtSO$, using the structural rules.
Let us point out already that the left rules above are even derivable in $\mlIKtSO$ (but not the right ones).
The only other rules we (implicitly) apply during proof search are weakenings, $\lr \wk, \rr \wk$, to delete extra copies of a labelled formula as already mentioned.
We shall omit mention of such bookeeping henceforth.

Notice that these rules are \textbf{monotone}: bottom-up, sequents only get bigger.
Let us point out two important consequences of this:
\begin{itemize}
    \item  These rules are also invertible: provability of the conclusion \emph{implies} provability of the premisses (by weakening).
    \item Any infinite branch of macro steps above has a well-defined (possibly infinite) limit by taking unions of relational contexts, unions of LHSs, and unions of RHSs.
\end{itemize}

\subsubsection{Enumerating steps by activity}
Now, fixing a conclusion sequent, note that the steps in \cref{eq:mon-rules} are determined by the choice of principal and auxiliary labelled formulas, as well as which side the principal formula occurs on.
Let us call these data the \textbf{activity} of an inference step.
As activities are determined by only finite data, they are only countably many so may be enumerated by $\omega$.
The idea of our proof search algorithm will be to apply each of these steps in turn, so that every possible labelled formula on either side of a sequent is eventually principal. 
To ensure we do not miss any possible steps, we shall employ an enumeration with sufficient redundancy:

\begin{convention}
[Adequate enumeration]
\label{adeq-enum-activ}
    We assume an enumeration $(\alpha_i)_{i<\omega}$ of activities in which every possible activity occurs infinitely often.
\end{convention}

\begin{remark}
    [Activity vs inference step]
    Note that distinct inference steps may have the same activity, as their contexts may be different, but once a concluding sequent is fixed the activity determines at most one inference step.
    It may also determine no correct inference step at all, e.g.\ if the corresponding principal formula does not occur in the sequent, or only appears on the wrong side.
\end{remark}

\subsection{The proof search branch}
\label{sec:prf-srch-brnch}
Let us set up some notation for convenience.
We shall write $\sequent , \sequent'$ etc.\ to vary over labelled sequents.
If $\sequent = \rels R \stoup \Gamma \seqar \Delta$ and $\sequent' = \rels R' \stoup \Gamma' \seqar \Delta'$, we write simply $\sequent \subseteq \sequent'$ if $\rels R \subseteq \rels R'$, $\Gamma \subseteq \Gamma'$ and $\Delta \subseteq \Delta'$.
For a set of sequents $\{\sequent_i = \rels R_i \stoup \Gamma_i \seqar \Delta_i \}_{i \in I}$, 
we write simply $\bigcup\limits_{i \in I} \sequent_i$ for the sequent $\bigcup\limits_{i\in I}\rels R_i \stoup \bigcup\limits_{i \in I}\Gamma_i \seqar \bigcup\limits_{i\in I} \Delta_i$.
We shall typically reserve this union operation for when we take limits of chains of sequents under $\subseteq$.

For the remainder of this section let us fix a sequent $\sequent_0 = \rels R_0 \stoup \Gamma_0 \seqar \Delta_0$ that is unprovable in $\labKtSO$, i.e.\ $\labKtSO \not \proves \sequent_0$.

\begin{definition}
[Proof search branch]
\label{prf-srch-brnch}
We extend $\sequent_0$ to a sequence $\prfsrchtree = (\sequent_i)_{i<\omega}$ of unprovable sequents defined as follows:
    \begin{itemize}
        \item If no inference step has conclusion $\sequent_i$ and activity $\alpha_i$ then just set $\sequent_{i+1} := \sequent_i$ (so $\sequent_{i+1}$ remains unprovable).
        \item 
        Otherwise let $\infrul_i$ be the inference step with conclusion $\sequent_i$ and activity $\alpha_i$.
        \item By assumption that $\sequent_i$ is unprovable, some premiss of $\infrul_i$ must be unprovable. 
        We set $\sequent_{i+1}$ to be some unprovable premiss.\footnote{It does not matter which, but for concreteness, we may take the leftmost if there is a choice.}
    \end{itemize}
\end{definition}

As mentioned earlier, since the macro rules of \cref{eq:mon-rules} are monotone, bottom-up, we indeed have a chain $\sequent_0 \subseteq \sequent_1 \subseteq \cdots$. 
So we set $\sequent_\omega = \rels R_\omega \stoup \Gamma_\omega \seqar \Delta_\omega $ to be the (infinite) limit of this sequence, i.e.\ $ \bigcup\limits_{i<\omega} \sequent_i$.

Importantly we have:
\begin{proposition}
    \label{LHS-RHS-disjoint-prf-search}
    $\Gamma_\omega \cap \Delta_\omega = \emptyset$.
\end{proposition}
\begin{proof}
    If $\Gamma_\omega $ and $\Delta_\omega$ intersect, then so does some $\Gamma_i$ and $\Delta_i$, by monotonicity. 
    This would mean that $\sequent_i$ is derivable by an $\id $ step (and weakenings), contradicting its unprovability.
\end{proof}

\subsection{Towards a countermodel: a pre-structure from proof search}
\label{sec:pre-structure-classical}
The proof search branch gives rise to a \emph{pre-structure}, that is a structure lacking a domain of sets (and thus also an interpretation of propositional symbols), that underlies our eventual countermodel for $\sequent_0$.
Let us define this now:

\begin{definition}
    [Pre-structure]
    \label{pre-structure}
    We define $\pre \RRR$ by the following data:
    \begin{itemize}
        \item The set of worlds is just $\war$.
        \item We define the accessibility relation to just be $\rels R_\omega$, i.e.\ $v\rels R_\omega w$ if $vRw \in \rels R_\omega$.
    \end{itemize}
\end{definition}

When $\sequent_0$ contains only first-order formulas (i.e.\ without quantifiers), the pre-structure above easily induces a bona fide countermodel (see, e.g., \cite[Theorem~8.17]{Takeuti1987:pt-book} for a similar exposition for first-order intuitionistic predicate logic).
The key difficulty for expanding $\pre \RRR$ in our second-order setting is to identify an appropriate domain of predicates/sets.

\subsection{Extracting a partial valuation}
\label{sec:partial-val-classical}
As in other countermodel constructions, the idea is to find a structure that forces all of $\Gamma_\omega$ true and all of $\Delta_\omega $ false.
This desideratum can be used to constrain what predicates may be,
but the issue is that this information is incomplete: the truth values of some formulas are not determined by $\Gamma_\omega, \Delta_\omega$.
Attempting to fix them one way or another may lead to inconsistencies, in particular due to contravariance of $\limp$.

Sch\"utte dubbed such an assignment a \emph{semivaluation} in \cite{Schutte1960:semivaluations}, or \emph{partial valuation} in \cite{Schutte1977:pt-book}.\footnote{Beware that what Sch\"utte calls `partial valuation' in \cite{Schutte1960:semivaluations} is slightly different, comprising a sort of expansion of a semivaluation so that it is closed under semantic clauses. We keep to the current terminology as it seems more suggestive and should be unlikely to cause confusion.}
This concept was further expanded into a bona fide 3-valued semantics of cut-free proofs by Girard \cite{Girard1987:pt-log-comp}.
We stop short of working explicitly with partial valuations, for the sake of reducing the technical development, rather simply using predicates for truth and falsity induced by the LHS $\Gamma_\omega$ and RHS $\Delta_\omega$, respectively, of the proof search branch $\prfsrchtree$.

Our construction of the proof search branch is designed to guarantee the following crucial property, corresponding to Sch\"utte's definition of semivaluation \cite[Definition~6.1]{Schutte1960:semivaluations} (later called partial valuation in \cite[Section 11.2]{Schutte1977:pt-book}):

\begin{proposition}
    [Partial valuation]
    \label{partial-val-props-classical}
    We have the following:
    \begin{enumerate}
        \item $v:A\limp B \in \Gamma_\omega \implies (v:A \in \Delta_\omega \text{ or } v: B \in \Gamma_\omega)$.
        \item $v:A\limp B \in \Delta_\omega \implies (v: A \in \Gamma_\omega \text{ and } v:B \in \Delta_\omega)$
        \item $v:\Box A \in \Gamma_\omega \implies \forall w \in \war \,  (v \rels R_\omega w \implies w:A \in \Gamma_\omega)$
        \item $v:\Box A \in \Delta_\omega \implies \exists w \in \war \,  (v \rels R_\omega w \text{ and } w:A \in \Delta_\omega)$
        \item $v:\blacksquare A \in \Gamma_\omega \implies \forall u \in \war \,  (u \rels R_\omega v \implies u:A \in \Gamma_\omega)$
        \item $v:\blacksquare A \in \Delta_\omega \implies \exists u \in \war \, (u\rels R_\omega v\text{ and } u:A \in \Delta_\omega)$
        \item $v:\forall X A \in \Gamma_\omega \implies \forall C \in \Fmla \, v: A[C/X] \in \Gamma_\omega$.
        \item $v:\forall X A \in \Delta_\omega \implies \exists C \in \Fmla \, v:A[C/X] \in \Delta_\omega$.
    \end{enumerate}
\end{proposition}
\begin{proof}
    If $v:A \in \Gamma_\omega$ then $v:A \in \Gamma_i$ for some $i$.
    Now, by the definition of the proof search branch and the enumeration of activities we take, cf.~\cref{adeq-enum-activ}, any arbitrary activity $\alpha$ with $v:A$ principal will be applied at some stage $j>i$.
    The properties above simply exhaust all the possibilities of activity for a given principal formula, and possibilities for extension of the proof search branch from $i$.
\end{proof}

Returning to Sch\"utte's partial valuations,
we can think of $v:A \in \Gamma_\omega$ as indicating that $A$ is true at world $v$, with respect to $\pre \RRR$ (or an appropriate expansion of it by predicates).
Symmetrically $v:A \in \Delta_\omega$ indicates that $A$ is false at world $v$.

\subsection{A compatible countermodel via possible values}
To define an appropriate domain of sets, Prawitz' idea in \cite{Prawitz1968:STT} was to simply take \emph{all} possibilities consistent with the proof search branch, by way of so-called \emph{possible values}. 
We shall follow a similar idea here, though again we avoid explicitly defining possible values.
Recall that we have already fixed an unprovable sequent $\sequent_0$ and extended it to the proof search branch $\prfsrchtree = (\sequent_i)_{i<\omega}$.

\begin{definition}
    [Possible extensions]
    \label{dfn:poss-ext-classical}
    A \textbf{possible extension} of a formula $C$ (with respect to $\prfsrchtree$) is a set $\possext C \subseteq \war$ s.t.:
    \begin{itemize}
        \item $v:C \in \Gamma_\omega \implies v \in \possext C$; and,
        \item $v:C \in \Delta_\omega \implies v \notin \possext C$.
    \end{itemize}
    Write $\possext C \textends C$ if $\possext C$ is a possible extension of $C$.
\end{definition}

It is obvious but pertinent to state that every formula $C$ admits a \textbf{minimal (possible) extension} $\minext C := \{v \in \war \ | \ v:C \in \Gamma^\omega \}$.
A unirelational structure may now be obtained from our pre-structure by allowing all possible extensions as sets:

\begin{definition}
[Countermodel]
    We expand $\pre \RRR$ into a structure $\RRR$ by including the following missing data:
    \begin{itemize}
        \item The domain of sets $\set W \subseteq \pow \war$ includes all possible extensions (of all formulas).
        \item We set $\interp { \RRR} {} P := \minext P$ for each $P \in \Prop$.\footnote{In fact the interpretation of propositional symbols is inconsequential, among possible extensions, but we choose the minmal one for concreteness.}
    \end{itemize}
\end{definition}

The key technical result we need about this structure relates the evaluation of formulas over possible extensions to the desideratum that $\Gamma_\omega$ be true and $\Delta_\omega $ be false.
Such compatibility with the partial valuation induced by proof search is what will allow us to show that $\RRR$ is, in fact, comprehensive.

\begin{lemma}
[Compatibility]
\label{compatibility}
    For formulas $A(\vec X)$ (all free variables among $\vec X = X_1, \dots , X_n$) and $\vec C = C_1, \dots , C_n$ we have:
    \begin{enumerate}
        \item $v:A(\vec C) \in \Gamma_\omega \implies v \modelsM\RRR A(\vec {\possext { C}})$ whenever $\vec {\possext C} \textends \vec C$.
        \item $v:A(\vec C) \in \Delta_\omega \implies v\not\modelsM \RRR A(\vec {\possext C})$ whenever $\vec {\possext C} \textends \vec C$.
    \end{enumerate}
    where we write $\vec {\possext C} \textends \vec C$ for $\possext C_1 \textends C_1 , \dots, \possext C_n \textends C_n $.
\end{lemma}

Before proving this, let us point out an immediate consequence.
Write $\modelsM\RRR \rels R\stoup \Gamma \seqar \Delta$ if $\rels R \subseteq \rels R_\omega$ and, for each $v:A \in \Gamma $ (or $v:A \in \Delta$) we have $v\modelsM\RRR A$ (or $v\not\modelsM\RRR A$, respectively). 
As a special case of \cref{compatibility} above we have:
\begin{proposition}
\label{R-is-countermodel}
$\not\modelsM\RRR \sequent_0$.
\end{proposition}

\begin{proof}
[Proof of \cref{compatibility}]
    By induction on $A(\vec X)$:
    \begin{itemize}
        \item Suppose $A(\vec X) = X$, and fix $C$ and $\possext C \textends C$. We have:
        \[
        \begin{array}{r@{\ \implies \ }ll}
        v:C \in \Gamma_\omega & v \in \possext C & \text{since $\possext C\textends C$} \\
            & v \modelsM\RRR \possext C & \text{by definition of $\modelsM\RRR$} \\
        \noalign{\smallskip}
        v:C \in \Delta_\omega & v \notin \possext C & \text{since $\possext C \textends C$} \\
            & v \not\modelsM\RRR \possext C & \text{by definition of $\modelsM\RRR$}
        \end{array}
        \]
        \item Suppose $A(\vec X) = A_0(\vec X) \limp A_1 (\vec X) $, and fix $\vec C$ and $\vec {\possext C} \textends \vec C$. We have:
        \[
        \begin{array}{r@{\ \implies \ }ll}
            v:A(\vec C) \in \Gamma_\omega & v:A_0(\vec C) \in \Delta_\omega \text{ or } v:A_1 (\vec C) \in \Gamma_\omega & \text{by \cref{partial-val-props-classical}} \\
                & v \not\modelsM \RRR A_0(\vec{ \possext C}) \text{ or } v \modelsM\RRR A_1(\vec {\possext C}) & \text{by IH} \\
                & v \modelsM\RRR A(\vec {\possext C}) & \text{by definition of $\modelsM\RRR$} \\
            \noalign{\smallskip}
            v:A(\vec C) \in \Delta_\omega & v:A_0(\vec C) \in \Gamma_\omega \text{ and } v:A_1(\vec C) \in \Delta_\omega & \text{by \cref{partial-val-props-classical}} \\
                & v\modelsM\RRR A_0 (\vec {\possext C}) \text{ and } v:\not\modelsM\RRR A_1(\vec {\possext C}) & \text{by IH} \\
                & v\not\modelsM \RRR A(\vec {\possext C}) & \text{by definition of $\modelsM\RRR$}
        \end{array}
        \]
        \item Suppose $A(\vec X) = \Box A'(\vec X)$, and fix $\vec C$ and $\vec {\possext C} \textends \vec C$. We have:
        \[
        \begin{array}{r@{\ \implies \ }ll}
        v:A(\vec C) \in \Gamma_\omega & \forall w \in \war \, ( v \rels R_\omega w \implies w:A'(\vec C) \in \Gamma_\omega  ) & \text{by \cref{partial-val-props-classical}} \\
            & \forall w \in \war\,  ( v\rels R_\omega w\implies w \modelsM\RRR A'(\vec {\possext C}) )& \text{by IH} \\
            & v \modelsM\RRR A(\vec {\possext C}) & \text{by definition of $\modelsM\RRR$}
        \\
        \noalign{\smallskip}
        v:A(\vec C) \in \Delta_\omega & \exists w \in \war\, (v \rels R_\omega w \text{ and } w:A'(\vec C) \in \Delta_\omega) & \text{by \cref{partial-val-props-classical}} \\
            & \exists w \in \war\, (v\rels R_\omega w \text{ and } w\not \modelsM\RRR A'(\vec {\possext C})) & \text{by IH} \\
            & v \not\modelsM \RRR A(\vec {\possext C}) & \text{by definition of $\modelsM \RRR$}
        \end{array}
        \]
        \item The case when $A(\vec X) = \blacksquare A'(\vec X)$ is similar to the one above.
        \item Suppose $A(\vec X) = \forall X A'(X , \vec X)$, and fix $\vec C$ and $\vec {\possext C} \textends \vec C$. We have:
        \[
        \begin{array}{r@{\ \implies \ }ll}
        v:A(\vec C) \in \Gamma_\omega & \forall C \in \Fmla \, v:A'(C,\vec C) \in \Gamma_\omega & \text{by \cref{partial-val-props-classical}} \\
            & \forall C \in \Fmla \, \forall \possext C \textends C \, v \modelsM \RRR A'(\possext C, \vec {\possext C}) & \text{by IH} \\
            & v \modelsM \RRR A(\vec {\possext C}) & \text{by definition of $\modelsM \RRR$} \\
        \noalign{\smallskip}
        v:A(\vec C) \in \Delta_\omega & \exists C \in \Fmla\, v:A'(C,\vec C) \in \Delta_\omega & \text{by \cref{partial-val-props-classical}} \\
            & \exists C\in \Fmla \, v\not\modelsM \RRR A'(\minext C , \vec {\possext C} ) & \text{by IH}\\
            & v\not\modelsM \RRR A(\vec {\possext C}) & \text{by definition of $\modelsM\RRR$} \qedhere
        \end{array}
        \]
    \end{itemize}
\end{proof}

\subsection{Putting it all together: comprehensivity via a total valuation}
In light of \cref{R-is-countermodel}, for our completeness result it remains to show that $\RRR$ is comprehensive.
Since $\RRR$ is already defined, we can do this in a somewhat backwards way:
\begin{definition}
    [Interpreting comprehension]
    Define $\extension C := \{ v  \in \war \, |\, v\modelsM \RRR C\}$.
\end{definition}

Now, \cref{compatibility} has another useful consequence:

\begin{proposition}
    $\extension C $ is a possible extension of $C$, for all formulas $C$.
\end{proposition}
\begin{proof}
    Simply set $A(\vec X) = C$ and $\vec C = \emptyset = \vec X$ in \cref{compatibility}.
\end{proof}

\begin{corollary}
\label{R-is-comprehensive}
    $\RRR$ is comprehensive, and so is a relational model.
\end{corollary}

In the terminology of Sch\"utte, $\extension -$ corresponds to a \emph{total valuation} of formulas \cite[Definition~7.2]{Schutte1960:semivaluations}.
Note the crucial use of \emph{impredicativity} here: we are only able to define $\extension C$ in terms of a structure, $\RRR$, that includes all possible extensions, including $\extension C$ itself.

We have now established the main result of this section:

\begin{proof}
    [Proof of \cref{cut-free-completeness-classical}]
    By contraposition. Set $\sequent_0 := \cdot \stoup \cdot \seqar v:A$ throughout this section and conclude by \cref{R-is-comprehensive,R-is-countermodel}.
\end{proof}

\section{Completeness via proof search: the intuitionistic case}\label{sec:completeness}
We now turn to the intuitionistic case for cut-free completeness.
As we have presented the classical case in detail, we shall focus in this section on how the intuitionistic case deviates from the classical one.
Let us summarise the key points:
\begin{itemize}
    \item Not all rules of $\labIKtSO$ or $\mlIKtSO$ can be made invertible.
    For the former this is because we lack right structural rules and so, e.g., $\lr \limp$ is not invertible. 
    For the latter this is because the right logical rules are not invertible.
    \item As a result, instead of a proof search branch, we will construct a proof search \emph{tree}. This is unsurprising given the shape of intuitionistic structures, which are partially ordered.
    \item We will need to deal with limit sequents at each branching point of the countermodel construction, so our proof system must now work with genuinely \emph{infinite sequents}.
\end{itemize}

    The obtention of an appropriate proof search tree from which we can extract countermodels is facilitated by working directly with the calculus $\mlIKtSO$, instead of $\labIKtSO$.
    Cut-free completeness results for multi-succedent calculi have also appeared in \cite[pp.~52-59]{Takeuti1987:pt-book} for first-order intuitionistic predicate logic over Kripke models, and in \cite{prawitz1970some} for second-order intuitionistic predicate logic over Beth models.
    Note that, in our setting we furthermore recover cut-free completeness of the single-succedent calculus $\labIKtSO$ thanks to \cref{prop:multitosingle}, even over all birelational models thanks to \cref{prop:breltoprel}. 
 In summary, the main results of this section are:

\begin{theorem}
    [Intuitionistic cut-free completeness]\label{thm:completeness}
    $\forall \PP \, \modelsM\PP A \implies \mlIKtSO ~\setminus~\cut~\proves~A$.
\end{theorem}
\begin{corollary}
    $\forall \B \, \modelsM\B A \implies \labIKtSO ~\setminus~\cut~\proves~A$.
\end{corollary}

The remainder of this section is structured as the previous one.

\subsection{Setting up proof search}
We now construe $\mlIKtSO$ over (possibly) infinite sequents $\rels R \stoup \Gamma \seqar \Delta$, where some/all of $\rels R, \Gamma, \Delta$ may be (countably) infinite.
We say that such a sequent is \textbf{provable} in, or a \textbf{theorem} of, $\mlIKtSO$ if there are finite $\rels R'\subseteq \rels R$, $\Gamma' \subseteq \Gamma$ and $\Delta' \subseteq \Delta$ such that $\rels R' \stoup \Gamma' \seqar \Delta'$ is provable in $\mlIKtSO$, in the usual sense.
Note that infinitude of sequents does not affect the fact that inference steps on them are still determined by activity of finite data, cf.~\cref{adeq-enum-activ}.

All the discussion of \cref{sec:setting-up-prf-srch} still applies when building proofs of $\mlIKtSO$, however we shall consider only the left macro rules of \cref{eq:mon-rules}, as the right ones are not derivable in $\mlIKtSO$.
We thus restrict our enumeration of activities from \cref{adeq-enum-activ} to ones with left principal formulas.
Our proof search algorithm will now alternate between two phases: the \emph{LHS phase}, applying left rules, and the \emph{RHS phase}, applying right rules.
We will again recast proof search as a countermodel construction, building a predicate structure where the LHS phase determines interpretations \emph{within} an intuitionistic world and the RHS phase determines a tree structure inducing the corresponding partial order.

\smallskip

A bit more formally, the \textbf{LHS phase} initiates at some (possibly infinite) sequent $\sequent_0$ unprovable in $\mlIKtSO$ and is defined just like the proof search branch in \cref{prf-srch-brnch}, producing a branch $(\sequent_i)_{i<\omega}$ of unprovable sequents. 
That is, at the $i$\textsuperscript{th} stage it applies the step with activity $\alpha_i$ and conclusion $\sequent_i$ (if it exists), setting $\sequent_{i+1}$ to be some unprovable premiss (or $\sequent_i$, respectively).
Just like in \cref{sec:prf-srch-brnch} this phase is monotone: $i\leq j \implies \sequent_i \subseteq \sequent_j$. Thus the steps applied are also invertible, by weakening.
Writing again $\sequent_\omega = \bigcup\limits_{i<\omega}\sequent_i$, note that, if $\sequent_\omega$ were provable then so also would be $\sequent_i$ for some $i<\omega$, by definition of provability of infinite sequents, so we have:

\begin{observation}
\label{limit-stages-remain-unprovable}
    If $\mlIKtSO\not \proves \sequent_0$ and $\sequent_\omega$ is the limit of the LHS phase from $\sequent_0$, then also $\mlIKtSO \not \proves \sequent_\omega$.
\end{observation}

The \textbf{RHS phase} initiates at a (possibly infinite) sequent $\sequent = \rels R \stoup \Gamma \seqar \Delta$ unprovable in $\mlIKtSO$ and simply applies, bottom-up, a single right logical step of $\mlIKtSO$.
Such a step is determined by the choice of a principal formula in $\Delta$ (up to renaming of fresh symbols).
Again, the (only) premiss must remain unprovable.
Note that, by inspection of the right logical steps, the RHS phase always ends at a sequent $\sequent' = \rels R' \stoup \Gamma' \seqar \Delta'$ where $\Delta'$ is a \emph{singleton}.
Note that the RHS phase is \emph{not} monotone for the RHS, but remains monotone for the LHS, i.e.\ we still have
 $\rels R' \supseteq \rels R$ and $\Gamma'\supseteq \Gamma$ but not, in general, $\Delta' \supseteq \Delta$.
 Thus it is not, in general, invertible.

\subsection{The proof search tree}
As already mentioned, branching in our proof search tree will occur in the RHS phase. 
Since the RHS phase is determined by a choice of principal labelled formula in the RHS, we shall name the nodes of our tree structure accordingly.
Write $\sigma, \tau,$ etc.\ to vary over $\Baire$. 
We shall write $::$ for concatenation of finite sequences.
Write $\sqsubseteq$ for the prefix order on $\Baire$, i.e.\ $\sigma \sqsubseteq \tau$ if $\tau $ can be written as $ \sigma :: \sigma'$.  
Of course, $\sqsubseteq$ is indeed a partial order.

For the remainder of this section we fix a (possibly infinite) sequent $\sequent^\epsilon_0 = \rels R^\epsilon_0 \stoup \Gamma^\epsilon_0 \seqar \Delta^\epsilon_0$  that is unprovable in $\mlIKtSO$.

\begin{definition}
    [Proof search tree]
    \label{dfn:prf-srch-tree}
    The \textbf{proof search tree} $\mathfrak S $ consists of a tree $ T \subseteq \Baire$ and sequents $\{\sequent^\sigma_i = \rels R^\sigma \stoup \Gamma^\sigma_i \seqar \Delta^\sigma_i\}_{\sigma \in T, i<\omega}$
defined as follows:
\begin{itemize}
    \item For any sequent $\sequent^\sigma_0$, apply the LHS phase to construct a chain $\sequent^\sigma_0 \subseteq \sequent^\sigma_1 \subseteq \cdots$ of unprovable sequents. 
    Define $\sequent^\sigma = \rels R^\sigma \stoup \Gamma^\sigma \seqar \Delta^\sigma $ to be the limit $\sequent^\sigma_\omega$ of this chain (which again must be unprovable, cf.~\cref{limit-stages-remain-unprovable}).
    \item For each formula $v:A \in \Delta^\sigma$, $\sigma$ has a child $\sigma ::(v:A)$ in $T$.
    We set $\sequent^{\sigma :: (w:A)}_0$ to be the (unique) premisse of the (unique, up to renaming of fresh symbols) inference step with conclusion $\sequent^\sigma$ and principal formula $v:A$.
\end{itemize}    
\end{definition}

Once again, since $\mathfrak S$ has been constructed to include only unprovable sequents,  like \cref{LHS-RHS-disjoint-prf-search} in the classical case we importantly have:

\begin{proposition}
\label{LHS-RHS-disjoint-prf-search-intuitionistic}
   $\Gamma^\sigma \cap \Delta^\sigma = \emptyset$, for all $\sigma \in T$.
\end{proposition}

\subsection{Towards a countermodel: a pre-structure from proof search}
Just like in \cref{sec:pre-structure-classical}, the proof search tree $\mathfrak S$ we constructed gives rise to a `pre-structure', i.e.\ a predicate structure lacking a domain of sets (and thus also an interpretation of propositional symbols):

\begin{definition}
    [Pre-structure]
We define the predicate structure $\pre \PP$ by:
\begin{itemize}
    \item The set of intuitionistic worlds is just $T\subseteq \Baire$.
    \item The partial order is just the restriction of $\sqsubseteq $ to $T$.
    \item The set of modal worlds is just $\war$.
    \item The accessibility relation at $\sigma \in T$ is just $\rels R^\sigma$, i.e.\ $v\rels R^\sigma w$ if $vRw \in \rels R^\sigma$.
\end{itemize}
\end{definition}

Note that $\sigma \sqsubseteq \tau \implies \rels R^\sigma \subseteq \rels R^\tau$, by monotonicity in the LHS and RHS phases.

\subsection{Extracting a partial valuation}
This part of the argument is similar to the classical case, in \cref{sec:partial-val-classical}, instead establishing bespoke local properties compatible with our predicate semantics:
\begin{proposition}
    [Partial valuation]
    \label{t-induces-partial-val}
    We have the following:
    \begin{enumerate}
        \item $v:A\limp B \in \Gamma^\sigma \implies \forall \tau \sqsupseteq \sigma\, ( v : A \in \Delta^\tau \text{ or } v : B \in \Gamma^\tau )$
        \item $v: A \limp B \in \Delta^\sigma \implies \exists \tau \sqsupseteq \sigma \, (v:A \in \Gamma^\tau \text{ and } v:B \in \Delta^\tau )$
        \item $v: \Box A \in \Gamma^\sigma \implies \forall \tau \sqsubseteq \sigma\,  \forall w \in \war\, ( v  \rels R^\tau w \implies w:A \in \Gamma^\tau) $
        \item $v: \Box A \in \Delta^\sigma \implies \exists \tau \sqsupseteq \sigma \, \exists w \in \war\, (v \rels R^\tau w \text{ and } w:A \in \Delta^\tau)$
        \item $v: \blacksquare A \in \Gamma^\sigma \implies \forall\tau \sqsupseteq \sigma \, \forall u \in \war\, (u \rels R^\tau v \implies u:A \in \Gamma^\tau)$
        \item $v: \blacksquare A \in \Delta^\sigma \implies \exists \tau\sqsupseteq \sigma \, \exists u \in \war \, (u \rels R^\tau v \text{ and } u:A \in \Delta^\tau) $
        \item $v: \forall X A \in \Gamma^\sigma \implies \forall \tau \sqsupseteq \sigma\,  \forall C \in \Fmla\ v: A[C/X] \in \Gamma^\tau$
        \item $v:\forall X A \in \Delta^\sigma \implies \exists \tau \sqsupseteq \sigma \, \exists C \in \Fmla\ v: A[C/X] \in \Delta^\tau$
    \end{enumerate}
\end{proposition}
\begin{proof}
We consider each case separately:
    \begin{enumerate}
        \item Let $\tau \sqsupseteq \sigma$. By monotonicity we have $v:A\limp B \in \Gamma^\tau$, and so either $v:A \in \Delta^\tau$ or $v:B \in \Gamma^\tau$, depending on which direction $\Gamma^\sigma $ takes at the corresponding $\lr \limp $ step.
        \item If $v:A\limp B \in \Delta^\sigma$ then we can just set $\tau = \sigma :: (v:A\limp B)$.
        \item Let $\tau \sqsupseteq \sigma $ and $v \rels R^\tau w$. By monotonicity we have $v:\Box A \in \Gamma^\tau$, and so also $w:A \in \Gamma^\tau$ by the corresponding $\lr \Box $ step.
        \item If $v:\Box A \in \Delta^\sigma$, we can just set $\tau = \sigma ::(v:\Box A)$ and $w$ the fresh world variable of the corresponding $\rr \Box $ step.
        \item (Similar to (3)).
        \item (Similar to (4)).
        \item Let $\tau \sqsupseteq \sigma $ and $C \in \Fmla$. By monotonicity we have $v:\forall X A \in \Gamma^\tau$, and so also $v:A[C/X] \in \Gamma^\tau$ by the corresponding $\lr \forall $ step.
        \item If $v:\forall X A \in \Delta^\sigma$, we can just set $\tau = \sigma :: (v:\forall X A)$ and $C$ the propositional eigenvariable of the corresponding $\rr \forall $ step. \qedhere
    \end{enumerate}
\end{proof}

\subsection{A compatible countermodel via possible values}
We continue to adapt the machinery of \cref{sec:completeness-classical} to the intuitionistic setting.
Recall that we have already fixed an unprovable sequent $\sequent^\epsilon_0$ and extended it to the proof search tree $\mathfrak S $, cf.~\cref{dfn:prf-srch-tree}, in particular including the limit sequents $ (\sequent^\sigma)_{\sigma \in T}$.
We duly adapt \cref{dfn:poss-ext-classical} to the intuitionistic setting:

\begin{definition}
    [Possible extensions]
    \label{def:possible-extensions}
    A \textbf{possible extension} of a formula $C$ (wrt $\mathfrak S$) is a family $\possext C = \{\possext  C^\sigma\subseteq \war \}_{\sigma \in T}$ such that:
    \begin{itemize}
        \item $\sigma \sqsubseteq \tau \implies \possext C^\sigma \subseteq \possext C^\tau$; and,
        \item $v:C \in \Gamma^\sigma \implies v \in \possext C^\sigma$; and,
        \item $v:C \in \Delta^\sigma \implies v \notin \possext C^\sigma$.
    \end{itemize}
    Let us write $\possext C \geq_{\mathfrak S} C$ if $\possext C$ is a possible extension of $C$, wrt $\mathfrak S$.
\end{definition}

Note that every formula $C$ still admits a \textbf{minimal (possible) extension} $\minext C$ with $\minext C^\sigma := \{v \in \war \ | \ v:C \in \Gamma^\sigma\}$.
Note that monoticity wrt $\sqsubseteq$ is inherited from monotonicity of LHSs during proof search: if $\sigma \sqsubseteq\tau $ then $ \Gamma^\sigma \subseteq \Gamma^\tau$, and so $\minext C^\sigma \subseteq \minext C^\tau$, for any formula $C$.
Again we obtain a predicate structure whose sets include all possible extensions:

\begin{definition}
    [Countermodel]
    We expand $\pre \PP$ into a structure $\PP$ by including the following missing data:
    \begin{itemize}
        \item The class $\set W$ of modal predicates include all possible extensions all formulas. 
        For each extension $\possext C$, the interpretation $\possext  C^\sigma$ is as defined in \cref{def:possible-extensions}.
        \item We identify each $P \in \Prop$ with the possible extension $\minext P$.
    \end{itemize}
\end{definition}

Once again, we need a key compatibility result:

\begin{lemma}
[Compatibility]
\label{sat-compatible-tree}
    For formulas $A(\vec X)$ and $\vec C$ we have:
    \begin{enumerate}
        \item $v : A(\vec C) \in \Gamma^\sigma \implies \sigma, v \modelsM \PP A(\vec {\possext C})$ whenever $\vec {\possext C} \textends \vec C$.
        \item $ v : A(\vec C) \in \Delta^\sigma \implies \sigma, v \not\modelsM \PP A(\vec{\possext C})$ whenever $\vec {\possext C} \textends \vec C$.
    \end{enumerate}
\end{lemma}

Once again, before proving this, let us point out an immediate consequence.
Write $\sigma \modelsM\PP \rels R \stoup \Gamma \seqar \Delta$ if $\rels R \subseteq \rels R^\sigma$ and, for each $v:A \in \Gamma$ (or $v:A \in \Delta$) we have $\sigma, v \modelsM\PP A$ (or $\sigma,v \not\modelsM\PP A$, respectively). 
As a special case of the \cref{sat-compatible-tree} above we have:
\begin{proposition}
\label{P-is-countermodel}
    $\epsilon \not\modelsM\PP \sequent^\epsilon_0$.
\end{proposition}
\begin{proof}
[Proof of \cref{sat-compatible-tree}]
    By induction on $A(\vec X)$:
    \begin{itemize}
        \item Suppose $A(\vec X) = X$, and fix $C$ and $\possext C \textends C$. We have:
        \[
        \begin{array}{r@{\ \implies \ }ll}
        v:C \in \Gamma^\sigma & v \in \possext C^\sigma & 
        \text{since $\possext C \textends C$} 
            \\
            & \sigma, v \modelsM \PP \possext C & 
            \text{by definition of $\modelsM \PP$}
        \\ \noalign{\medskip}
        v:C \in \Delta^\sigma & v \notin \possext C^\sigma & \text{since $\possext C \textends C$} \\
            & \sigma,v \not\modelsM \PP \possext C    & \text{by definition of $\modelsM\PP$}
        \end{array}
        \]
        \item Suppose $A(\vec X) = A_0(\vec X) \limp A_1 (\vec X)$, and fix $\vec C$ and $\vec {\possext C} \textends \vec C$. We have:
        \[
        \begin{array}{r@{\ \implies \ }ll}
        v:A(\vec C) \in \Gamma^\sigma &  \forall \tau \sqsupseteq\sigma   \left[ v:A_0(\vec C) \in \Delta^\tau \text{ or } v:A_1(\vec C) \in \Gamma^\tau \right] & \text{by \cref{t-induces-partial-val}} \\
            & \forall \tau \sqsupseteq \sigma  \left[ \tau , v \not\modelsM \PP A_0(\vec {\possext C}) \text{ or } \tau, v \modelsM \PP A_1 (\vec {\possext C}) \right] & \text{by IH} \\
            & \sigma, v \modelsM \PP A(\vec {\possext C}) & \text{by definition of $\modelsM \PP$}
        \\ 
        \noalign{\medskip}
        v:A(\vec C) \in \Delta^\sigma
         & \exists \tau \sqsupseteq \sigma \left[ v:A_0(\vec C) \in \Gamma^\tau \text{ and } v:A_1(\vec C) \in \Delta^\tau \right] & \text{by \cref{t-induces-partial-val}} \\
            & \exists \tau \sqsupseteq \sigma \left[ \tau, v \modelsM \PP A_0(\vec{\possext C}) \text{ and } \tau, v \not \modelsM \PP A_1 (\vec {\possext C})   \right] & \text{by IH} \\
            & \sigma , v \not\modelsM \PP A(\vec {\possext C}) & \text{by definition of $\modelsM \PP$}
        \end{array}
        \]
        \item Suppose $A(\vec X) = \Box A'(\vec X)$, and fix $\vec C$ and $\vec {\possext C} \textends \vec C$. We have:
        \[
        \begin{array}{r@{\ \implies \ }ll}
        v:A(\vec C) \in \Gamma^\sigma & \forall \tau \sqsupseteq \sigma \, \forall w \in \war \left[ v\rels R^\tau w \implies w:A'(\vec C) \in \Gamma^\tau \right] & \text{by \cref{t-induces-partial-val}} \\
            & \forall \tau \sqsupseteq \sigma \, \forall w \in \war \left[ v \rels R^\tau w \implies \tau , w \modelsM\PP A'(\vec{\possext C}) \right] & \text{by IH} \\
            & \sigma , v \modelsM \PP A(\vec {\possext C}) & \text{by definition of $\modelsM \PP$} 
            \\ 
            \noalign{\medskip}
            v:A(\vec C) \in \Delta^\sigma & \exists \tau \sqsupseteq \sigma \, \exists w \in \war \left[ v\rels R^\tau w \text{ and } w:A'(\vec C) \in \Delta^\tau \right] & \text{by \cref{t-induces-partial-val}} \\
                & \exists \tau \sqsupseteq \sigma \, \exists w \in \war \left[ v\rels R^\tau w \text{ and } \tau, w \not\modelsM \PP A'(\vec{\possext C}) \right] & \text{by IH} \\
                & \sigma , v \not\modelsM \PP A(\vec{\possext C}) & \text{by definition of $\modelsM\PP$}
        \end{array}
        \]
        \item The case when $A(\vec X) = \blacksquare A'(\vec X)$ is similar to the one above.
        \item Suppose $A(\vec X) = \forall X A'(X,\vec X)$, and fix $\vec C $ and $\vec {\possext C} \textends \vec C$. We have:
        \[
        \begin{array}{r@{\ \implies\ }ll}
        v:A(\vec C) \in \Gamma^\sigma & \forall \tau\sqsupseteq \sigma\, \forall C \in \Fmla \ v:A'(C,\vec C) \in \Gamma^\tau & \text{by \cref{t-induces-partial-val}} \\
            & \forall \tau\sqsupseteq \sigma\, \forall C \in \Fmla\, \forall\,  \possext C \textends C \ \tau , v \modelsM\PP A'(\possext C, \vec{\possext C}) & \text{by IH} \\
            & \sigma, v \modelsM\PP A(\vec{\possext C}) & \text{by definition of $\modelsM\PP$}
            \\
            \noalign{\medskip}
        v:A(\vec C) \in \Delta^\sigma & \exists \tau \sqsupseteq \sigma\, \exists C \in \Fmla \ v:A'(C, \vec C) \in \Delta^\tau & \text{by \cref{t-induces-partial-val}} \\
            & \exists \tau \sqsupseteq \sigma\, \exists C \in \Fmla \ \tau , v \not \modelsM\PP A'(\minext C ,\vec{\possext C}) & \text{by IH} \\
            & \sigma, v \not \modelsM\PP A(\vec{\possext C})  & \text{by definition of $\modelsM \PP$}\qedhere
        \end{array}
        \]
    \end{itemize}
\end{proof}

\subsection{Putting it all together: comprehensivity via a total valuation}
Again, in light of \cref{P-is-countermodel}, for our completeness result it remains to show that $\PP$ is comprehensive.
We take the same impredicative approach as the classical setting.

\begin{definition}
    [Interpreting comprehension]
    Define $\extension C := \{\extension C^\sigma\}_{\sigma \in T} $ where:
$$\extension C^\sigma := \{v \in \war \ | \ \sigma, v \modelsM \PP C\}$$
\end{definition}

Now, \cref{sat-compatible-tree} has another useful consequence:
\begin{proposition}
    $\extension C$ is a possible extension of $C$, with respect to $\mathfrak S$.
\end{proposition}
\begin{proof}
    Simply set $A(\vec X) = C$ and $\vec C = \emptyset$ in \cref{sat-compatible-tree}.
\end{proof}
\begin{corollary}
\label{P-is-comprehensive}
    $\PP$ is comprehensive, and so is a predicate model.
\end{corollary}

We have now established the main result of this section:

\begin{proof}
[Proof of \cref{thm:completeness}]
    By contraposition. Set $\sequent^\epsilon_0 := \cdot \stoup \cdot \seqar v:A$ throughout this section and conclude by \cref{P-is-comprehensive,P-is-countermodel}.
\end{proof}

\section{Simulating labelled proofs axiomatically}\label{sec:labelsoundness}
The goal of this section is to establish the soundness of the labelled sequent calculi $\labIKtSO$ and $\labKtSO$, wrt.\ $\IKtSO$ and $\KtSO$ respectively.
To this end we show (i) each sequent can be interpreted as a formula of our syntax; and (ii) each rule can be interpreted as an admissible rule of the axiomatisation.
The main result of this section is:

\begin{theorem} [Axiomatic soundness]
	\label{prop:axiomsoundness}\label{prop:axiomsoundness-classical}
	We have the following:
	\begin{enumerate}
		\item\label{axiomsoundness-intuitionistic}
		If $\labIKtSO \proves A$ then $\IKtSO \proves A$.
		
		\item\label{axiomsoundness-classical}
		If $\labKtSO \proves A$ then $\KtSO \proves A$.
	\end{enumerate}
\end{theorem}

\subsection{Formula interpretation}

A \textbf{polytree} is a directed acyclic graph whose underlying undirected graph is a tree. 
As a consequence, it is connected, i.e., there exists exactly one path of undirected edges between any pair of distinct nodes.
A \textbf{labelled polytree sequent} is a labelled sequent $\RR\stoup\Gamma \seqar \Delta$ where $\RR$ encodes a labelled polytree and all the labels occurring in $\Gamma$ and $\Delta$ are connected by $\RR$ (as an undirected graph).

By~\cite[Lemma~5.2]{ciabattoni2021display} any derivation in $\labKt$ of a labelled polytree sequent (and a fortiori of a labelled formula) contains only labelled polytree sequents,
which generalises to our systems:

\begin{lemma} \label{lem:tree-like}
	Any labelled polytree sequent provable in $\labKtSO$ (or $\labIKtSO$) has a derivation in $\labKtSO$ (or $\labIKtSO$ respectively) that contains only labelled polytree sequents.
\end{lemma}

Henceforth we will work with only labelled polytree sequents.
This is used in~\cite{ciabattoni2021display} to translate labelled sequents into \emph{nested sequents}, which on the other hand can readily be interpreted into the tense language~\cite{gore2011correspondence}.
We define here a direct interpretation of labelled polytree sequents into tense formulas.

We write $u \overset{\RR}{\leftrightsquigarrow} v$ to indicate that $u$ and $v$ are \emph{connected} in the underlying undirected graph encoded by $\RR$.
If $uRv \in \RR$ (or $vRu\in\RR$), we write
$\RAtMin uv$ for the set of atoms $xRy \in \RR$ such that $x$ is connected to $u$ in $\RR\setminus\{uRv\}$ (or $\RR\setminus\{vRu\}$ respectively).

\begin{definition}[Left interpretation]\label{def:left-interpretation}
	For a label $u$ occuring in $\RR$ or $\Gamma$,
	the \textbf{left (formula) interpretation}
	$\lfm{u}{\RR \stoup \Gamma}$ at $u$ is given as:

	$
	\Band{u:B \in \Gamma}B
	\wedge
	\Band{uRv \in \RR} \Diamond \lfm{v}{\RAtMin vu \stoup \Gamma} 
	\wedge 
	\Band{vRu \in \RR} \Diamondblack \lfm{v}{\RAtMin vu \stoup \Gamma}   
	$
\end{definition}
If $\RR = \emptyset$ or $\Gamma = \emptyset$ this is well-defined,
and we set $\lfm{u}{\cdot \stoup \cdot} = \TOP$.
Note that this translation utilises the definition of $\AND$, $\TOP$, $\DIA$ and $\DIAB$ in terms of second-order quantifiers given in~\cref{eq:so-dfns-of-nonmodal-connectives,eq:diamond-definitions}.

\begin{definition}[Right interpretation]\label{def:rigth-interpretation}
	For a label $u$ occurring in $\RR$ or $\Delta$, 
	the \textbf{right (formula) interpretation} 
	$\rfm{u}{\RR \stoup \Delta}$ at $u$ is: 
	
	$
	\Bor{u:B \in \Delta}B
	\OR
	\Bor{uRv \in \RR} \Box \rfm{v}{\RAtMin vu \stoup \Delta} 
	\OR 
	\Bor{vRu \in \RR} \Boxblack \rfm{v}{\RAtMin vu \stoup \Delta}   
	$
\end{definition}

If $\RR = \emptyset$ or $\Delta = \emptyset$ this is well-defined,
and we set $\rfm{u}{\cdot \stoup \cdot} = \BOT$.
Note that this translation utilises the definition of $\OR$ and $\BOT$ in terms of second-order quantifiers given in~\cref{eq:so-dfns-of-nonmodal-connectives}.

\begin{definition}[Classical formula interpretation]\label{def:classical-interpretation}
	For a label $u$ in $\RR$, $\Gamma$ or $\Delta$, 
	the \textbf{(classical) formula interpretation} at $u$
	$\cfm{u}{\RR \stoup \Gamma \seqar \Delta}$ at $u$ is given as:
	$
	\lfm{u}{\RR \stoup \Gamma} \limp \rfm{u}{\RR \stoup \Delta}$
\end{definition}

In the intuitionistic setting, the asymmetry between the LHS and RHS of a sequent means that the formula interpretation requires a more careful analysis of the polytree structure of $\RR$:

\begin{definition}[Intuitionistic formula interpretation]\label{def:intuitionistic-interpretation}
	For a label $u$ occurring in $\RR$ or $\Gamma$, or for $u=w$, 
	the \textbf{(intuitionistic) formula interpretation}
	$\ifm{u}{\RR \stoup\Gamma \seqar w:A}$ at $u$ is defined as:
	\begin{itemize}
		\item $\lfm{w}{\RR \stoup \Gamma} \to A$, if $u=w$.
		
		\item     
		$\lfm{u}{\RAtMin uv \stoup \Gamma} \to \Box \ifm{v}{ \RAtMin vu \stoup \Gamma \seqar w: A}$, if there exists $v$ such that $v \conn{\RR\setminus\{uRv\}} w$ and $uRv \in \RR$.
		
		\item     
		$\lfm{u}{\RAtMin uv \stoup \Gamma} \to \Boxblack \ifm{v}{ \RAtMin vu \stoup \Gamma \seqar w: A}$,
		if there exists $v$ such that $v \conn{\RR\setminus\{vRu\}} w$ and $vRu \in \RR$.
	\end{itemize}

\end{definition}
\noindent
If $\RR = \emptyset$, this is only defined if $u = w$,
e.g.,
$\ifm{w}{\cdot \stoup \seqar w:A} = A$.

\begin{fact}
	A property of the left interpretation:
	\begin{itemize}
		\item $\lfm{u}{\RR\stoup\Gamma,u:B}
		= \lfm{u}{\RR\stoup\Gamma} \AND B$
		\item and if there is $x$ such that $uRx\in\RR$ and $x\conn{\RR}w$
		$\lfm{u}{\RR\stoup\Gamma,w:B}
		= \lfm{u}{\RAtMin ux\stoup\Gamma,w:B} \AND \DIA\lfm{x}{\RAtMin xu\stoup\Gamma,w:B}$
		\item and if there is $x$ such that $xRu\in\RR$ and $x\conn{\RR}w$
		$\lfm{u}{\RR\stoup\Gamma,w:B}
		= \lfm{u}{\RAtMin ux\stoup\Gamma,w:B} \AND \DIAB\lfm{x}{\RAtMin xu\stoup\Gamma,w:B} $        
	\end{itemize}
\end{fact}

\begin{fact}
	\label{obs:lfm-weaken}\label{obs:weaken}
	If $v$ does not occur in $\RR$:
	\begin{itemize}
		\item $\lfm{u}{\lseq{\rels R}{\Gamma, \fm{v}{A}}} = \lfm{u}{\lseq{\rels R}{\Gamma}}$
		\item $\ifm{u}{\seq{\rels R}{\Gamma, \fm{v}{A}}{w:C}} = \ifm{u}{\seq{\rels R}{\Gamma}{w:C}}$
	\end{itemize}

\end{fact}

\subsection{Formula contexts}

As the formula translation changes due to the position of the succedent formula, we define two kinds of ``formula contexts''. 
These are formulas of a certain shape with one unique atom $\{\}$, called the hole which can be substituted for a formula. The following kind of formula context has the shape of a left-hand-side interpretation.
We call the following a \emph{conjunction context}.
\[
F^\land \{\} ::= \{\} \spa A \land F^\land \{\} \spa \Dmnd F^\land \{\} \spa \tDmnd F^\land \{\}
\]

We write $F^\land \{A \}$ when substituting for the hole in $F^\land \{ \}$ the formula $A$.
We also define the empty substitution $F^\land \{\emptyset \}$ 
as $F^\land \{ \top \}$.
We also write $F^\land$ instead of $F^\land \{ \emptyset \}$.

For the full formula interpretation, we define an \emph{implication context} $F^\to \{\}$:
\[
F^\to \{\} ::= \{\} \spa A \to F^\to \{\} \spa \Box F^\to \{\} \spa \tBox F^\to \{\}
\]
Note that we can translate between these contexts as follows:
\begin{itemize}
\item $F^\land \{\} = \{\}$ iff $F^\to \{\} = \{\}$ 
\item $F^\land \{\} = A \land F^{\land \prime} \{\}$ iff $F^\to \{\} = A \to F^{\to\prime} \{\}$
\item $F^\land \{\} = \Dmnd F^{\land \prime} \{\}$ iff $F^\to \{\} = \Box F^{\to\prime} \{\}$
\item $F^\land \{\} = \tDmnd F^{\land \prime} \{\}$ iff $F^\to \{\} = \tBox F^{\to\prime} \{\}$
\end{itemize}

Substitution for implication contexts $F^\to \{A\}$ is defined by substituting the formula $A$ for the hole (here, we do not need substitution for the empty context).

This allows us to translate these contexts into one another which will be helpful for the rules cut and $\to_l$ which swap the label of the succedent formula.

Here we show that the formula interpretation of a sequent can be written in the form of contexts.

\begin{lemma}\label{lem:fm-ctx}
Let $u$ be a label in the support of~$\rels R$. 
Then, there exists $F$ such that
\begin{enumerate}
	\item $\lfm{u}{\lseq{\rels R}{\Gamma, \fm{v}{A}}} = \fand{A}$; and
	\item $\ifm{u}{\seq{\rels R}{\Gamma}{\fm{v}{A}}} = \fimp{A}$.
\end{enumerate}
\end{lemma}
\begin{proof}
We proceed by induction on the path $u\conn{R}v$

For the base case 
we have $u = v$.

$\lfm{u}{\lseq{\rels R}{\Gamma, \fm{u}{A}}} = \lfm{u}{\RR \stoup \Gamma} \AND A$
and 
$\ifm{u}{\seq{\rels R}{\Gamma}{\fm{u}{A}}} = \lfm{u}{\lseq{\rels R}{\Gamma}} \IMP A$
So we can set $\fand{} = \lfm{u}{\RR \stoup \Gamma} \AND \{\}$ for (1), and therefore $\fimp{} = \lfm{u}{\lseq{\rels R}{\Gamma}} \IMP \{ \}$ for (2).

For the inductive case 
there is a label $x$ such that $uRx$ or $xRu$, and $x\conn{R}v$.

$\lfm{u}{\lseq{\rels R}{\Gamma, \fm{v}{A}}} = \lfm{u}{\RAtMin ux \stoup \Gamma} \AND \Diamonddot
\lfm{x}{\lseq{\RAtMin xu}{\Gamma,v:A}}$
and
$\ifm{u}{\seq{\rels R}{\Gamma}{\fm{w}{B}}} = 
\lfm{u}{\lseq{\RAtMin ux}{\Gamma}} 
\IMP 
\symb \ifm{x}{\seq{\RAtMin xu}{\Gamma}{\fm{w}{B}}}$
where $\symbdiamond = \DIA$ and $\symb = \BOX$ if $uRx\in\RR$, and $\symbdiamond = \DIAB$ and $\symb = \BOXB$ if $xRu\in\RR$.

By the inductive hypothesis there is $F_1$ s.t.~$\lfm{x}{\lseq{\RAtMin xu}{\Gamma, \fm{v}{A}}} = \fand[F_1]{A}$
and
$\ifm{x}{\seq{\RAtMin xu}{\Gamma}{\fm{v}{A}}} = 
\fimp[F_1]{A}$

So we can set
$\Fl{} = \lfm{u}{\RAtMin ux \stoup \Gamma} \AND \Diamonddot\fand[F_1]{}$ for (1), and therefore 
$\fimp{} = \lfm{u}{\lseq{\RAtMin ux}{\Gamma}} \IMP \symb \fimp[F_1]{}$ for (2).
\end{proof}

\begin{lemma}\label{lem:fm-ctx-2}
Given a sequent $\seq{\rels R}{\Gamma, \fm{v}{A}}{\fm{w}{B}}$, with label $u$ in the support of~$\rels R$. 
Then, there exist $F_1, F_2, F_3$ such that 
\begin{enumerate}
	\item $\ifm{u}{\seq{\rels R}{\Gamma, \fm{v}{A}}{\fm{w}{B}}} = \fimp[F_1]{\fand[F_2]{A} \IMP \fimp[F_3]{B}}$; and
	
	\item
	$\pprove{\ifm{u}{\seq{\rels R}{\Gamma, \fm{w}{B}}{\fm{v}{A}}} \IFF \fimp[F_1]{\fand[F_3]{B} \IMP \fimp[F_2]{A}}}$.
\end{enumerate}
\end{lemma}
\begin{proof}   
Proof of (1):
We proceed by induction on the path $u\conn{\RR}w$.
\begin{itemize}
	\item $u = w$:
	$\ifm{w}{\rels R \stoup \Gamma, v:A \seqar w:B}
	= \lfm{w}{\rels R \stoup \Gamma, v:A} \IMP B$
	
	Using~\cref{lem:fm-ctx}(1), there is $F_2$ s.t.~$\lfm{u}{\lseq{\rels R}{\Gamma, \fm{v}{A}}} = \fand[F_2]{A}$.
	So,
	$$\ifm{w}{\rels R \stoup \Gamma, v:A \seqar w:B}
	= \Fl[2]{A} \IMP B$$
	
	\item there is a label $x$ such that $uRx$ or $xRu\in\RR$ and $x \conn{\RR} w$: then, there are two possible cases:
	
	\begin{itemize}
		\item $v\in\RAtMin ux$: then, by~\cref{obs:weaken}
		\begin{multline*}
			\ifm{u}{\seq{\rels R}{\Gamma, \fm{v}{A}}{\fm{w}{B}}}
			\\
			= 
			\lfm{u}{\lseq{\RAtMin ux}{\Gamma, \fm{v}{A}}} 
			\IMP 
			\symb \ifm{x}{\seq{\RAtMin xu}{\Gamma}{\fm{w}{B}}}
		\end{multline*}
		where $\symb = \BOX$ if $uRx \in \RR$ and $\symb = \BOXB$ if $xRu\in\RR$.
		
		Using~\cref{lem:fm-ctx}(1), there is $F_2$ s.t.~$\lfm{u}{\lseq{\RAtMin ux}{\Gamma, \fm{v}{A}}} = \fand[F_2]{A}$.
		
		Using~\cref{lem:fm-ctx}(2), there is $F'_3$ s.t.~$\ifm{w_1}{\seq{\RAtMin xu}{\Gamma}{\fm{w}{B}}} = \fimp[{F'}_3]{B}$.
		Set $\fimp[F_3]{\ } = \symb \fimp[{F'}_3]{\ }$ 
		and we get
		$$
		\ifm{u}{\seq{\rels R}{\Gamma, \fm{v}{A}}{\fm{w}{B}}} = \fand[F_2]{A} \IMP \fimp[F_3]{B}
		$$
		
		\item $v\in\RAtMin xu$: then, by~\cref{obs:lfm-weaken}
		\begin{multline*}
			\ifm{u}{\seq{\rels R}{\Gamma, \fm{v}{A}}{\fm{w}{B}}}
			\\
			=
			\lfm{u}{\lseq{\RAtMin ux}{\Gamma}} 
			\IMP 
			\symb \ifm{x}{\seq{\RAtMin xu}{\Gamma, \fm{v}{A}}{\fm{w}{B}}}
		\end{multline*}
		where $\symb = \BOX$ if $uRx \in \RR$ and $\symb = \BOXB$ if $xRu\in\RR$.
		
		By the inductive hypothesis, there exist $F_1', F_2, F_3$ s.t.
		$$
		\ifm{x}{\seq{\RAtMin xu}{\Gamma, \fm{v}{A}}{\fm{w}{B}}} = \fimp[F_1']{\fand[F_2]{A} \IMP \fimp[F_3]{B}}  $$
		
		Set $\fimp[F_1]{\ } = \lfm{u}{\lseq{\RAtMin ux}{\Gamma}} \IMP \symb \fimp[F_1']{ \ }$ and we get
		$$
		\ifm{u}{\seq{\RR}{\Gamma, \fm{v}{A}}{\fm{w}{B}}} = \fimp[F_1]{\fand[F_2]{A} \IMP \fimp[F_3]{B}}
		$$
	\end{itemize}
	
\end{itemize}

Proof of (2):
We proceed by induction on the path $u\conn{\RR}w$.
\begin{itemize}
	\item $u = w$:

	\begin{itemize}
		\item either $u=v$:
		\begin{flalign*}
			\ifm{u}{\seq{\rels R}{\Gamma, \fm{u}{A}}{\fm{u}{B}}} 
			&= 
			\lfm{u}{\lseq{\RR}{\Gamma, u:A}}
			\IMP B 
			\\
			&=
			(\lfm{u}{\lseq{\RR}{\Gamma}} \AND A) 
			\IMP B
			\\
			\ifm{u}{\RR \stoup \Gamma, u:B \seqar u:A}
			&=
			\lfm{u}{\RR \stoup \Gamma, u:B} \IMP A
			\\
			&=
			(B \AND \lfm{u}{\RR \stoup\Gamma}) \IMP A
			\\
			&
			\IFF
			B \IMP \lfm{u}{\RR \stoup\Gamma} \IMP A
		\end{flalign*}

		\item or there is a label $y$ such that $uRy$ or $yRu$, and $y\conn{\RR}v$
		
		\begin{multline*}
			\ifm{u}{\seq{\rels R}{\Gamma, \fm{v}{A}}{\fm{u}{B}}} 
			= 
			\lfm{u}{\lseq{\RAtMin ux}{\Gamma, \fm{v}{A}}} 
			\IMP B
			\\
			=
			(\lfm{u}{\lseq{\RAtMinMin uxy}{\Gamma}} \AND \Diamondtimes\lfm{y}{\RAtMin yu\stoup\Gamma, v:A}) 
			\IMP B
		\end{multline*}
		\begin{multline*}
			\ifm{u}{\RR \stoup \Gamma, w:B \seqar v:A}
			\\ =
			\lfm{u}{\lseq{\RAtMin uy}{\Gamma, \fm{w}{B}}} 
			\IMP 
			\Boxplus \ifm{y}{\seq{\RAtMin yu}{\Gamma}{\fm{v}{A}}}
			\\
			=
			(\lfm{u}{\lseq{\RAtMinMin uxy}{\Gamma}} \AND B)
			\IMP 
			\Boxplus \ifm{y}{\seq{\RAtMin yu}{\Gamma}{\fm{v}{A}}}
			\\
			\IFF
			B \IMP \lfm{u}{\lseq{\RAtMinMin uxy}{\Gamma}} \IMP
			\Boxplus \ifm{y}{\seq{\RAtMin yu}{\Gamma}{\fm{v}{A}}}
		\end{multline*}
		where $\Boxplus = \BOX$ if $uRy \in \RR$ and $\Boxplus = \BOXB$ if $yRu\in\RR$.
		
	\end{itemize}

	\item there is a label $x$ such that $uRx$ or $xRu\in\RR$ and $x \conn{\RR} w$ and $v\in\RAtMin ux$: 
	
	\begin{itemize}
		\item either $u=v$:
		\begin{multline*}
			\ifm{v}{\seq{\rels R}{\Gamma, \fm{v}{A}}{\fm{w}{B}}} 
			\\= 
			\lfm{v}{\lseq{\RAtMin vx}{\Gamma, v:A}}
			\IMP 
			\symb \ifm{x}{\seq{\RAtMin xv}{\Gamma}{\fm{w}{B}}} 
			\\
			=
			(\lfm{v}{\lseq{\RAtMin vx}{\Gamma}} \AND A) 
			\IMP 
			\symb \ifm{x}{\seq{\RAtMin xv}{\Gamma}{\fm{w}{B}}} 
		\end{multline*}
		\begin{multline*}
			\ifm{v}{\RR \stoup \Gamma, w:B \seqar v:A}
			\\=
			\lfm{v}{\RR \stoup \Gamma, w:B} \IMP A
			\\
			=
			(\Diamonddot\lfm{x}{\RAtMin xv\stoup\Gamma,w:B} \AND \lfm{v}{\RAtMin vx\stoup\Gamma}) \IMP A
			\\
			\IFF
			\Diamonddot\lfm{x}{\RAtMin xv\stoup\Gamma,w:B} \IMP \lfm{v}{\RAtMin vx\stoup\Gamma} \IMP A
		\end{multline*}

		\item or there is a label $y$ such that $uRy$ or $yRu$, and $y\conn{\RR}v$

		\begin{multline*}
			\ifm{u}{\seq{\rels R}{\Gamma, \fm{v}{A}}{\fm{w}{B}}} 
			\\ = 
			\lfm{u}{\lseq{\RAtMin ux}{\Gamma, \fm{v}{A}}} 
			\IMP 
			\symb \ifm{x}{\seq{\RAtMin xu}{\Gamma}{\fm{w}{B}}}
			\\
			=
			(\lfm{u}{\lseq{\RAtMinMin uxy}{\Gamma}} \AND \Diamondtimes\lfm{y}{\RAtMin yu\stoup\Gamma, v:A}) 
			\\
			\IMP 
			\symb \ifm{x}{\seq{\RAtMin xu}{\Gamma}{\fm{w}{B}}}
		\end{multline*}
		\begin{multline*}
			\ifm{u}{\RR \stoup \Gamma, w:B \seqar v:A}
			\\ =
			\lfm{u}{\lseq{\RAtMin uy}{\Gamma, \fm{w}{B}}} 
			\IMP 
			\Boxplus \ifm{y}{\seq{\RAtMin yu}{\Gamma}{\fm{v}{A}}}
			\\
			=
			(\lfm{u}{\lseq{\RAtMinMin uxy}{\Gamma}} \AND \Diamonddot\lfm{x}{\lseq{\RAtMin xu}{\Gamma, \fm{w}{B}}})
			\\
			\IMP 
			\Boxplus \ifm{y}{\seq{\RAtMin yu}{\Gamma}{\fm{v}{A}}}
			\\
			\IFF
			\Diamonddot\lfm{x}{\lseq{\RAtMin xu}{\Gamma, \fm{w}{B}}}
			\\
			\IMP \lfm{u}{\lseq{\RAtMinMin uxy}{\Gamma}} \IMP
			\Boxplus \ifm{y}{\seq{\RAtMin yu}{\Gamma}{\fm{v}{A}}}
		\end{multline*}
		where $\symb = \BOX$ if $uRx \in \RR$ and $\symb = \BOXB$ if $xRu\in\RR$, and
		where $\Boxplus = \BOX$ if $uRy \in \RR$ and $\Boxplus = \BOXB$ if $yRu\in\RR$.
		
	\end{itemize}
	
	\item there is a label $x$ such that $uRx$ or $xRu\in\RR$ and $x \conn{\RR} w$ and $v\in\RAtMin xu$: 
	\begin{multline*}
		\ifm{u}{\seq{\rels R}{\Gamma, \fm{v}{A}}{\fm{w}{B}}} 
		\\=
		\lfm{u}{\lseq{\RAtMin ux}{\Gamma}} 
		\IMP 
		\symb \ifm{x}{\seq{\RAtMin xu}{\Gamma, \fm{v}{A}}{\fm{w}{B}}}
		\\
		=
		\lfm{u}{\lseq{\RAtMin ux}{\Gamma}} \IMP \symb \fimp[F_1']{\fand[F_2]{A} \IMP \fimp[F_3]{B}}
	\end{multline*}
	where $\symb = \BOX$ if $uRx \in \RR$ and $\symb = \BOXB$ if $xRu\in\RR$.

	By inductive hypothesis, 
	$\ifm{x}{\seq{\RAtMin xu}{\Gamma, \fm{w}{B}}{\fm{v}{A}}}$
	is equivalent to
	$\fimp[F_1']{\fand[F_3]{B} \IMP \fimp[F_2]{A}}$, so:
	\begin{multline*}
		\ifm{u}{\RR \stoup \Gamma, w:B \seqar v:A}
		\\ = 
		\lfm{u}{\lseq{\RAtMin ux}{\Gamma}} 
		\IMP 
		\symb \ifm{x}{\seq{\RAtMin xu}{\Gamma, \fm{w}{B}}{\fm{v}{A}}}
		\\
		\IFF
		\lfm{u}{\lseq{\RAtMin ux}{\Gamma}} 
		\IMP 
		\symb \fimp[F_1']{\fand[F_3]{B} \IMP \fimp[F_2]{A}}
	\end{multline*}
	\qedhere
\end{itemize}

\end{proof}

\subsection{Technical lemmas}

The following lemmas capture the behaviour of contexts which we need for proving the soundness of the sequent rules.
Note that we can view most of these lemmas as generalisations of modal rules and axioms such as necessitation or functoriality of modalities.

		\begin{lemma} \label{lem:GenBoxNec}
If $\IKtSO \vdash A$ then $\IKtSO \vdash \fimp{A}$.
\end{lemma}
\begin{proof}
We proceed by induction over the structure of $\fimp$. 

The base case $\fimp{} = \{\}$ follows immediately.

If $\fimp{A} = B \to \fand[F_1]{A}$:
By inductive hypothesis, $\IKtSO \vdash \fand[F_1]{A}$.
By $\mp$ on the axiom $\fimp[F_1]{A} \IMP B \IMP \fimp[F_1]{A}$: 
$\IKtSO \vdash B \to \fimp[F_1]{A}$.

If $\fimp{A} = \Box \fimp[F_1]{A}$ or $\tBox \fimp[F_1]{A}$:
By inductive hypothesis, $\IKtSO \vdash \fimp[F_1]{A}$.
By $\necw$ or $\necb$, $\IKtSO \vdash \Box\fimp[F_1]{A}$ or $\tBox\fimp[F_1]{A}$ respectively.
\end{proof}

\begin{lemma} \label{lem:GenKBOX}
$\IKtSO \vdash \fimp{A \IMP B} \IMP \fimp{A} \IMP \fimp{B}$
\end{lemma}
\begin{proof}
We proceed by induction on the structure of $\fimp{}$.

The base case $\fimp{} = \{ \}$ follows from the fact that $(A \IMP B) \IMP A \IMP B$ is a theorem of $\IPL$.

We have 3 inductive cases:
\begin{itemize}
	\item $\fimp{} = C \IMP \fimp[F_1]{}$.
	By inductive hypothesis $\pprove{\fimp[F_1]{A \IMP B}  \IMP \fimp[F_1]{A} \IMP \fimp[F_1]{B}}$.
	So propositional reasoning gives 
	$\pprove{(C \IMP \fimp[F_1]{A \IMP B}) \IMP (C \IMP \fimp[F_1]{A}) \IMP (C \IMP \fimp[F_1]{B})}$.
	
	\item $\fimp{} = \BOX \fimp[F_1]{}$.
	By the inductive hypothesis, $\pprove{\fimp[F_1]{A \IMP B}  \IMP \fimp[F_1]{A} \IMP \fimp[F_1]{B}}$.
	Apply $\necw$ to get $\pprove{\BOX (\fimp[F_1]{A \IMP B}  \IMP \fimp[F_1]{A} \IMP \fimp[F_1]{B})}$.
	By $\functwb$ and $\mp$ $\pprove{\BOX \fimp[F_1]{A \IMP B} \IMP \BOX \fimp[F_1]{A} \IMP \BOX \fimp[F_1]{B}}$.

	\item $\fimp{} = \BOXB \fimp[F_1]{}$.
	By the inductive hypothesis $\pprove{\fimp[F_1]{A \IMP B}  \IMP \fimp[F_1]{A} \IMP \fimp[F_1]{B}}$.
	Apply $\necb$ to get $\pprove{\BOXB (\fimp[F_1]{A \IMP B}  \IMP \fimp[F_1]{A} \IMP \fimp[F_1]{B})}$.
	By $\functbb$ and $\mp$ $\pprove{\BOXB \fimp[F_1]{A \IMP B} \IMP \BOXB \fimp[F_1]{A} \IMP \BOXB \fimp[F_1]{B}}$.  \qedhere
\end{itemize}
\end{proof}

\begin{lemma} \label{lem:GenKDIA}
$\IKtSO \vdash \fimp{A \IMP B} \IMP \fand{A} \IMP \fand{B}$
\end{lemma}
\begin{proof}
We proceed by induction on the structure of $\fand{}$.

The base case $\fand{} = \{\}$ is follows from the fact that $(A \IMP B) \IMP A \IMP B$ is a theorem of IPL.

We have 3 inductive cases:
\begin{itemize}
	\item $\fand{} = C \AND \fand[F_1]{}$ hence $\fimp{} = C \IMP \fimp[F_1]{}$.
	By the inductive hypothesis, 
	$\IKtSO \vdash \fimp[F_1]{A \IMP B} \IMP \fand[F_1]{A} \IMP \fand[F_1]{B}$
	and so propositional reasoning gives 
	$\pprove{(C \IMP \fimp[F_1]{A \IMP B}) \IMP (C \AND \fand[F_1]{A}) \IMP (C \AND \fand[F_1]{B})}$.
	
	\item $\fand{} = \DIA \fand[F_1]{}$ and $\fimp{} = \BOX \fimp[F_1]{}$.
	By the inductive hypothesis, 
	$\IKtSO \vdash \fimp[F_1]{A \IMP B} \IMP \fand[F_1]{A} \IMP \fand[F_1]{B}$.
	Apply $\necw$ to get $\pprove{\BOX (\fimp[F_1]{A \IMP B} \IMP \fand[F_1]{A} \IMP \fand[F_1]{B})}$.
	By $\functwb$, $\functwd$ and $\mp$ $\pprove{\BOX \fimp[F_1]{A \IMP B} \IMP \DIA\fand[F_1]{A} \IMP \DIA\fand[F_1]{B}}$ 

	\item $\fand{} = \DIAB \fand[F_1]{}$ and $\fimp{} = \BOXB \fimp[F_1]{}$.
	By the inductive hypothesis, 
	$\IKtSO \vdash \fimp[F_1]{A \IMP B} \IMP \fand[F_1]{A} \IMP \fand[F_1]{B}$
	Apply $\necb$ to get $\pprove{\BOXB (\fimp[F_1]{A \IMP B} \IMP \fand[F_1]{A} \IMP \fand[F_1]{B})}$.
	By $\functbb$, $\functbd$ and $\mp$ $\pprove{\BOXB \fimp[F_1]{A \IMP B} \IMP \DIAB \fand[F_1]{A} \IMP \DIAB \fand[F_1]{B}}$. \qedhere
	\end{itemize}
\end{proof}

\begin{lemma} \label{lem:GenFSDiaBox}
$\IKtSO \vdash \fand{A \IMP B} \IMP \fimp{A} \IMP \fand{B}$.
\end{lemma}
\begin{proof}
We proceed by induction on the structure of $\fand{}$.

The base case $\fand{} = \{ \}$ follows from the fact that $(A \IMP B) \IMP A \IMP B$ is a theorem of $\IPL$.

We have 3 inductive cases: %
\begin{itemize}
	\item $\fand{} = C \AND \fand[F_1]{}$ hence $\fimp{} = C \IMP \fimp[F_1]{}$
	By the inductive hypothesis, $\pprove{\fand[F_1]{A \IMP B} \IMP \fimp[F_1]{A} \IMP \fand[F_1]{B}}$, 
	and so propositional reasoning gives 
	$\pprove{(C \AND \fand[F_1]{A \IMP B}) \IMP (C \IMP \fimp[F_1]{A}) \IMP (C \AND \fimp[F_1]{B})}$.
	
	\item $\fand{} = \DIA \fand[F_1]{}$ hence $\fimp{} = \BOX \fimp[F_1]{}$.
	By the inductive hypothesis, $\pprove{\fand[F_1]{A \IMP B} \IMP \fimp[F_1]{A } \IMP \fand[F_1]{B}}$.
	Apply $\necw$ to get $\pprove{\BOX (\fand[F_1]{A \IMP B} \IMP \fimp[F_1]{A } \IMP \fand[F_1]{B})}$.
	
	By axioms $\functwd$ and $\mp$, 
	$\pprove{\DIA \fand[F_1]{A \IMP B} \IMP \DIA(\fimp[F_1]{A } \IMP \fand[F_1]{B})}$.
	With~\cref{prop:IK-thm}(5), transitivity of implication gives 
	$\pprove{\DIA \fimp[F_1]{A \IMP B} \IMP \BOX \fimp[F_1]{A} \IMP \DIA \fimp[F_1]{B}}$.

	\item $\fand{} = \DIAB \fand[F_1]{}$ hence $\fimp{} = \BOXB \fimp[F_1]{}$
	By the inductive hypothesis, $\pprove{\fand[F_1]{A \IMP B} \IMP \fimp[F_1]{A } \IMP \fand[F_1]{B}}$.
	Apply $\necb$ to get $\pprove{\BOXB (\fand[F_1]{A \IMP B} \IMP \fimp[F_1]{A } \IMP \fand[F_1]{B})}$.
	By axioms $\functbd$ and $\mp$, 
	$\pprove{\DIAB \fand[F_1]{A \IMP B} \IMP \DIAB(\fimp[F_1]{A } \IMP \fand[F_1]{B})}$.
	With~\cref{prop:IK-thm}(6), transitivity of implication gives 
	$\pprove{\DIAB \fimp[F_1]{A \IMP B} \IMP \BOXB \fimp[F_1]{A} \IMP \DIAB \fimp[F_1]{B}}$. \qedhere
\end{itemize}
\end{proof}

\begin{lemma} \label{lem:GenInterDiaBox}
$\IKtSO \vdash (\fand{A} \IMP \fimp{B}) \IMP \fimp{A \IMP B} $.
\end{lemma}
\begin{proof}
We proceed by induction on the structure of $\fand{}$.

The base case $\fand{} = \{ \}$ follows from the fact that 
$(A \IMP B) \IMP A \IMP B$ %
is a theorem of $\IPL$.

We have 3 inductive cases:
\begin{itemize}
	\item $\fand{} = C \AND \fand[F_1]{}$ hence $\fimp{} = C \IMP \fimp[F_1]{}$.
	By the inductive hypothesis $\pprove{(\fand[F_1]{A} \IMP \fimp[F_1]{B}) \IMP \fimp[F_1]{A \IMP B}}$ 

	and propositional reasoning gives 
	$\pprove{((C \AND \fand[F_1]{A}) \IMP (C \IMP \fimp[F_1]{B})) \IMP (C \IMP \fimp{A \IMP B})}$.
	
	\item $\fand{} = \DIA \fand[F_1]{}$ and $\fimp{} = \BOX \fimp[F_1]{}$.
	
	By the inductive hypothesis $\pprove{(\fand[F_1]{A} \IMP \fimp[F_1]{B}) \IMP \fimp{A \IMP B}}$.
	Applying $\necw$, $\pprove{\BOX (\fand[F_1]{A} \IMP \fimp[F_1]{B} \IMP \fimp[F_1]{A \IMP B}})$.
	Then, $\functwb$ and $\mp$ gives $\pprove{\BOX (\fand[F_1]{A} \IMP \fimp[F_1]{B}) \IMP \BOX \fimp[F_1]{A \IMP B}}$.
	With~\cref{prop:IK-thm}(3), transitivity of implication gives

	$\pprove{(\DIA \fand[F_1]{A} \IMP \BOX \fimp[F_1]{A}) \IMP \BOX \fimp[F_1]{A \IMP B}}$.

	\item $\fand{} = \DIAB \fand[F_1]{}$ hence $\fimp{} = \BOXB \fand[F_1]{}$.

	By the inductive hypothesis $\pprove{(\fand[F_1]{A} \IMP \fimp[F_1]{B}) \IMP \fimp{A \IMP B}}$.
	Applying $\necb$, $\pprove{\BOXB (\fand[F_1]{A} \IMP \fimp[F_1]{B} \IMP \fimp[F_1]{A \IMP B}})$.
	Then, $\functbb$ and $\mp$ gives $\pprove{\BOXB (\fand[F_1]{A} \IMP \fimp[F_1]{B}) \IMP \BOXB \fimp[F_1]{A \IMP B}}$.
	With~\cref{prop:IK-thm}(4), transitivity of implication gives 
	$\pprove{(\DIAB \fandempt[F_1] \IMP \BOXB \fimp{A}) \IMP \BOXB \fimp[F_1]{A}}$. \qedhere
	%
	%
\end{itemize}
\end{proof}

\begin{lemma} \label{lem:GenGen}
$\IKtSO \vdash \forall X \fimp{A} \to  \fimp{\forall X A}$
whenever $X$ does not occur freely in $\fimp{}$. 
\end{lemma}
\begin{proof}
We proceed by induction on $\fimp{}$.
The base case $\fimp{} = \{\}$ follows immediately.

We have 3 inductive cases:
\begin{itemize}
	\item $\fimp{} = C \to \fimp[F_1]{}$ with $X \notin \fv C$: 
	The inductive hypothesis gives us $\IKtSO \vdash  \forall X \fimp[F_1]{A} \IMP \fimp[F_1]{\forall X A}$, 
	By combining $\functforall$ and $\Vacax$ 
	$\IKtSO \vdash \forall X ( C \to \fimp[F_1]{A} ) \to C \to \forall X \fimp[F_1]{A}$.
	which by transitivity of implication gives 
	$\IKtSO \vdash \forall X ( C \to \fimp[F_1]{A} ) \to C \to  \fimp[F_1]{ \forall X A}$.
	
	\item $\fimp{} = \Box \fimp[F_1]{}$: 
	By inductive hypothesis we have $\IKtSO \vdash \forall X \fimp[F_1]{A} \to \fimp[F_1]{\forall X A} $. 
	Applying $\necw$ and $\functwb$, we have $\IKtSO \vdash \Box \forall X \fimp[F_1]{A} \to \Box \fimp[F_1]{\forall X A}$.
	With~\cref{ex:Barcan}, transitivity of implication gives 
	$\IKtSO \vdash \forall X \Box \fimp[F_1]{ A} \to \Box \fimp[F_1]{\forall X A}$.

	\item $\fimp{} = \tBox \fimp[F_1]{}$: 
	By inductive hypothesis we have $\IKtSO \vdash \forall X \fimp[F_1]{A} \to \fimp[F_1]{\forall X A} $. 
	Applying $\necb$ and $\functbb$, we have $\IKtSO \vdash \tBox \forall X \fimp[F_1]{A} \to \tBox \fimp[F_1]{\forall X A}$.
	With~\cref{ex:forall-distributes-over-box}(4), transitivity of implication gives 
	$\IKtSO \vdash \forall X \tBox \fimp[F_1]{ A} \to \tBox \fimp[F_1]{\forall X A}$. \qedhere
\end{itemize}
\end{proof}

\subsection{Axiomatic soundness for $\labIKtSO$}

To prove~\cref{prop:axiomsoundness}\eqref{axiomsoundness-intuitionistic}, we will establish a stronger statement:

\begin{lemma}[Soundness of interpretation]
\label{lem:sequent-soundness-intuitionistic}
If $\labIKtSO \proves \RR\stoup\Gamma\seqar w:A$ then $\IKtSO \proves \ifm{u}{\RR\stoup\Gamma\seqar w:A}$, for~$u$ occurring in $\rels R$ or~$u=w$.
\end{lemma}

This is proved by induction on the proof of $\RR\stoup\Gamma\seqar w:A$ in $\labIKtSO$, where both the base case and the inductive cases are provided by the following lemma.
Recall that we can assume that any sequent $\RR\stoup\Gamma\seqar w:A$ occurring in a proof is polytree labelled sequents; in particular labels in $\Gamma\cup \{w:A\}$ are connected by $\RR$.

\begin{lemma} [Local soundness]
\label{lem:local-soundness}
Let 
$\vlinf{}{}{\RR\stoup\Gamma\seqar x:A}{\{\RR\stoup\Gamma_i\seqar x_i:A_i\}_{i< n}}$
be a  rule instance of $\labIKtSO$ with $n= 0$, $1$ or $2$,
and let $u\in\RR$ or $u=x$.

If $\IKtSO \vdash \ifm{u}{\RR\stoup\Gamma_i\seqar x_i:A_i}$ for all $i< n$
then $\IKtSO \vdash \ifm{u}{\RR\stoup\Gamma\seqar x:A}$.
\end{lemma}

For most rules of $\labIKtSO$, the proof is similar to the one for \emph{nested sequents} for $\IK$~\cite{strassburger2013cut}, but over the current formula interpretation. 
We highlight some of the technicalities of the proof.

The rule $\lr\limp$ requires some attention in the way the context is interpreted into a formula
because the two premisses can have different labels on the RHS.
The RHS label determines how the formula interpretation is defined~(\cref{def:intuitionistic-interpretation}), in particular whether a relational atom $xRy\in\RR$ is read as a $\BOX$ (when $xRy$ belongs to the path from $u$ to $w$) or a $\DIA$ (otherwise).
For this reason, the interaction axioms $\diaimpbox: (\Diamond A \limp \Box B) \limp \Box (A\limp B)$ and $\diaimpboxb: (\Diamondblack A \limp \Boxblack B) \limp \Boxblack (A\limp B)$ are required in the axiomatic proof simulating $\lr\limp$.

The rule $\lr\Box$ is more subtle to handle in the tense case than it usually is in the modal case
because the indicated relational atom $vRw$ in the context can be read forwards or backwards when the formula interpretation is computed.
For this reason, the adjunction axioms $\bdwb: \DIAB \BOX A \IMP A$ and $\wdbb: \DIA \BOXB A \IMP A$ are needed in the axiomatic proof simulating $\lr\Box$.

The rule $\rr\forall$ also displays some interesting interaction with the modalities when read through the formula interpretation. 
In particular the distributivity axiom $\forall X \Box A \limp \Box \forall X A$ (\cref{ex:forall-distributes-over-box}) is required in the axiomatic proof simulating $\rr\forall$.

The following lemmas (\cref{lem:sound-id} to \cref{cor:sound-cut}) provide the proof of \cref{lem:local-soundness} via a case by case analysis of the rules of $\labIKtSO$.

\begin{lemma}[Soundness of $\id$] 
\label{lem:sound-id} 
$\IKtSO \vdash \ifm{u}{\RR \stoup  v : A \seqar v : A}$
\end{lemma}
\begin{proof}
We can write $\ifm{u}{\RR \stoup v:A \seqar v:A} = \Fr{A \IMP A}$ by~\cref{lem:fm-ctx-2}
Clearly, $A \to A$ is a theorem of $\IPL$ and thus of $\IKtSO$. 
By~\cref{lem:GenBoxNec}, $\IKtSO \vdash F^\to \{ A  \to A\}$.
\end{proof}

\begin{lemma}[Soundness of $\lr\wk$]
\label{lem:sound-weak}
If $\IKtSO \vdash \ifm{u}{\RR \stoup \Gamma \seqar w:C}$
then $\IKtSO \vdash \ifm{u}{\RR \stoup \Gamma , v : A \seqar w:C}$
\end{lemma}
\begin{proof}
We can write 
$\ifm{u}{\RR \stoup \Gamma \seqar w:C} =
\Fr[1]{ \Fl[2]{\TOP} \IMP \Fr[3]{C}}$ 
and %
$\ifm{u}{\RR \stoup \Gamma , v : A \seqar w:C} =
\Fr[1]{ \Fl[2] { A } \IMP \Fr[3] { C}}$
by~\cref{lem:fm-ctx-2}.

$\IPL \vdash A \IMP \TOP$ so by~\cref{lem:GenBoxNec,lem:GenKDIA} and $\mp$
we have $\IKtSO \vdash \Fl[2]{A} \IMP \Fl[2]{\TOP}$. 

$\IPL \vdash (A \to B) \to (B \to C) \to (A \to C)$ so by substitution
$\IKtSO \vdash ( \Fl[2]{A} \IMP \Fl[2]{\TOP}) \IMP ( \Fl[2]{\TOP} \to \Fr[3]{C}) \IMP (\Fl[2]{ A} \IMP \Fr[3]{B})$. 

By~\cref{lem:GenBoxNec,lem:GenKBOX} and $\mp$ 

$\IKtSO \vdash \Fr[1]{ \Fl[2]{\TOP} \IMP \Fr[3]{C}} \IMP \Fr[1]{ \Fl[2]{A } \IMP \Fr[3]{C}}$. 
\end{proof}

\begin{lemma}[Soundness of $\lr\cntr$]\label{lem:sound-cntr}
If $\IKtSO \vdash \ifm{u}{\RR \stoup \Gamma, v : A, v : A \seqar w:C}$
then $\IKtSO \vdash \ifm{u}{\RR \stoup \Gamma , v : A \seqar w:C}$
\end{lemma}
\begin{proof}
We can write
$\ifm{u}{\RR \stoup \Gamma, v : A, v : A \seqar w:C} =
\Fr[1] { \Fl[2] {A \land A} \IMP \Fr[3]{B}}$ 
while %
$\ifm{u}{\RR \stoup \Gamma , v : A \seqar w:C} =
F_1^\to \{ F_2^\land \{A\} \to F_3^\to\{B\}\}$
by~\cref{lem:fm-ctx-2}.

As $\IPL \vdash A \to (A \land A)$, 
by~\cref{lem:GenBoxNec,lem:GenKDIA} and $\mp$ 
$\IKtSO \vdash F_2^\land \{A \} \to F_2^\land \{A \land A\}$. 

$\IPL \vdash (A \to B) \to (B \to C) \to (A \to C)$ 
hence also 
$\IKtSO \vdash (\Fl[2]{A} \IMP \Fl[2]{A \land A}) \IMP ( \Fl[2]{A \land A} \IMP \Fr[3]{C}) \IMP (\Fl[2]{A} \IMP \Fr[3]{C})$. 

By~\cref{lem:GenBoxNec,lem:GenKBOX} and $\mp$ we have 

$\IKtSO \vdash  \Fr[1] {  \Fl[2] {A \land A}  \IMP \Fr[2]{B} } \IMP \Fr[1] { \Fl[2] {A} \IMP \Fr[3]{B}}$. 
\end{proof}

\begin{lemma}[Soundness of $\rr\limp$]\label{lem:sound-impr}
If $\IKtSO \vdash \ifm{u}{\RR \stoup \Gamma, v:A \seqar v:B}$
then $\IKtSO \vdash \ifm{u}{\RR \stoup \Gamma \seqar v : A \limp B}$
\end{lemma}
\begin{proof}
We can write %
$\ifm{u}{\RR \stoup \Gamma, v:A \seqar v:B} =
\Fr[1]{(A \land \Fl[2]{\TOP} ) \limp B }$ 
while %
$\ifm{u}{\RR \stoup \Gamma \seqar v : A \limp B} =
\Fr[1]{\Fl[2]{\TOP} \limp A \limp B }$
by~\cref{lem:fm-ctx-2}.

By $\IPL\vdash ((A \land C) \to B) \to  C \to A \to B$, we have $\IKtSO \vdash ((\Fl[2]{\TOP} \land A) \to B ) \to \Fl[2]{\TOP}  \to A \to B$, 

And by~\cref{lem:GenBoxNec,lem:GenKBOX} and $\mp$ 

$\IKtSO \vdash \Fr[1]{(\Fl[2]{\TOP} \land A) \to B }\to \Fr[1] {\Fl[2]{\TOP}  \to A \to B}$. 
\end{proof}

\begin{lemma}[Soundness of $\lr \forall$]\label{lem:sound-fal}
If $\IKtSO \vdash \ifm{u}{\RR \stoup \Gamma, v : A [B/X] \seqar w:C}$
then $\IKtSO \vdash \ifm{u}{\RR \stoup \Gamma, v : \forall X A\seqar w:C}$
\end{lemma}
\begin{proof}
We can write 
$\ifm{u}{\RR \stoup \Gamma, v : A [B/X] \seqar w:C} =
\Fr[1] { \Fl[2] {A[B/X]} \to \Fr[3]{C}}$ 
and %
$\ifm{u}{\RR \stoup \Gamma, v : \forall X A\seqar w:C} =
\Fr[1] { \Fl[2] {\forall X A} \to \Fr[3]{C}}$
by~\cref{lem:fm-ctx-2}.

$\IKtSO \vdash \forall X A \to A[B/X]$,
hence by~\cref{lem:GenBoxNec,lem:GenKDIA} and $\mp$
$\IKtSO \vdash \Fl[2] {\forall X A} \to \Fl[2] {A[B/X]}$. 

As $\IPL \vdash (A \to B) \to (B \to C) \to (A \to C)$, we have $\IKtSO \vdash (\Fl[2] {\forall X A} \to \Fl[2] {A[B/X]}) \to (\Fl[2] {A[B/X]} \to \Fr[3]{C}) \to (\Fl[2] {\forall X A} \to \Fr[3]{C})$. 

By $\mp$ and~\cref{lem:GenBoxNec,lem:GenKBOX} we have

$\IKtSO \vdash \Fr[1] {\Fl[2] {A[B/X]} \to \Fl[3]{C}} \to \Fr[1] {\Fl[2] {\forall X A} \to \Fr[3]{C}}$. 
\end{proof}

\begin{lemma}[Soundness of $\rr \forall$]\label{lem:sound-far}
If $\IKtSO \vdash \ifm{u}{\RR\stoup \Gamma \seqar v : A[P/X]}$
then $\IKtSO \vdash \ifm{u}{\RR\stoup \Gamma \seqar v: \forall X A }$.
\end{lemma}
\begin{proof}
We can write
$\ifm{u}{\RR\stoup \Gamma \seqar v : A[P/X]} =
\Fr{A[P/X]}$ 
while %
$\ifm{u}{\RR\stoup \Gamma \seqar v: \forall X A } =
\Fr{\forall X A}$ and $P$ does not occur in $\Fr{\forall X A}$
by~\cref{lem:fm-ctx}.

By $\gen$, if $\IKtSO \vdash \Fr{A[P/X]}$ then $\IKtSO \vdash \forall X\Fr{ A}$

By~\cref{lem:GenGen} we have $\IKtSO \vdash \forall X \Fr {A} \to \Fr { \forall X A}$. 

Hence by $\mp$, if $\IKtSO \vdash \Fr{A[P/X]}$ then $\IKtSO \vdash \Fr{\forall X A}$
\end{proof}

\begin{lemma}[Soundness of $\lr \Box$ and $\lr\Boxblack$]\label{lem:sound-boxl}
If $\IKtSO \vdash \ifm{u}{\RR , v \R w \stoup \Gamma, w : A \seqar x:C}$
then $\IKtSO \vdash \ifm{u}{\RR , v \R w \stoup \Gamma, v : \Box A \seqar x:C }$.
\end{lemma}
\begin{proof}

This requires a careful analysis of the polytree structure of $\RR$.

\begin{itemize}
	
	\item If $u\conn{\RR}v$ does not contain $vRw$
	and $x\in\RAtMin wv$: 
	$\ifm{u}{\RR , v \R w \stoup \Gamma, w : A \seqar x:C} =
	\fimp[F_1]{\Box((A \AND D) \to \fimp[F_2]{C})}$

	$\ifm{u}{\RR , v \R w \stoup \Gamma, v : \Box A \seqar x:C } = 
	\fimp[F_1]{\Box A \to \Box(D \IMP \fimp[F_2]{C})}$
	
	for some formula D and contexts $F_1, F_2$ by~\cref{lem:fm-ctx,lem:fm-ctx-2}.

	Using propositional reasoning 
	$\pprove{((A \AND D) \IMP \fimp[F_2]{C}) \IMP A \IMP (D \IMP \fimp[F_2]{C})}$,
	apply $\necw$ and $\kax{1}$ to get
	$
	\pprove{\BOX((A \AND D) \IMP \fimp[F_2]{C}) \IMP \BOX A \IMP \BOX (D \IMP \fimp[F_2]{C})}   $
	
	We can apply~\cref{lem:GenBoxNec,lem:GenKBOX} to get
	$
	\pprove{\fimp[F_1]{\BOX((A \AND D) \IMP \fimp[F_2]{C})} \IMP \fimp[F_1]{\BOX A \IMP \BOX (D \IMP \fimp[F_2]{C})}}   
	$

	\item If $u\conn{\RR}v$ does not contain $vRw$
	and $x\in\RAtMin vw$: 

	$\ifm{u}{\RR , v \R w \stoup \Gamma, w : A \seqar x:C} =
	\Fr[1] {\Fl[2] {\Diamond (A \land D) } \to \Fr[3] {C}}$ 

	$\ifm{u}{\RR , v \R w \stoup \Gamma, v : \Box A \seqar x:C } =
	\Fr[1] {\Fl[2] { \Box A \land  \Diamond D } \to \Fr[3] {C}}$ 
	
	for some formula $D$ and contexts $F_1, F_2, F_3$ by~\cref{lem:fm-ctx,lem:fm-ctx-2}.
	
	From~\cref{prop:IK-thm} we have 
	$\IKtSO \vdash (\Box A \land  \Diamond D) \to \Diamond (A \land D)$ 
	
	By~\cref{lem:GenBoxNec,lem:GenKDIA} and $\mp$
	$\IKtSO \vdash \Fl[2]{ \Box A \land  \Diamond D } \to \Fl[2]{\Diamond (A \land D) }$. 
	
	We have the $\pprove[\IPL]{(A \to B) \to (B \to C) \to (A \to C)}$, thus we also have
	$\IKtSO \vdash (\Fl[2]{ \Box A \land  \Diamond D } \to \Fl[2]{\Diamond (A \land D) }) \to (\Fl[2] {\Diamond (A \land D) } \to \Fr[3]{C}) \to (\Fl[2] { \Box A \land  \Dmnd D } \to \Fr[3] {C})$. 
	
	By~\cref{lem:GenBoxNec,lem:GenKBOX} and $\mp$
	$\IKtSO \vdash \Fr[1] {\Fl[2] {\Diamond (A \land D) } \to \Fr[3] {C}} \to \Fr[1] {\Fl[2] { \Box A \land  \Diamond D } \to \Fr[3] {C}}$.

	\item If $u\conn{\RR}v$ contains $vRw$
	and $x\in\RAtMin wv$: 

	$\ifm{u}{\RR , v \R w \stoup \Gamma, w : A \seqar x:C} =
	\Fr[1] {\Fl[2] {A \AND \DIAB D } \to \Fr[3] {C}}$ 

	$\ifm{u}{\RR , v \R w \stoup \Gamma, v : \Box A \seqar x:C } =
	\Fr[1] {\Fl[2] { \DIAB (\BOX A \AND D)} \to \Fr[3] {C}}$ 
	for some formula $D$ and contexts $F_1, F_2, F_3$ by~\cref{lem:fm-ctx,lem:fm-ctx-2}.

	Propositional reasoning gives $\pprove{(\BOX A \AND D) \IMP D}$.
	Applying $\necb$ and $\kaxb{2}$, we get $\pprove{\DIAB(\BOX A \AND D) \IMP \DIAB D}$.
	
	Similarly propositional reasoning gives $\pprove{(\BOX A \AND D) \IMP \BOX A}$
	and applying $\necb$ and $\kaxb{2}$, we get $\pprove{\DIAB(\BOX A \AND D) \IMP \DIAB \BOX A}$.
	
	Using the $\bdwb$ axiom $\DIAB \BOX A \IMP A$ and $\IMP$-transitivity, we get $\pprove{\DIAB(\BOX A \AND D) \IMP A}$.
	
	Combining the two with propositional reasoning, we get $\pprove{\DIAB(\BOX A \AND D) \IMP (A \AND \DIAB D)}$.
	By~\cref{lem:GenBoxNec,lem:GenKDIA} and $\mp$ we get
	$\pprove{\fand[F_2]{\DIAB(\BOX A \AND D)} \IMP \fand[F_2]{A \AND \DIAB D}}$.
	
	Propositional reasoning gives $\pprove{(\Fl[2] {A \AND \DIAB D } \to \Fr[3] {C}) \IMP (\Fl[2] { \DIAB (\BOX A \AND D)} \to \Fr[3] {C})}$.
	By~\cref{lem:GenBoxNec,lem:GenKBOX} and $\mp$
	$\pprove{\fimp[F_1]{\Fl[2] {A \AND \DIAB D } \to \Fr[3] {C}} \IMP \fimp[F_1]{\Fl[2] { \DIAB (\BOX A \AND D)} \to \Fr[3] {C}}}$.
	
	\item If $u\conn{\RR}v$ contains $vRw$
	and $x\in\RAtMin vw$: 

	$\ifm{u}{\RR , v \R w \stoup \Gamma, w : A \seqar x:C} =
	\Fr[1] {A \IMP \BOXB (D \IMP \Fr[3] {C})}$ and 
	
	$\ifm{u}{\RR , v \R w \stoup \Gamma, v : \BOX A \seqar x:C} =
	\Fr[1] {\BOXB ((\BOX A \AND D) \IMP \Fr[3] {C})}$ 
	
	for some formula $D$ and contexts $F_1, F_3$ by~\cref{lem:fm-ctx,lem:fm-ctx-2}.
	
	First, using the $\bdwb$ axiom $\DIAB \BOX A \IMP A$ and propositional reasoning, $\pprove{(A \IMP \BOXB (D \IMP \Fr[3] {C})) \IMP (\DIAB \BOX A \IMP \BOXB (D \IMP \Fr[3] {C}))}$.
	Using Proposition~\ref{prop:IK-thm}, $\pprove{(\DIAB \BOX A \IMP \BOXB (D \IMP \Fr[3] {C})) \IMP \BOXB (\BOX A \IMP D \IMP \Fr[3] {C})}$.
	Propositional reasoning, $\necb$ and $\kaxb{1}$ gives $\pprove{\BOXB (\BOX A \IMP D \IMP \Fr[3] {C}) \IMP \BOXB ((\BOX A \AND D) \IMP \Fr[3] {C})}$.
	
	Applying $\IMP$-transitivity, we get $\pprove{(A \IMP \BOXB (D \IMP \Fr[3] {C})) \IMP \BOXB ((\BOX A \AND D) \IMP \Fr[3] {C})}$.
	By~\cref{lem:GenBoxNec,lem:GenKBOX} and $\mp$, we get
	$\pprove{\fimp[F_1]{A \IMP \BOXB (D \IMP \Fr[3] {C})} \IMP \fimp[F_1]{\BOXB ((\BOX A \AND D) \IMP \Fr[3] {C})}}$. \qedhere
\end{itemize}
\end{proof}

\begin{lemma}[Soundness of $\rr\Box$ and $\rr\Boxblack$]\label{lem:sound-boxr}
If $\IKtSO \vdash \ifm{u}{\RR , v \R w \stoup \Gamma \seqar w : A }$
then $\IKtSO \vdash \ifm{u}{\RR \stoup \Gamma \seqar v : \Box A}$.
\end{lemma}
\begin{proof}
Both the premiss and the conclusion will be written as the same formula interpretation $F^\to \{\Box A\}$. Note that there cannot be further formulas inside the $\Box$ as the label in the premiss was fresh. Thus, the rule can be derived trivially.
\end{proof}

\begin{remark}
Note that in almost all rules, we were able to derive a formula which is directly in correspondence to the rule which we were translating. This is for all rules, except for $\forall_r$ which is directly corresponding to the rule of generalisation. In contrast, the cut-rule can also be derived as a single rule, as the next lemma shows.
\end{remark}

\begin{lemma}\label{lem:sound-implad}
If $\IKtSO \vdash \ifm{u}{\RR \stoup \Gamma \seqar v : A}$ 
and $\IKtSO \vdash \ifm{u}{\RR \stoup \Gamma, v :B \seqar w:C}$, 
then $ \IKtSO \vdash \ifm{u}{\RR \stoup \Gamma, v:A \limp B \seqar w:C}$
\end{lemma}
\begin{proof}
We can write
$\ifm{u}{\RR \stoup \Gamma \seqar v : A} =
\Fr[1] { \Fl[3]{\TOP} \to \Fr[2] { A }}$ 
and 
$\ifm{u}{\RR \stoup \Gamma' , v :B \seqar w:C} = 
\Fr[1] { \Fl[2] {B} \to \Fr[3] { C }}$ 
and %
$\ifm{u}{\RR \stoup \Gamma, \Gamma', v:A \limp B \seqar w:C} = 
\Fr[1] { \Fl[2] {A \to B } \to \Fr[3] { C }}$
by~\cref{lem:fm-ctx-2}.

Note that 
$\IPL \vdash (A \to B \to C) \to (( D \to E) \to E) \to (D \to B) \to (C \to E) \to (A \to E)$
and therefore by substitution also 
$\IKtSO \vdash 
(\Fl[2] {A \to B } \to \Fr[2] {A} \to \Fl[2] {B})) 
\to 
(( \Fl[3]{\TOP} \to \Fr[3] { C }) \to \Fr[3] { C })
\to 
(\Fl[3]{\TOP} \to \Fr[2] {A}) 
\to 
(\Fl[2] {B} \to \Fr[3] { C }) 
\to 
(\Fl[2] {A \to B } \to \Fr[3] { C })$. 

By~\cref{lem:GenFSDiaBox,lem:GenInterDiaBox} and $\mp$ 
$\IKtSO \vdash 
(\Fl[3]{\TOP} \to \Fr[2] {A }) 
\to 
(\Fl[2] {B } \to \Fr[3] { C }) 
\to 
(\Fl[2] {A \to B } \to \Fr[3] { C })$.

By~\cref{lem:GenBoxNec,lem:GenKBOX} and $\mp$, we also get 
$\IKtSO \vdash 
\Fr[1] {\Fl[3]{ \TOP } \to \Fr[2] {A }} 
\to 
\Fr[1] {\Fl[2] {B } \to \Fr[3] { C }} 
\to 
\Fr[1] {\Fl[2] {A \to B } \to \Fr[3] { C }}$.
\end{proof}

\begin{corollary}[Soundness of $\lr \limp$]\label{cor:sound-impl}
If $\IKtSO \vdash \ifm{u}{\RR \stoup \Gamma \seqar v : A}$ 
and $\IKtSO \vdash \ifm{u}{\RR \stoup \Gamma' , v :B \seqar w:C}$, 
then $ \IKtSO \vdash \ifm{u}{\RR \stoup \Gamma, \Gamma', v:A \limp B \seqar w:C}$
\end{corollary}
\begin{proof}
We use the fact that we can always derive from the formula interpretation of $\RR \stoup \Gamma \Ra w : C$ the formula interpretation of $\RR \stoup \Gamma , \Gamma' \Ra w : C$ by applying~\cref{lem:sound-weak} multiple times for $\wk$.
\end{proof}

\begin{lemma}\label{lem:sound-cutadd}
If $\IKtSO \vdash \ifm{u}{\RR \stoup \Gamma \seqar v : A}$ 
and $\IKtSO \vdash \ifm{u}{\RR \stoup \Gamma, v : A \seqar w:C}$, 
then $ \IKtSO \vdash \ifm{u}{\RR \stoup \Gamma\seqar w:C}$
\end{lemma}
\begin{proof}
We can write
$\ifm{u}{\RR \stoup \Gamma \seqar v : A} =
\Fr[1] { \Fl[3]{\TOP} \to \Fr[2] { A }}$ 
and 
$\ifm{u}{\RR \stoup \Gamma, v : A \seqar w:C} =
\Fr[1] { \Fl[2] {A} \to \Fr[3] { C }}$ 
and 
$\ifm{u}{\RR \stoup \Gamma\seqar w:C} = 
\Fr[1] { \Fl[2] \to \Fr[3] { C }}$
by~\cref{lem:fm-ctx-2}.

We use the fact that 
$\IPL \vdash (A \to B \to C) \to (( D \to E) \to E) \to (D \to B) \to (C \to E) \to (A \to E)$
and substitute to get
$\IKtSO \vdash 
(\Fl[2]{\TOP} \to \Fr[2] { A } \to \Fl[2] {A }) 
\to 
((\Fl[3]{\TOP}  \to \Fr[3] { C }) \to \Fr[3] { C })
\to 
( \Fl[3]{\TOP}  \to \Fr[2] { A }) 
\to 
(\Fl[2] {A } \to \Fr[3] { C }) 
\to 
(\Fl[2] \to \Fr[3] { C })$.

By~\cref{lem:GenFSDiaBox} we get
$\Fl[2]{\TOP} \to \Fr[2]{A } \to \Fl[2] { A})$.

By~\cref{lem:GenInterDiaBox} we get 
$(\Fl[3] {\TOP} \to \Fr[3]{C}) \to \Fr[3] {C}$ 

and apply $\mp$ to get
$\IKtSO \vdash 
(\Fl[3]{\TOP}  \to \Fr[2] { A }) 
\to 
(\Fl[2] {A } \to \Fr[3] { C }) 
\to 
(\Fl[2] \to \Fr[3] { C })$

Using~\cref{lem:GenBoxNec,lem:GenKBOX} with $\mp$, we can derive
$\IKtSO \vdash 
\Fr[1]{\Fl[3]{\TOP}  \to \Fr[2] { A }} 
\to 
\Fr[1]{\Fl[2] {A } \to \Fr[3] { C }}
\to 
\Fr[1]{\Fl[2] \to \Fr[3] { C }}$
\end{proof}

\begin{corollary}[Soundness of $\cut$]\label{cor:sound-cut}
If $\IKtSO \vdash \ifm{u}{\RR \stoup \Gamma \seqar v : A}$ 
and $\IKtSO \vdash \ifm{u}{\RR \stoup \Gamma' , v : A \seqar w:C}$, 
then $ \IKtSO \vdash \ifm{u}{\RR \stoup \Gamma, \Gamma' \seqar w:C}$
\end{corollary}
\begin{proof}
We use the fact that we can always derive from the formula interpretation of $\RR \stoup \Gamma \Ra w : C$ the formula interpretation of $\RR \stoup \Gamma , \Gamma' \Ra w : C$ by applying~\cref{lem:sound-weak} multiple times for $\wk$.
\end{proof}

\begin{proof}[Proof of~\cref{lem:sequent-soundness-intuitionistic}.] 
We proceed by induction on the proof of $\sequent$. 

If $\sequent$ is obtained by an application of $\id$, 
apply~\cref{lem:sound-id} to obtain an axiomatic proof of $\ifm{u}{\sequent}$ with $u$ ranging over the labels occurring in $\sequent$.

Let the last rule in the proof of $\sequent$ be a single-premiss rule with premiss $\sequent_1$.
	By the inductive hypothesis, we obtain axiomatic proofs of $\ifm{u}{\sequent_1}$ with $u$ ranging over the labels occurring in $\sequent_1$. 
	Note that all labels occurring in $\sequent_1$ are also occurring in $\sequent$.
	Applying the lemmas above appropriately, we obtain $\IKtSO$ proofs of $\ifm{u}{\sequent}$ with $u$ ranging over labels occurring in $\sequent$.
	
	Let the last rule in the proof of $\sequent$ be a dual-premiss rule with premisses $\sequent_1$ and $\sequent_2$.
By the inductive hypothesis, %
we have proofs of $\ifm{u} { \sequent_1}$ and $\ifm{u} { \sequent_2 }$ where $u$ ranges over the labels occurring in $\sequent$.
By applying the appropriate lemma, we obtain axiomatic proofs proofs of $\ifm{u} {\sequent}$ ($u$ ranges over labels in $\sequent$). 
\end{proof}

\subsection{Axiomatic soundness for $\labKtSO$}

We recover~\cref{prop:axiomsoundness-classical}\eqref{axiomsoundness-classical} for the classical system as a corollary of the intuitionistic one, again strengthening the statement first:

\begin{lemma}[Soundness of interpretation]
\label{lem:sequent-soundness-classical}
If $\labKtSO \proves \RR \stoup \Gamma \seqar \Delta$ then $\KtSO\proves\cfm{u}{\RR \stoup \Gamma \seqar \Delta}$ 
for any label $u$ occurring in $\RR$, $\Gamma$ or $\Delta$.
\end{lemma}

This result is factored through an alternate system for $\KtSO$, namely $\labIKtSO$ with a rule for double negation elimination:
$$\vlinf{(\neg\neg)}{}{\RR \stoup \Gamma \seqar v:A}{\RR \stoup \Gamma \seqar v:\neg\neg A}$$

The soundness of that additional rule $(\neg\neg)$ is directly obtained from the definition of $\KtSO$, which includes the axiom $\lnot \lnot A \limp A$.

\begin{lemma}\label{lem:lIKtSO+negneg}
If $\labKtSO \proves \RR \stoup \Gamma \seqar \Delta$, %
then $\labIKtSO + (\neg\neg) \proves \RR \stoup \Gamma , \neg\Delta\footnote{If $\Delta = v_1:A_1, \ldots, v_n:A_n$ then $\neg\Delta = v_1:\neg A_1, \ldots, v_n:\neg A_n$.}
\seqar x:\bot$
for any label $x$ occurring in $\RR$, $\Gamma$ or $\Delta$. %
\end{lemma}
\begin{proof}
By~\cref{lem:simbotrule}, $\labIKtSO\proves \seq{\rels R}{\Gamma, \fm{v}{\BOT}}{\fm{w}{A}}$ whenever $v$ and $w $ are connected in $\rels R$.

\begin{itemize}
\item $\vlinf{\id}{}{\rels R \stoup  v : A \seqar v : A}{}$
$\vlderivation{
	\vliin{\lr\limp}{}{\rels R \stoup  v : A, v : \neg A \seqar x:\BOT}{
		\vlin{\id}{}{\rels R \stoup  v : A \seqar v:A}{
			\vlhy{}
		}
	}{
		\vlin{}{}{\rels R \stoup  v : \BOT \seqar x:\BOT}{
			\vlhy{\text{\footnotesize{\cref{lem:simbotrule}}}}
		}
	}
}$

\item $\vliinf{\lr \limp}{}{\rels R \stoup \Gamma, \Gamma', v:A \limp B \seqar \Delta, \Delta'}{\rels R \stoup \Gamma \seqar \Delta, v : A}{\rels R \stoup \Gamma' , v :B \seqar \Delta'}$ (and similarly for $\cut$)

By inductive hypothesis, 
$\labIKtSO + (\neg\neg) \proves \RR \stoup \Gamma , \neg\Delta, v:\neg A \seqar v:\bot$
and
$\labIKtSO + (\neg\neg) \proves \RR \stoup \Gamma' , v:B, \neg\Delta' \seqar x:\bot$.
$$\vlderivation{
	\vliin{\lr\limp}{}{\RR\stoup\Gamma,\Gamma',v:A\IMP B,\neg\Delta,\neg\Delta'\seqar x:\BOT}{
		\vlin{(\neg\neg)}{}{\RR\stoup\Gamma,\neg\Delta \seqar v:A}{
			\vlin{\rr\limp}{}{\RR\stoup\Gamma,\neg\Delta \seqar v:\neg\neg A}{
				\vlhy{\RR\stoup\Gamma,\neg\Delta, v:\neg A \seqar v:\BOT}
			}
		}
	}{
		\vlhy{\RR\stoup\Gamma',v:B,\neg\Delta'\seqar x :\BOT}
	}
}
$$

\item $\vlinf{\rr \limp}{}{\rels R \stoup \Gamma \seqar \Delta, v : A \limp B}{\rels R \stoup \Gamma, v:A \seqar \Delta, v:B}$

By inductive hypothesis, $\labIKtSO + (\neg\neg) \proves \RR \stoup \Gamma, v:A , \neg\Delta, v:\neg B \seqar v:\bot$.
$$\vlderivation{
	\vliin{\lr\limp}{}{\RR\stoup\Gamma,\neg\Delta,v:\neg(A\limp B)\seqar x:\BOT}{
		\vlin{\rr\limp}{}{\RR\stoup\Gamma,\neg\Delta\seqar v:A\limp B}{
			\vlin{(\neg\neg)}{}{\RR\stoup\Gamma,v:A,\neg\Delta\seqar v:B}{
				\vlin{\rr\limp}{}{\RR\stoup\Gamma,v:A,\neg\Delta\seqar v:\neg\neg B}{
					\vlhy{\RR\stoup\Gamma,v:A,\neg\Delta, v:\neg B\seqar v:\BOT}
				}
			}
		}
	}{
		\vlin{}{}{\RR\stoup\Gamma,\neg\Delta,v:\bot\seqar x:\bot}{
			\vlhy{\text{\footnotesize{\cref{lem:simbotrule}}}}
		}
	}
}
$$

\item Non-branching left rules: $\vlinf{}{}{\rels R \stoup \Gamma, v : A \seqar \Delta }{\rels R , \stoup \Gamma, v' : A' \seqar \Delta}$

By inductive hypothesis, $\labIKtSO + (\neg\neg) \proves \RR \stoup \Gamma, v':A' , \neg\Delta \seqar x:\bot$.

It is enough to apply the same rule in $\labIKtSO$:
$$\vlinf{}{}{\RR\stoup \Gamma, v:A , \neg\Delta \seqar x:\bot}{\RR \stoup \Gamma, v':A' , \neg\Delta \seqar x:\bot}$$

\item $\vlinf{\rr \Box}{\text{\footnotesize $w$ fresh}}{\rels R \stoup \Gamma \seqar \Delta , v : \Box A}{\rels R , v \R w \stoup \Gamma \seqar \Delta, w : A }$ (and similarly for $\lr\Boxblack$ and $\lr\forall$)

By inductive hypothesis, $\labIKtSO + (\neg\neg) \proves \rels R , v \R w \stoup \Gamma, \neg\Delta, w : \neg A \seqar w:\bot$.
$$\vlderivation{
	\vliin{\lr\limp}{}{\rels R \stoup \Gamma,\neg\Delta , v : \neg\Box A \seqar x:\BOT}{
		\vlin{\rr\Box}{}{\rels R \stoup \Gamma,\neg\Delta  \seqar v : \Box A}{
			\vlin{(\neg\neg)}{}{\rels R, vRw \stoup \Gamma,\neg\Delta  \seqar w : A}{
				\vlin{\rr\limp}{}{\rels R,vRw \stoup \Gamma,\neg\Delta  \seqar w : \neg\neg A}{
					\vlhy{\rels R,vRw \stoup \Gamma,\neg\Delta , w : \neg A \seqar w:\bot}
				}
			}
		}
	}{
		\vlin{}{}{\rels R \stoup \Gamma,\neg\Delta , v : \bot \seqar x:\BOT}{
			\vlhy{\text{\footnotesize{\cref{lem:simbotrule}}}}
		}
	}
}$$
\end{itemize}
\end{proof}

\begin{proof}
[Proof of \cref{lem:sequent-soundness-classical}] %

By~\cref{lem:lIKtSO+negneg}, since $\labKtSO \proves \RR \stoup \Gamma \seqar \Delta$, also
$\labIKtSO + (\neg\neg) \proves \RR \stoup \Gamma , \neg\Delta \seqar x:\bot$.

We already established the local soundness of $\labIKtSO$ rules in the previous section.
It remains to show the local soundness of $\vlinf{(\neg\neg)}{}{\RR \stoup \Gamma \seqar w:A}{\RR \stoup \Gamma \seqar w:\neg\neg A}$

\noindent
Assume $\KtSO \proves \cfm{u}{\RR \stoup \Gamma \seqar w:\neg\neg A} = \lfm{u}{\RR\stoup\Gamma} \limp \rfm{u}{\RR\stoup w:\neg\neg A}$,
we need to show that $\KtSO \proves \cfm{u}{\RR \stoup \Gamma \seqar w: A} = \lfm{u}{\RR\stoup\Gamma} \limp \rfm{u}{\RR\stoup w: A}$

\noindent
By induction on the path $u\conn{\RR}v$, let us show that $\KtSO\proves\rfm{u}{\RR\stoup w:\neg\neg A}\IMP \rfm{u}{\RR\stoup w: A} $
\begin{itemize}
\item if $u=v$: immediate as $\KtSO \proves \lnot \lnot A \limp A$.

\item if there is $v$ such that $uRv\in\RR$ and $v\conn{\RR}w$:
By inductive hypothesis, $\KtSO\proves\rfm{v}{\RR\stoup w:\neg\neg A}\IMP \rfm{v}{\RR\stoup w: A} $.
Hence, by $\necw$, $\functwb$ and $\mp$, $\KtSO\proves\BOX\rfm{v}{\RR\stoup w:\neg\neg A}\IMP \BOX\rfm{v}{\RR\stoup w: A} $

\item if there is $v$ such that $vRu\in\RR$ and $v\conn{\RR}w$:
Similarly replacing $\BOX$ by $\BOXB$.
\end{itemize}

From this, we conclude that 
\begin{flalign*}
\KtSO \proves & \cfm{u}{\RR \stoup \Gamma , \neg\Delta \seqar u:\bot} 
\\
=& \lfm{u}{\RR\stoup\Gamma,\neg\Delta} \limp \rfm{u}{\RR\stoup u:\bot}
\\
=& (\lfm{u}{\RR\stoup\Gamma} \AND \lfm{u}{\RR\stoup\neg\Delta}) \limp \bot
\\
\IFF& \lfm{u}{\RR\stoup\Gamma} \IMP \lfm{u}{\RR\stoup\neg\Delta} \limp \bot
\\
=& \lfm{u}{\RR\stoup\Gamma} \IMP \rfm{u}{\RR\stoup\Delta}
\end{flalign*}

As from the definition of $\lfm{}$ and $\rfm{}$,
$\neg\lfm{u}{\RR\stoup\neg\Delta}
=
\rfm{u}{\RR\stoup\Delta}
$.
\qedhere
\end{proof}

\section{Perspectives}
\label{sec:further}

We have now completed the argument justifying our \cref{mainthm:soundness-completeness-classical,mainthm:soundness-completeness-intuitionistic,mainthm:hauptsatz}: they are obtained by the results we have presented according to the diagrams in \cref{fig:class-tour}, in the classical setting, and \cref{fig:int-tour}, in the intuitionistic setting.
Let us take a moment to reflect on (i) the relationship between the classical and intuitionistic theories we presented; and (ii) some interesting subsystems of second-order (intuitionistic) tense logic.

\subsection{Relating classical and intuitionistic: negative translations}
We have presented both classical and intuitionistic versions of second-order tense logic, so it would be natural to probe their relationship according to known techniques.
In particular, classical logic is interpreted by intuitionistic logic by the \textbf{negative} (or \textbf{double negation}) translations.
Since our language is formulated in the negative fragment, the \emph{G\"odel-Gentzen} translation is particularly easy to define, commuting with all but atomic formulas.\footnote{Note that we could have adapted other negative translations, such as Kolmogorov or Kuroda, but such a development is beyond the scope of this work. As in predicate logic, we suspect that all these translations would be equivalent over $\IKtSO$.}
Let us develop this here.
Recall that we write $\bot := \forall X X$ and $\lnot A := A \limp \bot$.

\begin{definition}
    [(Second-order modal) negative translation]
    For each formula $A$ define its \textbf{negative translation} $\gg A$ by:
    \[
    \begin{array}{r@{\ := \ }l}
         \gg P & \lnot \lnot P \\
         \gg X & \lnot \lnot X \\
         \gg {(A\limp B)} & \gg A \limp \gg B
    \end{array}
    \qquad
    \begin{array}{r@{\ := \ }l}
         \gg {(\Box A)} & \Box \gg A \\
         \gg {\blacksquare A)} & \blacksquare \gg A \\
         \gg {(\forall X A)} & \forall X \gg A
    \end{array}
    \]
\end{definition}

The main point of this subsection is to show the soundness of the $\gg \cdot$ translation, i.e.\ that it indeed embeds classical second-order tense logic into intuitionistic:
\begin{theorem}
\label{neg-trans-soundness}
    $\KtSO \proves A \implies \IKtSO \proves \n A$.
\end{theorem}

Before we prove this we first need some (expected) auxiliary results:
\begin{lemma} 
    [Negativity]
    \label{lem:negativity}
    $\IKtSO$ proves the following:
    \[
    \begin{array}{r@{\ \limp \  }l}
    \lnot \lnot \lnot A & \lnot A \\
         \lnot \lnot (A\limp B) & \lnot \lnot A \limp \lnot \lnot B \\
         \lnot \lnot \forall X A & \forall X \lnot \lnot A 
    \end{array}
    \qquad
     \begin{array}{r@{\ \limp \  }l}
         \lnot \lnot \Box A & \Box \lnot \lnot A \\
         \lnot \lnot \blacksquare A & \blacksquare \lnot \lnot A 
    \end{array}
    \]
\end{lemma}
\begin{proof}
    The left three items are well known (see, e.g., \cite[Sections~2.3 \& 11.5.9]{TroelstraSchwichtenberg2000}.
    Proofs of $\Box$, i.e.\ $\lnot \lnot \Box A \limp \Box \lnot \lnot A$ in $\IK$ were given in \cite{dasmarin:blog22} and \cite[Lemma~10]{das2023intuitionistic}, whence we obtain the same here since $\IKtSO$ contains $\IK$, under the impredicative encodings of positive connectives, cf.~\cref{sec:underlying-modal-tense-logics}.
    For self-containment let us repeat that proof here:
    \[
    \begin{array}{ll}
         A \limp \lnot \lnot A & \text{by $\IPL$ reasoning} \\
         \Box A \limp \Box \lnot \lnot A & \text{by $\necw$ and $\functwb$} \\
         \Box A \limp \Diamond \lnot A \limp \Diamond \bot & \text{by definition of $\lnot$ and $\functwd$} \\
         \Box A \limp \Diamond \lnot A \limp \bot & \text{by $\diadistbot$} \\
         \lnot \lnot\Box A \limp \Diamond \lnot A \limp \bot & \text{by $\IPL$ reasoning} \\
         \lnot \lnot \Box A \limp \Diamond \lnot A \limp \Box \bot & \text{by definition of $\bot$ and $\CA$} \\
         \lnot \lnot \Box A \limp \Box (\lnot A \limp \bot) & \text{by $\diaimpbox$} \\
         \lnot \lnot \Box A \limp \Box \lnot \lnot A & \text{by definition of $\lnot$}
    \end{array}
    \]
    References to $\diadistbot$ and $\diaimpbox$ are from \cref{eq:other-IK-axioms}, and were derived previously in \cref{sec:underlying-modal-tense-logics}.
    The black version, $\lnot \lnot \blacksquare A \limp \blacksquare \lnot \lnot A$ now just follows by symmetry.
\end{proof}

\begin{proposition} \label{cor:Neg-Transl-DNE}
    $\IKtSO \vdash \neg \neg A^N \liff A^N$ and $\IKtSO \proves \bot \liff \n \bot$.
\end{proposition}
\begin{proof}
[Proof sketch]
$\n A \limp \lnot \lnot \n A$ is already a consequence of $\IPL$.
   For the converse direction, $ \lnot \lnot \n A \limp \n A$, we proceed by induction on the structure of $A$, using the previous Negativity \cref{lem:negativity} at each step.

$\bot \limp \n \bot$ is an instance of comprehension $\CA$, as $\bot = \forall X X$. 
For the converse direction, $\n \bot \limp \bot$, note that $\n \bot  = \forall X \lnot \lnot X$.
By comprehension axiom $\CA$, we thus have $\n \bot \limp \lnot \lnot \bot$, whence indeed $\n \bot \limp \bot$ by $\IPL$ reasoning.
\end{proof}

Now the soundness of the negative translation is readily established:

\begin{proof}
    [Proof of \cref{neg-trans-soundness}]
    By induction on on a $\KtSO$ proof of $A$:
    \begin{itemize}
        \item The $\n \cdot$-translation of every axiom of $\IKtSO$ is again an axiom instance of $\IKtSO$.
        For the remaining axiom of $\KtSO$, namely $\lnot \lnot A \limp A$, 
        note that $\n{(\lnot \lnot A \limp A)} = ((\n A \limp \n \bot) \limp \n \bot ) \limp \n A$, which is provable in $\IKtSO$ by \cref{cor:Neg-Transl-DNE}.
        \item The $\n\cdot$-translation of both inference rules of $\KtSO$ are again instances of inference rules of $\IKtSO$. \qedhere
    \end{itemize}
\end{proof}

\subsection{Specialising to sublogics}
\label{sec:sublogics}
Second-order (intuitionstic) tense logic has several sublogics of interest.
As one would expect, the results of this paper allow us to \emph{inherit} some analogous results for certain sublogics.
In particular by specialising the grand tours of \cref{fig:class-tour,fig:int-tour} to the modality-free fragment of our syntax we inherit a proof theoretic account, namely cut-admissibility, of:
\begin{itemize}
    \item \emph{Classical second-order propositional logic.} 
    $\labKtSO$ now specialises to the usual sequent calculus for second-order propositional logic (see, e.g., \cite[Section~3.A.1]{Girard1987:pt-log-comp}, \cite[Section~5.1]{sep-proof-theory} or \cite[Definition~15.3]{Takeuti1987:pt-book}). 
    Let us point that this is somewhat a toy result, as it is known that even Boolean comprehension, setting $C = \bot$ or $C=\top $ in $\CA$, suffices in proof search, due to the Boolean valued semantics.
    \item \emph{Intuitionistic second-order propositional logic.} 
    $\labIKtSO$ now specialises to the usual second-order sequent calculus for intuitionistic propositional logic (i.e.\ the classical calculus with singleton RHS constraint).
    Cut-admissibility for the \emph{multi} succedent variant (i.e.\ the modality-free fragment of $\mlIKtSO$) was obtained by Prawitz in \cite{prawitz1970some}, as well as its completeness over Beth models using similar techniques.
    Thanks to our \emph{negative} formula syntax we gain completeness over Kripke models (i.e.\ predicate models where $W = \emptyset$) and cut-admissibility for the usual single-succedent calculus, cf.~\cref{prop:multitosingle}.
\end{itemize}

It is natural to wonder whether the results of this work similarly give rise to a treatment of second-order (intiuitionistic) \emph{modal} logic, without the black modalities.
Logics based on this syntax have been much more significantly explored in the literature~\cite{fine1970propositional,bull1969modal,kaplan1970s5,ten2006expressivity,kaminski1996expressive,belardinelli2018second-order,blackburn2023axiom}.
However, the fact that our $\Diamond$ is defined not only in terms of $\Box$ and $\forall$ but also $\blacksquare$ complicates the situation. 
One could envisage restricting our labelled system to black-free formulas, but
our axiomatic translation in \cref{sec:labelsoundness} introduces $\Diamond $s, and hence $\blacksquare$s. 
In fact adding a native $\Diamond$ to evade this issue, along with whatever modal reasoning is used in \cref{sec:labelsoundness}, still does not necessarily yield cut-admissibility, for a somewhat subtle reason: 
cut-free labelled proofs of $\blacksquare$-free formulas may still require formulas with $\blacksquare$, due to the comprehension steps involved (i.e.\ $\lr \forall$).
This exemplifies the \emph{non-analyticity} of second-order logic.
For example, here is a cut-free $\labIKtSO$ proof of an instance of the negativity of $\Box$ and $\bot$, cf.~\cref{lem:negativity}:
\[
\vlderivation{
\vlin{\rr \Box}{}{\cdot \stoup v:\lnot \lnot \Box \bot \seqar v:\Box \bot}{
\vliin{\lr \limp}{}{vRw \stoup v:\lnot \lnot \Box \bot \seqar w:\bot}{
    \vlin{\rr\limp}{}{vRw \stoup \cdot \seqar v:\lnot \Box \bot}{
    \vlin{\rr\forall}{}{vRw \stoup v: \Box \bot \seqar v:\forall X X }{
    \vlin{\lr \Box}{}{vRw \stoup v:\Box \bot \seqar v: P}{
    \vlin{\lr \forall}{}{vRw \stoup w:\forall X X \seqar v:P}{
    \vlin{\lr\blacksquare}{}{vRw \stoup w:\blacksquare P \seqar v:P}{
    \vlin{\id}{}{vRw \stoup v:P \seqar v:P}{\vlhy{}}
    }
    }
    }
    }
    }
}{
    \vlin{\lr \forall}{}{vRw \stoup v:\forall X X \seqar w:\bot}{
    \vlin{\lr\Box}{}{vRw \stoup v:\Box \bot \seqar w:\bot}{
    \vlin{\id}{}{vRw \stoup w:\bot \seqar w:\bot}{\vlhy{}}
    }
    }
}
}
}
\]
This is a cut-free proof of a $\blacksquare$-free theorem (in particular without $\Diamond$s, native or otherwise), that nonetheless uses $\blacksquare$.
How should one prove this theorem intuitionistically without using $\blacksquare$?

Notice that the same issue does not present classically, as the theorem is an instance of the double negation-elimination axiom.
We suspect that the same proof search argument as in \cref{sec:completeness-classical} should go through in the $\blacksquare$-free fragment, classically.
Inspecting again the axiomatic translation of \cref{sec:labelsoundness}  with a classical sensitivity, notice that the formula translation of a $\blacksquare$-free sequent remains $\blacksquare$-free, as long as we interpret $\Diamond A := \lnot \Box \lnot A$.
Axiomatically, apart from second-order classical propositional logic $\mathsf{CPL}2 := \IPL 2 + \lnot \lnot A \limp A$, the white modal axioms from \cref{modal-axioms-rules-wb}, we also required $\mathsf B : \forall X \Box A \limp \Box \forall X A$ (cf.~\cref{ex:forall-distributes-over-box}).
Let us point out that the resulting axiomatisation, $\mathsf{CPL} 2 + \functwb + \functwd + \necw + \mathsf B$, almost matches the proposal of second-order (classical) modal logic in \cite{belardinelli2018second-order}, but for the fact that they admit only quantifier-free comprehension.
It would be interesting to develop more formally our arguments for the modal-only setting (without black modalities), and compare the resulting logic(s) with that of \cite{belardinelli2018second-order}.

\section{Conclusions}\label{sec:conclusions}
In this work we developed the axiomatics, semantics and proof theory of a second-order extension of tense logic, over both classical and intuitionistic bases.
We ultimately showed that several natural definitions of the intuitionistic or classical theory respectively coincide, showcasing the robustness of each logic.
Along the way we established fundamental metalogical results, namely soundness and completeness of axiomatisations with respect to certain (bi)relational semantics, and proof theoretic results, namely cut-admissibility for associated calculi based on labelled sequents.
We employed a \emph{proof search} based approach to both of these results, establishing both simultaneously by way of our `grand tours' in \cref{fig:class-tour,fig:int-tour}.
We conclude this work by discussing some further interesting directions of research.

Second-order logic has the capacity to define least and greatest \emph{fixed points} of positive formulas, by encoding Knaster-Tarski style definitions (see, e.g., \cite[Chapter~1]{iter-ind-dfns} or, at a higher level, \cite[Section~5.3]{sep-proof-theory}).
In second-order logic with modalities we can do the same if we include a \emph{global} modality, say $\Boxtimes$, where $\Boxtimes A$ should be read as ``everywhere $A$''.
In particular, the least fixed point of the operator $X\mapsto A(X)$, where $X$ appears only positively in $A(X)$, is given by:
\[
\mu X A(X) 
\ := \ 
\forall X (\Boxtimes (A(X) \limp X) \limp X)
\]
This encoding already appears in \cite[Section~5]{Stirling1996Games}, where it is observed that a second-order modal syntax can express all of the \emph{modal $\mu$-calculus} \cite{kozen82:mu-calc}.
It would be interesting to investigate the extension of our logics by a global modality according to the axiomatic, semantic and proof theoretic disciplines herein.
In the presence of tense modalities such an extension would subsume the \emph{two-way} modal $\mu$-calculus \cite{Vardi98:two-way}. 
Note that this logic has recently received a proof theoretic treatment via \emph{cyclic proofs} \cite{AfshEnqvLeiMartVen25:two-way}.

Finally it would be natural to recast second-order intuitionistic tense logic from a \emph{proofs-as-programs} viewpoint, \`a la Curry-Howard (see, e.g., \cite{sorensen2006lectures}). 
The logic $\CK$, a sublogic of $\IK$, and some extensions have already been studied from a proofs-as-programs perspective and translated into \emph{modal lambda calculi} \cite{bellin2001extended,davies2001modal} (see~\cite{kavvos2016many} for a survey).
\emph{Modal type theory} has independently emerged as a way to encapsulate computations with effects~\cite{moggi1989computational}, receiving categorical foundations in \cite{gratzer2022modalities,shulman2023semantics}, and has been implemented in mainstream programming languages~\cite{tang2025modal,lorenzen2024oxidizing}.
Recently adjoint modalities $\Diamondblack$ and $\Box$ have also been crucial in a proposal of Kavvos, linking relational semantics to the type theoretic approach for modal logic~\cite{kavvos2024two-dimensionalI}.
To this end it would be interesting to develop a \emph{natural deduction} formulation of $\IKtSO$, with corresponding term annotation, and prove its \emph{strong normalisation}.
Naturally the method of \emph{reducibility candidates}, due to Girard \cite{Girard1972:thesis}, should be applicable.
At the same time such an endeavour would further bolster the proof theoretic underpinnings of second-order (intuitionistic) tense logic.

\bibliographystyle{alpha}
\bibliography{references}

\end{document}